\documentclass{article}
\usepackage[a4paper]{geometry}
\usepackage[utf8]{inputenc}
\usepackage{amsmath,amsfonts,amssymb,latexsym,epsfig}
\usepackage{amsthm}
\usepackage{subcaption}
\usepackage{comment}
\usepackage{xcolor}
\usepackage{caption}
\usepackage{makecell}
\usepackage{url}
\usepackage{multirow}
\usepackage{fullpage}

\usepackage{algorithm}
\usepackage{algpseudocode}

\newtheorem{theorem}{Theorem}[section]
\newtheorem{lemma}[theorem]{Lemma}

\newtheorem{proposition}[theorem]{Proposition}

\newcommand{\review}[1]{\textcolor{black}{#1}}

\newcommand{\lk}{f_y}
\newcommand{\pr}{g}
\newcommand{\pot}{U}

\newcommand{\prox}{\mathrm{prox}}
\newcommand{\R}{\mathbb{R}}

\newcommand{\norm}[1]{\left\lVert#1\right\rVert}

\title{Accelerated Bayesian imaging by relaxed proximal-point Langevin sampling}
\author{Teresa Klatzer $^{1,3}$ \and Paul Dobson$^{2,3}$\and  Yoann Altmann$^{4}$ \and Marcelo Pereyra$^{2,3}$ \and Jes\'{u}s Mar\'{i}a Sanz-Serna$^{5}$ \and Konstantinos C. Zygalakis$^{1,3}$}
\date{\today}

\begin{document}

\maketitle

\begin{abstract}
This paper presents a new accelerated proximal Markov chain Monte Carlo methodology to perform Bayesian inference in imaging inverse problems with an underlying convex geometry. The proposed strategy takes the form of a stochastic relaxed proximal-point iteration that admits two complementary interpretations. For models that are smooth or regularised by Moreau-Yosida smoothing, the algorithm is equivalent to an implicit midpoint discretisation of an overdamped Langevin diffusion targeting the posterior distribution of interest. This discretisation is asymptotically unbiased for Gaussian targets and shown to converge in an accelerated manner for any target that is $\kappa$-strongly log-concave (i.e., requiring in the order of $\sqrt{\kappa}$ iterations to converge, similarly to accelerated optimisation schemes), comparing favorably to [M. Pereyra, L. Vargas Mieles, K.C. Zygalakis, SIAM J. Imaging Sciences, 13,2 (2020), pp. 905-935] which is only provably accelerated for Gaussian targets and has bias. For models that are not smooth, the algorithm is equivalent to a Leimkuhler–Matthews discretisation of a Langevin diffusion targeting a Moreau-Yosida approximation of the posterior distribution of interest, and hence achieves a significantly lower bias than conventional unadjusted Langevin strategies based on the Euler-Maruyama discretisation. For targets that are $\kappa$-strongly log-concave, the provided non-asymptotic convergence analysis also identifies the optimal time step which maximizes the convergence speed. The proposed methodology is demonstrated through a range of experiments related to image deconvolution with Gaussian and Poisson noise, with assumption-driven and data-driven convex priors. \review{Source codes for the numerical experiments of this paper are available from \url{https://github.com/MI2G/accelerated-langevin-imla}.}
\end{abstract}

\footnotetext[1]{School of Mathematics, University of Edinburgh, Edinburgh, Scotland, UK}
\footnotetext[2]{School of Mathematical and Computer Sciences, Heriot-Watt University, Edinburgh, Scotland, UK}
\footnotetext[3]{Maxwell Institute for Mathematical Sciences, Bayes Centre, 47 Potterrow, Edinburgh, Scotland, UK}
\footnotetext[4]{School of Engineering and Physical Sciences, Heriot-Watt University, Edinburgh Scotland, UK}
\footnotetext[5]{Departamento de Matemáticas, Universidad Carlos III de Madrid, Leganés (Madrid), Spain}

\section{Introduction}
The problem of estimating an unknown image from noisy measurements is central to imaging sciences (see, e.g., \cite{Kaipio:stat-inv-problems05,chambolle_pock_2016,arridge_maass_oktem_schonlieb_2019}, for examples related to image denoising, inpainting, tomographic reconstruction, medical imaging, resolution enhancement and computer vision). Estimation by direct inversion of the forward model relating the unknown image to the data is not
usually possible, since the inverse problem is often severely ill-conditioned or ill-posed. The literature describes a range of
mathematical frameworks to incorporate regularisation and formulate well-posed solutions (see, e.g., \cite{Kaipio:stat-inv-problems05,chambolle_pock_2016,arridge_maass_oktem_schonlieb_2019}). 

We consider imaging methodology rooted in the Bayesian statistical framework \cite[Chapter 3]{Kaipio:stat-inv-problems05}. The Bayesian statistical framework is an intrinsically probabilistic paradigm in which statistical models represent the data observation process and the prior information available, and solutions are delivered using inference techniques from Bayesian decision theory \cite{robert2007bayesian}. In particular, the Bayesian decision-theoretic approach provides a range of powerful strategies to quantify and summarise the uncertainty in the solutions delivered, which is important for applications that use the restored images as evidence for decision-making or in subsequent scientific analysis (see, e.g., \cite{durmus2017, pereyra-map-2017, adler2018deep, repetti2019scalable, zhou2020bayesian, Yao_2022, MDAMZ}). 

Applying Bayesian strategies to imaging inverse problems involves some significant modelling and computational challenges. This paper focuses on the computational aspects of performing Bayesian inference in imaging problems with an underlying convex geometry (i.e., with a posterior distribution that is log-concave), for which we can provide detailed theoretical guarantees. We consider both assumption-driven \cite{chambolle_pock_2016} and data-driven strategies \cite{goujon2022crrnn,kobler-vn-2017,mukherjee2021icnn}, with particular attention to situations with no or very limited ground truth data available, where adopting a large data-driven model would not be possible. 

Modern approaches for Bayesian computation in imaging inverse problems can be broadly classified into the following three groups. The first group encompasses approximation methods, such as message-passing and variational Bayes algorithms (see, e.g., \cite[Section 3]{pereyra:hal-01312917, Yao_2022}) which include the increasingly prevalent denoising diffusion approaches and other data-driven strategies based on large machine learning models (see, e.g., \cite{kawar2022denoising,adler2018deep}). Strategies in this first group are computationally efficient and can support a variety of inferences but are relatively model specific, offer limited guarantees, and can exhibit convergence issues. The second group encompasses deterministic and stochastic optimisation strategies \cite[Section 4]{pereyra:hal-01312917}, including proximal convex optimisation approaches \cite{chambolle_pock_2016}. These are used predominantly to compute maximum-a-posteriori solutions but can also be useful for some basic forms of uncertainty quantification (see, e.g., \cite{pereyra-map-2017, repetti-pey-2019}). Within their limited scope of action, optimisation approaches scale efficiently to large problems and offer strong guarantees on the solutions delivered. The third group encompasses Monte Carlo integration methods \cite{pereyra:hal-01312917}, including modern Markov chain Monte Carlo (MCMC) strategies based on approximations of the Langevin diffusion \cite{pereyra:hal-01312917, durmus2017,8625467,PVZ20,MDAMZ,Laumont2021}. Monte Carlo strategies are very versatile and can support complex Bayesian inferences, at the expense of a significantly higher computational cost. Some modern MCMC strategies also provide strong guarantees on the delivered solutions \cite{durmus2017,Laumont2021}.

In this paper, we focus on MCMC methodology for Bayesian inference in imaging problems with an underlying convex geometry. Efficient MCMC simulation in this context is challenging for three main reasons: the high dimensionality of the solution space, lack of smoothness, and poor conditioning (i.e., the posterior density is highly anisotropic). Modern MCMC methods typically address the high dimensionality involved by mimicking a Langevin diffusion process that exploits first-order (gradient) information to scale to large problems very efficiently \cite[Section 2]{pereyra:hal-01312917}. Unfortunately, Bayesian imaging models often violate the regularity conditions that Langevin strategies require in order to be effective. This difficulty can be addressed by exploiting ideas and techniques from non-smooth convex optimisation as proposed in \cite{PereyraMarcelo2016PMcM}, resulting in proximal MCMC methods such as the Moreau Yoshida Unadjusted Langevin Algorithm (MYULA) \cite{durmus2017} and the Split-and-Augmented Gibbs sampler \cite{8625467,pereyra2023split} (see Section \ref{subsec:nonsmooth} for more details). It is also possible to derive proximal MCMC methods from a reflected Langevin diffusion that naturally incorporates constraints in the solution space \cite{MDAMZ}. 

MCMC methods derived from the Langevin diffusion process can become pathologically slow when the posterior distribution of interest is highly anisotropic, as often encountered in problems that are severely ill-conditioned or ill-posed. Indeed, similarly to gradient descent methods, standard Langevin sampling methods have a computational cost in the order of $\kappa$ for targets that are $\kappa$-strongly log-concave. This phenomenon can be demonstrated empirically on simple Gaussian models (see \cite[Figure 2]{PVZ20}) as well as established theoretically by conducting a non-asymptotic convergence analysis (see, e.g., \cite{durmus2017}). This analysis also reveals that the difficulties in dealing with $\kappa \gg 1$ stem from the choice of the discrete-time approximation of the Langevin diffusion used, based on a Euler-Maruyama scheme, and not from the continuous-time process itself.

Overcoming this type of behaviour with respect to the condition number $\kappa$ is very important and there are currently two main strategies to improve the convergence speed when $\kappa \gg 1$. One strategy is to exploit the structure of the problem in order to use some form of preconditioning \cite{CorbineauKCTP19,vono2022-admm,pereyra2023split}. This can be highly effective in some cases, but it is difficult to apply generally. A promising alternative is to leverage ideas from the numerical analysis of stochastic differential equations in order to develop MCMC methods from better discrete-time approximation of the Langevin diffusion that are inherently more robust to $\kappa \gg 1$. A prime example of this approach is \review{to use a stochastic second kind orthogonal Runge-Kutta-Chebyshev method} (SKROCK) \cite{PVZ20}, which exhibits \emph{acceleration}\footnote{We use the term acceleration in a manner akin to the convex optimisation literature; other works, e.g., \cite{Jordan21}, use this term to describe an improvement w.r.t. the problem dimension.} in that it only requires in the order of $\sqrt{\kappa}$ gradient evaluations to converge and hence dramatically outperforms MYULA when $\kappa$ is large. However, theoretically analysing SKROCK is notoriously difficult and current convergence results only hold for Gaussian target densities.

In this paper, we develop and study the Implicit Midpoint Langevin Algorithm (IMLA) for computation in Bayesian imaging models that are log-concave. Each iteration of IMLA can be viewed as a relaxed proximal-point step with a stochastic perturbation. We show that IMLA is asymptotically unbiased for Gaussian target distributions and converges in a provably accelerated manner for all distributions that are $\kappa$-strongly log-concave and smooth (or that are regularised by smoothing, e.g., by using a Moreau-Yosida approximation as described in Section 2), providing a rigorous alternative to SKROCK \cite{PVZ20} which is only provably accelerated for Gaussian models and has bias.

{The remainder of this paper is structured as follows: Section \ref{sec:Problem} introduces the class of Bayesian models considered as well as the overdamped Langevin diffusion process and its numerical approximations for models that are smooth and non-smooth. Section \ref{sec:discrete} presents the proposed IMLA method and provides detailed theoretical convergence guarantees for Gaussian and strongly log-concave models, including the optimal time-step choice to optimise convergence speed; the presented theory is illustrated with a range of toy examples. Following on from this, Section \ref{sec:Numerical} demonstrates IMLA on two challenging image restoration problems and by comparisons with SKROCK. Conclusions and perspectives for future work are finally reported in Section \ref{sec:conclusion}.}




\section{Problem Statement}\label{sec:Problem}

\subsection{Bayesian inference and imaging inverse problems}\label{subsec:BayesianInference}

We consider imaging inverse problems where we seek to estimate the unknown image $x \in \mathcal{X}$ from an observation $y$, related to $x$ by a forward statistical model with likelihood function $p(y|x)$. Typical examples include:
\begin{itemize}
    \item Gaussian likelihood $y \sim \mathcal{N}(Ax,\sigma^2 \mathbb{I})$:
    \begin{equation}\label{eq:GaussianLikelihood} p(y|x) \propto \exp\left( -\frac{||Ax-y||^2_2}{2\sigma^2}\right), \qquad \mathcal{X} = \R^d,\end{equation}
    \item Poisson likelihood $y \sim \mathcal{P}(Ax,\beta)$:
    \begin{equation}\label{eq:PoissonLikelihood}
    p(y|x) \propto \exp\left\{\sum_{i=1}^p[(Ax)_i + \beta -y_i\log((Ax)_i + \beta)]\right\}, \quad \mathcal{X} =\{x: (Ax)_i\geq 0\},
    \end{equation}
\end{itemize}
where $A \in \R^{p\times d }$ is a linear operator which models physical properties of the observation process, $\sigma>0$ in \eqref{eq:GaussianLikelihood}
is related to the intensity of the additive observation noise, and $\beta >0$ in \eqref{eq:PoissonLikelihood} represents a background Poisson noise level.

The operator $A \in \R^{p\times d }$ is usually rank-deficient or $A^TA$ has a poor condition number $\kappa \gg 1$, making the imaging problem ill-conditioned or ill-posed. To make the problem well-posed and deliver meaningful solutions one needs to incorporate regularisation \cite{Kaipio:stat-inv-problems05}. The Bayesian statistical framework introduces regularisation by defining a prior distribution $p(x)$ describing expected properties of the unknown image $x$. Observed and prior information are then combined by using Bayes' Theorem to obtain the model's posterior distribution, with density given by \cite{Kaipio:stat-inv-problems05}
\begin{equation}
    \pi (x) := p(x|y) = \frac{p(y|x)p(x)}{\int_{\mathcal{X}}p(y|x)p(x)dx}.
\end{equation}
For Bayesian computation, it is convenient to reformulate $\pi$ as follows:
\begin{equation}\label{eq:def_pi}
    \pi(x) \propto e^{-\pot(x)} = e^{-f_y(x) - g(x)}\, ,
\end{equation} 
with $f_y: \mathcal{X} \rightarrow \R$ and $g : \mathcal{X} \rightarrow \R $; representing the data likelihood and the prior information respectively. 

In this paper, we assume that $f_{y}$ and $g$ satisfy the following conditions:
\begin{enumerate}
    \item $f_{y}$ is convex with Lipschitz continuous gradient.
    \item $g$ is proper, convex and lower semi-continuous, but potentially non-smooth.
\end{enumerate}
Requiring $g$ to be convex on $\mathcal{X}$ significantly restricts the class of data-driven priors that can be considered. However, convexity provides some significant benefits in terms of guaranteeing the well-posedness of $p(x|y)$ and of posterior moments of interest (i.e., existence and stability w.r.t. changes in $y$, see \cite{hosseini-17} for details). Convexity is also central to the analysis of the convergence properties of the Bayesian computation methods used to perform inference w.r.t. $p(x|y)$: convexity allows for faster convergence rates, tighter non-asymptotic bounds, and stronger guarantees on the delivered solutions. 
Choices of $g$ that respect convexity include, e.g., the $\ell_1$-norm on some sparsifying basis, {the} total-variation pseudo-norm, {as well as} a number of recently proposed data-driven priors such as the convex ridge neural network regularizer \cite{goujon2022crrnn}, {the} convex variant of the variational network \cite{kobler-vn-2017}, or {the} input convex neural network \cite{mukherjee2021icnn}. Moreover, with regards to the Lipschitz condition on $\nabla f_y$, {even though} this condition is {not satisfied} in the case of {some} important non-Gaussian likelihoods (e.g. Poisson with $\beta = 0$, Binomial, and Geometric),  an appropriate approximation of the likelihood resolves this issue (see, e.g., \cite{MDAMZ}).

Drawing inferences about $x$ from $y$ by operating directly with $p(x|y)$ is usually not possible in imaging sciences because of the high dimensionality involved. As a result, imaging methods compute summaries of $p(x|y)$, such as Bayesian point estimators, posterior probabilities, moments, and expectations of interest. In particular, for the class of log-concave models considered here, a common choice is to compute the maximum-a-posteriori (MAP) estimator
\begin{equation}\label{eq:map-def}
    \hat{x}_{\mathrm{MAP}} = \arg\max_{x\in\mathcal{X}} \pi(x) = \arg\min_{x\in\mathcal{X}} \{f_y(x)+g(x)\},
\end{equation}
which is the optimal Bayesian estimator of $x$ w.r.t. to the $(f_y+g)$-Bregman loss \cite{pereyra-map} and can be efficiently computed by proximal convex optimisation techniques \cite{chambolle_pock_2016, Combettes2011}. Other Bayesian point estimators and more complicated Bayesian analyses such as uncertainty quantification, calibration of model parameters, model selection and hypothesis testing, are generally beyond the scope of optimisation algorithms as they require calculating probabilities and moments with respect to $\pi$. However, as mentioned previously, posterior probabilities and moments of interest can be computed efficiently by using proximal MCMC methods derived from the Langevin diffusion process \cite{langevinmeetsmoreau18}. These algorithms can be applied directly to perform inference w.r.t. $\pi$ (see, e.g., \cite{10.1093/mnras/sty2004,MDAMZ}), or embedded in more complex inference machinery to calibrate unknown model parameters \cite{vidal2019maximum} or perform model selection \cite{Cai2022}.

\subsection{Sampling using the overdamped Langevin equation}\label{subsec:Lang_cont}

The fastest provably accurate Bayesian computation algorithms to perform inference w.r.t. log-concave models such as $\pi$ are derived from the overdamped\footnote{\review{We use the term overdamped here to distinguish from the underdamped Langevin SDE also called the kinetic Langevin SDE, the overdamped Langevin SDE is obtained in the large friction limit of the underdamped Langevin SDE \cite{Lelievre2010}.}} Langevin SDE (OLSDE)\review{\cite{roberts1996,RobertChristianP.2004MCsm}}.  Let $\bar{\pi}$ be a smooth density that is log-concave on $\mathcal{X}$ and consider the OLSDE
\begin{equation}\label{eq:OLSDE}
    dX_t= \nabla \log\bar{\pi}(X_t) dt + \sqrt{2}dW_t,
\end{equation}
where $\{W_t\}_{t\geq 0}$ is a $d$-dimensional Brownian motion. Under mild conditions on $\bar{\pi}$, this SDE is well-posed and converges provably quickly to $\bar{\pi}$ (its unique invariant distribution), providing a foundation for fast and rigorous MCMC algorithms to perform computations for $\bar{\pi}$.

The OLSDE does not have an analytic solution in general so we 
consider numerical approximations of this SDE. The most widely used approximation is the Unadjusted Langevin Algorithm (ULA) given by
\begin{equation} \label{eq:ULA}
X_{n+1}=X_{n}+\delta \nabla \log{\bar{\pi}(X_{n})}+\sqrt{2\delta}\xi_{n},
\end{equation}
where $\delta >0$ and $\{\xi_{n}\}_{n \geq 1}$ is a sequence of $d$-dimensional standard Gaussian random variables. Despite its simplicity, ULA can be remarkably fast and scale efficiently to very large problems \cite{durmus2017}. However, a necessary condition to establish the convergence of ULA is that $\delta < L^{-1}$ where $L$ is the Lipschitz constant of $\nabla \log\bar\pi(x)$ \cite{durmus2017}. Unfortunately, in problems that are ill-conditioned, as commonly encountered in imaging sciences, this leads to $\delta$ being very small in comparison to the time scales required for convergence to $\pi$, which can be shown to be of order $\kappa = L/m$ when $\pi$ is $m$-strongly log-concave. The slow convergence leads to highly correlated samples and to an inefficient exploration of the solution space. This issue can be addressed by designing specialized stochastic integrators such as the SKROCK algorithm \cite{abdulle18,PVZ20}, which behave in an accelerated manner in the sense that the iterates converge to $\pi$ with a time scale of order $\sqrt{\kappa}$. Because often in Bayesian imaging problems $\sqrt{\kappa} \ll \kappa$, SKROCK can lead to very significant reductions in computing time w.r.t. ULA (e.g., \cite{PVZ20,pereyra2023split} report speed-up factors of the order of \review{ $20$ to $30$}). 

Despite this remarkable empirical performance, very little is known about the theoretical convergence properties of SKROCK, with existing results only covering the case where $\bar\pi$ is Gaussian. As a result, whenever SKROCK is applied to a new imaging problem it is necessary to benchmark its performance against a more reliable method (e.g., ULA) on some test problems. In addition, the lack of theory for SKROCK makes it difficult to embed it within more complex inference schemes such as \cite{vidal2019maximum}, or to provide practitioners with principled strategies to set its internal parameters. The method that we present in Section \ref{sec:discrete} seeks to address these issues by providing a new accelerated Langevin MCMC scheme that is empirically as computationally efficient as SKROCK, but has a much stronger theoretical footing. Before concluding this section, we briefly discuss the case where $\pi$ is not smooth.

\subsection{Non-smooth distributions}\label{subsec:nonsmooth}
The Bayesian models encountered in imaging sciences are often not smooth (e.g., because of the use of non-smooth norms or because of constraints on the solution space), and this presents some challenges to the application of gradient-based strategies such as ULA and SKROCK that require smoothness. For models with an underlying convex geometry, proximal MCMC methods \cite{PereyraMarcelo2016PMcM,langevinmeetsmoreau18} deal with the lack of smoothness by replacing $\pi$ with a smooth approximation $\pi^\lambda(x)$ such as 
\begin{equation}\label{eq:MYposterior}
    \pi^\lambda(x) \propto \exp\left(-\lk(x)-\pr^\lambda(x)\right)
\end{equation}
where $g^\lambda$ is the Moreau-Yosida envelope of $g$ defined by
\begin{equation*}
    g^\lambda(x) = \min_{u\in\R^d} \left\{g(u)+\frac{1}{2\lambda}\lVert x-u\rVert^2\right\}.
\end{equation*}
Approximations based on other splittings can be considered too (see, e.g., \cite{vidal2019maximum}). The density $\pi^\lambda$ is log-concave and Lipschitz smooth, with gradient
    \begin{align}\label{eq:moreauIdentity}
        \nabla \log \pi^\lambda(x) &= -\nabla \lk(x) - \nabla \pr^\lambda(x)\\
        &=  -\nabla \lk(x) - \frac{1}{\lambda}(x-\prox_{\theta g}^\lambda(x)),
    \end{align}
    with Lipschitz constant $L\leq L_f+\lambda^{-1}$. Here $\prox_{g}^\lambda$ is the proximal operator of $g$ given by
    \begin{equation}\label{eq:prox}
        \prox_g^\lambda (x)  = \arg\min_{u\in \R^d} \left\{ g(u) + \frac{1}{2\lambda}\lVert x-u\rVert^2\right\}.
    \end{equation}
Since the OLSDE with $\pi$ replaced by $\pi^\lambda$ is well-defined, we can consider sampling algorithms derived from discretisations of this SDE. {In  particular, applying an OLSDE-based algorithm such as ULA or SKROCK to $\pi^\lambda$ is considered in \cite{langevinmeetsmoreau18,PVZ20} as a means of performing approximate computation w.r.t. $\pi$. Indeed, since $\pi^\lambda$ converges to $\pi$ in Total Variation as $\lambda \rightarrow 0$ and can be made arbitrarily close to $\pi$ by reducing $\lambda$, see \cite[Proposition~1]{langevinmeetsmoreau18}, the samples generated by targeting $\pi^\lambda$ reliably approximate $\pi$.}

It is worth mentioning at this point that {both ULA and SKROCK} produce biased Monte Carlo samples because of the error stemming from the discretisation of the OLSDE. Of course, there is also non-asymptotic bias stemming from the fact that the algorithms run for a finite number of iterations. In principle, the asymptotic bias related to the discretisation error could be removed by using a Metropolis-Hastings (MH) correction step, as in the Metropolis Adjusted Langevin Algorithm (MALA) \cite{pereyra:hal-01312917}. However, in high-dimensional problems that are ill-conditioned or ill-posed, introducing an accept-reject mechanism can lead to drastically slower convergence and to a large increase in the non-asymptotic error (see, e.g., \cite{langevinmeetsmoreau18}). The development of unadjusted proximal Langevin methods with smaller bias is therefore valuable, particularly in imaging problems that require computing tail probabilities or high-order moments that are more sensitive to the estimation bias than the posterior mean.



\section{Proposed method: IMLA} \label{sec:discrete}
We are now ready to present our proposed accelerated Bayesian computation method to perform inference w.r.t. $\pi$, for which we are able to provide detailed convergence results including non-asymptotic convergence bounds and an optimal choice of the discretisation time step to maximize the convergence speed. The method is based on the following recursion:
\begin{equation}\label{eq:IMLAprox}
\textrm{IMLA}:\quad   X_{n+1} = \left(1-\frac{1}{\theta}\right)X_n + \frac{1}{\theta}\prox_{\pot}^{\delta\theta}(X_n+\theta\sqrt{2\delta}\xi_n)\,,
\end{equation}
for $\theta\in [0,1]$, $\delta>0$ and $\{\xi_n\}_{n\geq 1}$ a sequence of independent $d$-dimensional standard Gaussian random variables. Recursion \eqref{eq:IMLAprox} is a stochastic modification of a relaxed proximal point algorithm to compute the MAP solution $\hat{x}_{MAP}$ (definition see \eqref{eq:map-def}), or equivalently to minimize the potential $U(\cdot) = -\log \pi(\cdot)$ \cite{Combettes2011}.

In order to analyze \eqref{eq:IMLAprox} it is useful to express it explicitly as a minimisation problem.\footnote{Recall that $\pot$ is defined in \eqref{eq:def_pi}.}
\begin{equation}\label{eq:implicitscheme}
    \begin{aligned}
        X_{n+1}&=\mathrm{arg min}_{x\in \mathbb{R}^d} F(x;X_n;\xi_{n+1}),\\
        F(x;u,z) &:= \theta^{-1}\pot(\theta x+(1-\theta) u) +\frac{1}{2\delta}\lVert x-u-\sqrt{2\delta}z\rVert^2\,.
    \end{aligned}
\end{equation}
\review{To obtain this formulation note that \eqref{eq:implicitscheme} expresses $X_{n+1}$ in terms of the proximal operator for the function $U$ precomposed with an affine transformation, that is, for $X_n$ fixed, set $\tilde{U}(x) = U(\theta x+(1-\theta) X_n)$ then \eqref{eq:implicitscheme} is equivalent to the expression
\begin{equation}\label{eq:implicitscheme_prox_form}
    X_{n+1}=\prox_{\tilde{U}}^{\delta/\theta}(X_n+\sqrt{2\delta}\xi_{n+1}).
\end{equation}
We now use the following property of proximal operators (see \cite[Section~2.2]{parikh2014proximal} for details) that if $\varphi(x)=\psi(\alpha x+\beta)$ for $\alpha,\lambda>0$ and $\beta \in\R^d$ then
\begin{equation*}
    \prox_\varphi^\lambda (x) = \frac{1}{\alpha}( \prox^{\alpha^2\lambda}_\psi(\alpha x+\beta)-\beta).
\end{equation*}
Setting $\psi=U, \alpha =\theta, \lambda =\delta/\theta$ and $\beta = (1-\theta) X_n$ we have
\begin{equation}\label{eq:prox_precomposed}
    \prox_{\tilde{U}}^{\delta/\theta}(X_n+\sqrt{2\delta}\xi_{n+1}) = \frac{1}{\theta}\left( \prox^{\theta \delta}_U( X_n+\sqrt{2\delta}\theta\xi_{n+1})-(1-\theta) X_n\right).
\end{equation}
By combining \eqref{eq:implicitscheme_prox_form} and \eqref{eq:prox_precomposed} we have that \eqref{eq:implicitscheme} is equivalent to \eqref{eq:IMLAprox}.}
Note that the scheme \eqref{eq:implicitscheme} is well defined when $\pot$ is convex or {under} the weaker condition {$\pot \in \mathcal{C}^{2}$} with the second order derivative $\delta^{-1} I + \theta \nabla^2\pot \succeq 0$. 

When $U \in \mathcal{C}^1$, \review{$X_{n+1}$ is the minimiser of $F$ if and only if $\nabla F(X_{n+1};X_n,\xi_{n+1}) = 0$} and therefore \review{by rearranging $\nabla F(X_{n+1};X_n,\xi_{n+1})=0$ the expressions} \eqref{eq:IMLAprox} and \eqref{eq:implicitscheme} are also equivalent to the following implicit (midpoint) approximation of an OLSDE \eqref{eq:OLSDE} targeting the posterior density of interest $\pi(x) \propto \exp\{-\pot(x)\}$:
\begin{equation}\label{eq:thetamethod}
     X_{n+1}=X_n -\delta\nabla U\left(\theta X_{n+1}+(1-\theta) X_n\right) + \sqrt{2\delta} \xi_{n+1}\, .
\end{equation}
The recursion \eqref{eq:thetamethod} belongs to the class of numerical integration methods commonly known as ``$\theta$-methods'', which combine explicit and implicit first-order gradient steps \cite{HNW91}. We are particularly interested in the case $\theta=1/2$, where \eqref{eq:thetamethod} is a stochastic modification of the implicit midpoint method, motivating the name of our proposed algorithm: Implicit Midpoint Langevin Algorithm (IMLA). Observe that when $\theta=0$, \eqref{eq:thetamethod} reduces to ULA, and when $\theta=1$, this corresponds to the implicit Euler scheme applied to OLSDE \cite{hodgkinson2021implicit}(ILA).

Moreover, by using the convexity of $U$, the identity \eqref{eq:moreauIdentity}, and setting $\theta = 1/2$, we find that the sequence of iterates $Y_n = X_n+\theta\sqrt{2\delta}\xi_n$ is described by the following non-Markovian\footnote{\review{This scheme is non-Markovian since $Y_{n+1}$ depends on the noise term $\xi_{n+2}$ which is correlated with the noise introduced in $Y_{n+2}$.}} scheme:
\begin{equation}\label{eq:LM}
    Y_{n+1}=Y_n -\delta\nabla U^{\delta/2}\left(Y_n\right) + \sqrt{2\delta} {\left(\frac{\xi_{n+1}}{2}+\frac{\xi_{n+2}}{2}\right)}\, ,
\end{equation}
which corresponds to a Leimkuhler–Matthews (LM) discretisation \cite{LM12} of a Langevin diffusion targeting the Moreau-Yosida approximation $\pi^{\delta/2}$ of $\pi$ given by \eqref{eq:MYposterior}, which is well-posed even when $U$ is not smooth. Accordingly, the post-processed  \cite{V15} iterates $X_n = Y_n - \sqrt{\delta/2}\xi_n$ converge to an approximation of $\pi$ that stems from correcting $\pi^{\delta/2}$ by performing a deconvolution of $\pi^{\delta/2}$ with the inverse of a Gaussian smoothing kernel of bandwidth $\delta/2$.

With regards to the convergence properties of the method, both the midpoint approximation \eqref{eq:thetamethod} with $\theta = 1/2$ and the LM approximation \eqref{eq:LM} are known to be exact for Gaussian target distributions. That is, despite the discretisation of the OLSDE, the asymptotic bias vanishes completely as $n$ increases \cite{KCZ11,LM12}. For more general log-concave distributions that are not Gaussian, the methods exhibit some bias, but are usually considerably more accurate than the conventional ULA schemes based on an Euler-Maruyama approximation of the OLSDE, such as MYULA \cite{durmus2017}. 

Moreover, while a small bias is desirable, a primary concern in Bayesian computation for imaging problems is to achieve a fast convergence rate without incurring excessive bias. And while the dimension of the problem can affect the rate of convergence, the main challenge is to mitigate the impact of the high anisotropy of the target density $\pi$, which is inherent to imaging problems that are ill-conditioned or ill-posed. More precisely, as mentioned previously, the convergence rate of ULA strategies deteriorates as $\kappa$ for target densities that are $\kappa$-strongly log-concave \cite{durmus2017}. The accelerated algorithm SKROCK \cite{PVZ20} empirically achieves a significantly better rate that only deteriorates as $\sqrt{\kappa}$, but this can only be verified theoretically in the case of a Gaussian target distribution.

Sections \ref{subsec:Gaussian} and \ref{subsec:Convex} below present a detailed non-asymptotic convergence analysis of \eqref{eq:IMLAprox} which shows that the method achieves acceleration for all smooth distributions tested that are strongly log-concave. The analysis also allows for identifying the optimal time step $\delta$ that maximizes the convergence speed of the method. It is worth mentioning at this point that our proposed method has similarities to the implicit schemes proposed in \cite{hodgkinson2021implicit, Wibisono18}, as they all coincide  in the Gaussian case. We have chosen to use \eqref{eq:IMLAprox} as opposed to the schemes described in \cite{hodgkinson2021implicit, Wibisono18} because \eqref{eq:IMLAprox} is well-posed for models that are not smooth, and because our analysis shows that \eqref{eq:IMLAprox} has better non-asymptotic convergence bounds.

It is also worth mentioning at this point that if $\pot$ is $L$-smooth and $m$-strongly convex, then $F$ is strongly convex with conditioning number $\frac{1 + \theta L \delta}{1 + \theta m \delta}$ meaning that from a computational point of view, for any $\delta$, one step of the algorithm is at most as expensive as minimizing $U$ directly to compute a MAP solution. In practice, a step of the algorithm is much cheaper because it is possible to use past iterations to warm-start computations. This, combined with an efficient solver, allows implementing \eqref{eq:IMLAprox} with a computational cost that is comparable to explicit accelerated methods from the state of the art such as SKROCK \cite{PVZ20}.

\subsection{Analysis in the Gaussian case}\label{subsec:Gaussian}
We start our analysis for our proposed method from the representation \eqref{eq:thetamethod} in the case where 
\begin{equation} \label{eq:quad}
\pi(x) \propto \exp \left( -\frac{1}{2}x^{T}\Sigma^{-1} x \right), \quad \Sigma=\text{diag}(\sigma^{2}_{1}, \cdots, \sigma^{2}_{d}),
\end{equation}
and study its convergence in the $2$-Wasserstein distance, as a function of the parameter $\theta$ and the condition number $\kappa=\sigma^{2}_{\text{max}}/\sigma^{2}_{\text{min}}.$ In this case,  because $\log{\pi}$ is quadratic, we can solve directly for $X_{n+1}$ in \eqref{eq:thetamethod}. As the covariance matrix is diagonal, each co-ordinate in \eqref{eq:OLSDE} becomes independent from the rest. {Note that the same analysis could be done for any covariance matrix $\Sigma$ using singular value decomposition}.

One step of IMLA for the $i$-th coordinate in the case of \eqref{eq:quad} is given by 
\begin{equation} \label{eqn:generalNumMethod}
    X^{i}_{n+1}=R_{1}(z_{i})X^{i}_{n}+\sqrt{2\delta }R_{2}(z_{i})\xi^{i}_{n}, \quad \xi^{i}_{n}\sim \mathcal{N}(0,1),
\end{equation}
where $z_{i}=-\Delta t /\sigma^{2}_{i}$ and 
$X_{0}=(X^{1}_{0},\cdots,X^{d}_{0})$ is a deterministic initial condition, while
\begin{equation}\label{eq:R1R2}
R_{1}(z) =\frac{1+(1-\theta)z}{1-\theta z}, \quad R_{2}(z)=\frac{1}{1-\theta z}.
\end{equation}
Using the fact that Gaussian distributions are closed under 
linear transformations, and assuming that the initial 
condition  $X_{0}$ is deterministic, we derive the 
distribution of $X_{n}$ for any $\delta > 0$. In fact, we have 
the following proposition \cite{PVZ20}
\begin{proposition} \label{prop:first}
	Let $\pi (x) \propto \exp{(-\frac{1}{2} x^T \Sigma^{-1} x)}$ with $\Sigma = \text{diag}(\sigma_1^2,...,\sigma_d^2)$, and let $Q_{n}$ be the probability measure associated with $n$ iterations of the generic Markov kernel \eqref{eqn:generalNumMethod}. Then the 2-Wasserstein distance between $\pi$ and $Q_{n}$ is given by
	\begin{equation}\label{eqn:wassersteinDistanceFinal}
	W_2(\pi;Q_n )^2 = \sum_{i=1}^d \left(D_n(z_i,x_0^i) + B_n(z_i,\sigma_i) \right)
	\end{equation}
	where
	\[
	D_n(z,x)  = (R_1(z))^{2n} x^2, \quad 
	B_n(z,\sigma) = \left[\sigma - \sqrt{2\delta} R_2(z) \left( \frac{1 - (R_1(z))^{2n}}{1 - (R_1(z))^2} \right)^{1/2} \right]^2. 
	\]
	{In addition the following bound holds
	\begin{equation}\label{eq:bbound}
	W_{2}(\pi;Q_{n} )  \leq  W_{2}(\pi;\tilde{\pi} )+C^{n} W_{2}(\tilde{\pi},Q_{0} )
	\end{equation}
	where 
	\begin{equation}\label{eq:numericalinvariant}
	\tilde{\pi} = \mathcal{N}\left(0,2\delta (R_{2}(z))^{2}\left[\frac{1}{1-R^{2}_{1}(z)}\right]   \right)
	\end{equation}
	is the numerical invariant measure and
	\begin{equation} \label{eq:contb}
	C=\sqrt{\max_{1\leq i \leq d}R_{1}(z_{i})^{2}}.
	\end{equation}
	}
\end{proposition}
A closer look at the bound \eqref{eq:bbound} reveals that the constant $C$ controls how fast the numerical algorithm converges to its equilibrium behaviour and the term $W_{2}(\pi;\tilde{\pi} )$ characterizes the bias of the numerical algorithm at equilibrium. The first step in analyzing the convergence of the algorithm is thus to understand the behaviour of the constant $C$ for different values of the parameter $\delta$. Surprisingly, this analysis will also be relevant in the general (non-Gaussian) log-concave case, since the convergence to equilibrium of this numerical scheme is controlled by the same constant. The analysis of the behaviour of the constant $C$ as a function of $\delta$ for different values of $\theta$ is presented 
in detail in the Appendix, and can be used to determine the optimal number of steps $n$ required such that $W_{2}(\pi,Q_{n})^{2}<\epsilon^{2}$ according to the following proposition: 
\begin{proposition} \label{prop:summary}
Let $Q_{n}$ be the probability measure associated with the $n$-th iteration of the generic Markov kernel \eqref{eqn:generalNumMethod} starting at $X_{0}$. Then the number of steps $n$ required such that $W_{2}(\pi,Q_{n})^{2}<\epsilon^{2}$ is given by 
\begin{equation} \label{eq:n_steps}
n \approx
\begin{cases}
	\frac{\sqrt{\kappa}}{2}\big[\log\big(W_2(\pi,Q_0)\big)-\log(\epsilon)\big]		, & \theta=\frac{1}{2} \\
& \\
            \min\left( \frac{d\sigma_{\max}^{2}}{\epsilon^2}, \frac{\sqrt{d\kappa}\sigma_{\max}}{2\epsilon}\right) \big[\log(W_2(\pi,Q_0))-\log(2^{-1}\epsilon)\big], & \theta=1
		 \end{cases}
\end{equation}
with $\delta$ given by
\begin{equation*}
    \delta = \begin{cases}
        \delta_\ast, & \theta=\frac{1}{2}\\
         \max\left(\frac{\epsilon^2}{d},\frac{2\epsilon\sigma_{\min}}{\sqrt{d}}\right) & \theta=1.
    \end{cases}
\end{equation*}
where $\delta_\ast=\frac{2}{\sqrt{Lm}}=2\sigma_{\min}\sigma_{\max}$.
\end{proposition}
\begin{proof}
The proof of this proposition can be found in the Appendix \ref{proof-prop-2}.
\end{proof}
{In principle, we are able to repeat the analysis in Proposition \ref{prop:summary} for all the different values of $\theta$, but the two cases of interest are the ones where $\theta=1/2$ or $\theta=1$. In the case $\theta=1/2$, the expression for $n$ does not depend on $d$, while in the case of $\theta=1$ it does depend on $d$. For $\theta=1$, the dependence on $\epsilon$ is worse than logarithmic and the expression of $n$ remains bounded as $\kappa\uparrow \infty$. }
{These bounds show that $\theta =1/2$ will outperform $\theta = 1$ unless $d\sigma_{\max}^2/(\sqrt{\kappa} \epsilon^2)$ is small.}




 In order to visually illustrate the results of Proposition \ref{prop:summary} we now proceed with the following numerical experiment: We choose $X_0 = 1/\sqrt{d}$ (so that $\norm{X_0} = 1$) and $\sigma_i$ geometrically spaced between $1$ and $1/\sqrt{\kappa}$ {to mimic an ill-conditioned posterior distribution}. We also plot results for $\theta=0$ to compare with an explicit method in addition to the cases $\theta=1/2$ and $\theta=1$. 
In the case of $\theta=1/2$ since $W_{2}(\tilde{\pi},\pi)=0$ we have performed all our calculations for $\delta=\delta_{*}$, because this choice maximizes the convergence speed to equilibrium.  For $\theta=1$, we numerically determined the largest $\delta$ for which the bias $W_2(\pi,\tilde \pi)$ (which has a closed-form expression) is $\leq \epsilon/2$. An explicit calculation is presented in the Appendix and shows the following dependence 
\begin{equation}\label{eq:delta_theta1}
\delta = \max\left(\frac{\epsilon^2}{d},\frac{2\epsilon}{\sqrt{d\kappa}}\right).
\end{equation}
Similarly, for $\theta = 0$, we numerically computed the largest $\delta$ such that $C<1$ and $W_2(\pi,\tilde \pi)\leq \epsilon/2$, and subsequently determined the smallest $n$ for which $W_2(\pi,Q_n)$ given by \eqref{eqn:wassersteinDistanceFinal} is $\leq \epsilon$.

\begin{figure}[t]
\centering
\begin{subfigure}{.49\textwidth}
\includegraphics[width=\linewidth]{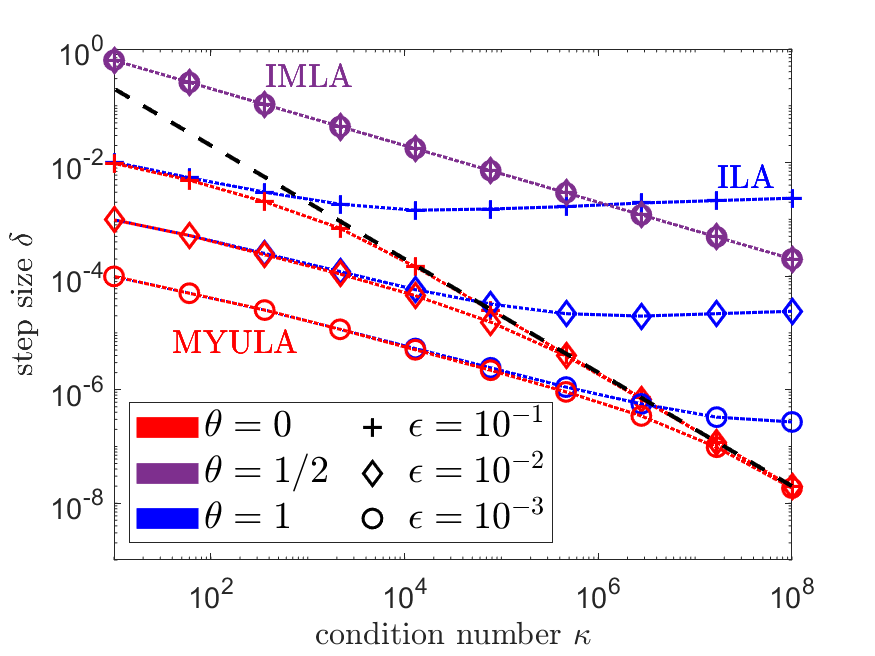}
\caption{$\delta$ against $\kappa$}
\label{figure1}
\end{subfigure}
\begin{subfigure}{.49\textwidth}
\includegraphics[width=\linewidth]{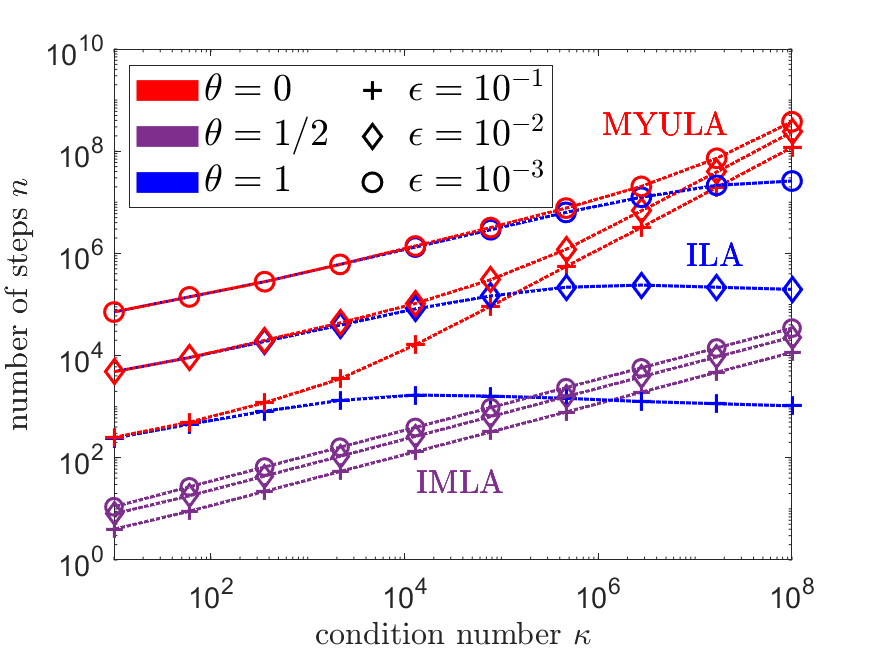}
\caption{$n$ against $\kappa$}
\label{figure2}
\end{subfigure}
\caption{(a) step-size $\delta$ as a function of the condition number $\kappa$ for different choices of $\theta$ (the black dotted line corresponds to $\kappa\delta =2$, the stability limit of the explicit integrator) and level of accuracy $\epsilon$. (b) Number of steps $n$ such that  $W_{2}(\pi,Q_{n})^{2}<\epsilon^{2}$ as a function of the condition number $\kappa$ for different choices of $\theta$ and $\epsilon$.}
\label{fig:conv-plot}
\end{figure}

{We summarise the results of these experiments\footnote{We have used $d=100$ and $\epsilon = 0.1, 0.01, 0.001$.} in Figures \ref{fig:conv-plot}(\subref*{figure1}) and  \ref{fig:conv-plot}(\subref*{figure2}). In Figure \ref{fig:conv-plot}(\subref*{figure1}) we plot the step size $\delta$ as a function of the conditioning number $\kappa$. We observe that the behavior is as predicted by \eqref{eq:delta_theta1} in Proposition \ref{prop:summary}: For $\theta = 1$ and fixed $\epsilon$, the step size $\delta$ decreases as $1/\sqrt{\kappa}$ initially, and eventually stops decreasing once $\delta$ becomes the order of $1/\kappa$ for increasing $\kappa$. The behaviour of $\delta$ as $\epsilon$ changes is as follows: On the left-hand side of the figure, dividing $\epsilon$ by 10 results in a division by 10 of $\delta$, whereas on the right-hand side of the figure, dividing $\epsilon$ by 10 leads to a decrease in $\delta$ by a factor of 100. The figure shows that the value of $\delta$ for $\theta = 1$ is smaller than the optimal value for $\theta = 1/2$ except for $\epsilon = 0.1$ and very large $\kappa$. In the case where $\theta = 0$, $\epsilon$ is fixed and $\kappa$ is varying there are two regimes: For small $\kappa$, the value of $\delta$ is chosen to ensure that the bias is sufficiently small, and the results are very similar\footnote{For $\theta=0$ the bias can be shown similarly to the proof of the proposition \eqref{eq:delta_theta1} to have a bound proportional to $\sqrt{d\kappa}\delta$ provided that $\delta$ is bounded away from the stability limit $2/\kappa$.} to those corresponding to $\theta=1$. However, for $\kappa$ large, the value of $\delta$ is really determined by the stability limit.}

{Figure \ref{fig:conv-plot}(\subref*{figure2}) gives $n$ as a function of $\kappa$ for the runs in Figure \ref{fig:conv-plot}(\subref*{figure1}). We again observe that the behavior is as predicted by \eqref{eq:delta_theta1} in Proposition \ref{prop:summary}: For $\theta=1/2$, increasing $\kappa$ and fixed $\epsilon$, $n$ increases in the order of $\sqrt{\kappa}$ and changes insignificantly depending on $\epsilon$. This is not the case for $\theta = 1$ or $\theta = 0$; in these cases, a reduction in $\epsilon$ implies a reduction in $\delta$ in order to decrease the bias, which in turn increases the contractivity constant $C\approx 1-\delta$. As a consequence, in the case of $\theta = 1$, $n$ initially increases like $\sqrt{\kappa}$ and then becomes constant as predicted. Only when $\epsilon = 0.1$ and $\kappa$ is very large, $\theta = 1$ improves on $\theta =1/2$. In the case $\theta =0$ and $\kappa$ large, the number of steps required behaves as $1/\kappa$ and is up to four orders of magnitude larger than the number of steps for $\theta= 1/2$. Therefore, for large $\kappa$, the implicit integrator with $\theta=1/2$ is to be preferred to the explicit algorithm even if its computational complexity per step is substantially higher.}

\subsection{Analysis in the strongly log-concave case}\label{subsec:Convex}

We extend the Gaussian non-asymptotic analysis to the strongly log-concave case. The following theorem is the analogue of Proposition \ref{prop:first}.

\begin{theorem}\label{thm:nonasymconv}
    Let $\pot=-\log \pi$ and suppose that $\pot\in C^2$, $m$-strongly convex and has gradient which is $L$-Lipschitz. For any probability distribution, $Q_0$, let $Q_n$ denote the probability distribution of $X_{n}$ where $X_0\sim Q_0$ and $X_k$ is given according to the scheme \eqref{eq:implicitscheme}. We assume that each step of the scheme \eqref{eq:implicitscheme} is solved to a tolerance of $\varepsilon$, i.e. that $\lVert \nabla F(X_{n+1};X_n,\xi_n)\rVert \leq \varepsilon$ for every $n$. Then for any initial probability distribution $Q_0$ with finite second moments we have
\begin{align}\label{eq:non_asm_bound}
W_2(Q_n,\pi)\leq & C^{n}W_2( Q_0,\pi) + \frac{1-C^{n+1}}{1-C}\frac{\frac{1}{2}\delta^2 L^{\frac{3}{2}}\sqrt{d} +\frac{2}{3}L\delta^{\frac{3}{2}}\sqrt{2d}+\varepsilon\delta}{1+\theta\delta m}
\end{align}
and
\begin{equation}\label{eq:Cnonlinear}
    C=\sqrt{\max_{z\in [m\delta,L\delta]}R_{1}(-z)^{2}}.
\end{equation}
\end{theorem}

\begin{proof}
The proof is  an adaptation of  \cite[Theorem 2]{hodgkinson2021implicit} for the numerical scheme \eqref{eq:implicitscheme}  and is deferred to the Appendix \ref{proof-theorem}.
\end{proof}

The bound we have derived in \eqref{eq:non_asm_bound} has two terms, the numerical contraction rate and the bias. The rate of contraction $C$ we obtain in Theorem \ref{thm:nonasymconv} is sharp. The constant $C$ here is the same as the contraction rate for the Gaussian problem. This is a significant improvement upon the conclusions found in \cite{hodgkinson2021implicit}. Indeed, \cite[Theorem 2]{hodgkinson2021implicit} derives a contraction rate of $\kappa_\delta C^n$ for the numerical scheme 
\begin{equation}\label{eq:old_theta}
    X_{k+1}=X_{k}-\theta\delta\nabla \pot( X_{k+1})+(1-\theta)\nabla \pot(X_k) +\sqrt{2\delta}\xi_k.
\end{equation}
Here, $\kappa_\delta$ is given by $\kappa_\delta=\frac{2+\theta \delta L}{2+\theta \delta m}$. 
On the other hand, the bias term in \eqref{eq:non_asm_bound} is not sharp and is consistent with the bias found in \cite[Theorem 2]{hodgkinson2021implicit}. As in the Gaussian case, the following proposition uses the bound in \eqref{eq:non_asm_bound} to determine the optimal number of steps to achieve accuracy $\epsilon$.
\begin{proposition}\label{prop:n_nonlinear}
    Let $Q_n$ be the probability measure associated with the $n$-th iteration of the Markov kernel \eqref{eq:thetamethod} with $\theta=1/2$ starting at $X_0$. Then for any $\epsilon>0$ there exists $\delta>0$ such that for 
    \begin{equation*}
        n\approx \max\left\{\frac{\sqrt{\kappa}}{4}, \frac{2\kappa\sqrt{\kappa d}}{2\sqrt{m}\epsilon}, \frac{64}{9}\frac{d\kappa}{m \epsilon^2}\right\} [\log(\epsilon/2)-\log(W_2(\pi,Q_0))]
    \end{equation*}
    we have $W_2(Q_n,\pi)\leq \epsilon$.
\end{proposition}
 As noted above, the bias in \eqref{eq:non_asm_bound} is not sharp and as a result the description of the small $\epsilon$ regime is suboptimal. However, if $\epsilon$ is sufficiently large then the number of steps grows with $\sqrt{\kappa}$ and hence in this regime the result is consistent with the Gaussian setting.

\subsection{Illustrative experiment: a toy denoising model based on a Gaussian mixture prior}\label{sec:gmm-model}
{We now consider} a  Gaussian mixture model for denoising an image $x$ from observations $y = x + \xi$ with $\xi \sim \mathcal{N}(0, \sigma^2 I_d)$ with noise variance $\sigma^2$. This choice allows us to demonstrate our approach with a posterior distribution that deviates from a Gaussian density but for which we can compute quantities exactly. More precisely, 
the prior is the product of $d$ 1-dimensional identical Gaussian mixtures of two Gaussian  random variables $\tilde{Z}_k$ with means $m_k \in \mathbb{R}$, variances $\sigma_k^2 \in \R$ for $k\in\{0, 1\}$ and weights $\tilde{\omega}$ and $1-\tilde{\omega}$ for some $\tilde{\omega} \in[0,1]$.  It is possible to show that the posterior distribution is now a product  of $d$ 1-dimensional Gaussian mixtures of two independent Gaussian random variables $Z_k$ with means $\mu_k(y_i)$, variances $\delta_k^2$ and weights $w(y_i)$. Note that the means and the weights of each of the $d$  Gaussian mixtures are now depending on the pixel values $y_{i}$ (the precise relationship can be found in the Appendix \ref{gmm-appendix}).

\begin{table}[t]
    \centering
    \begin{tabular}{|c|c|c|c|c|c|c|}
    \hline
       $m_0$  & $m_1$ & $\sigma^{2}_{0}$ & $\sigma^{2}_{1}$  & $\sigma^{2}$ & $\tilde{\omega}$ \\
       \hline\hline
        0 & 0 & 0.0025 & 0.0809 & 0.0016 & 0.9 \\
         \hline
        \end{tabular}
   \caption{Parameters for the Gaussian mixture model for denoising.}
    \label{tab:model_param}
\end{table}

In Section \ref{subsec:Gaussian} we showed that IMLA samples are exact for Gaussian distributions. Here, we investigate the bias for a non-Gaussian problem. Note that the posterior distribution associated with this mixture prior is not usually strongly log-concave. However, for the range of value of pixels $y_i$ that we have used in our experiments, as well as  our choices of prior parameters $(m_{0},m_{1},\sigma^{2}_{0},\sigma^{2}_{1},\tilde{\omega})$ and noise variance $\sigma^{2}$ that can be seen in Table \ref{tab:model_param}, the corresponding posterior distribution is indeed strongly log-concave and fits the assumptions of our theory.

\begin{figure}[!t]
\centering
    \includegraphics[width=0.8\linewidth]{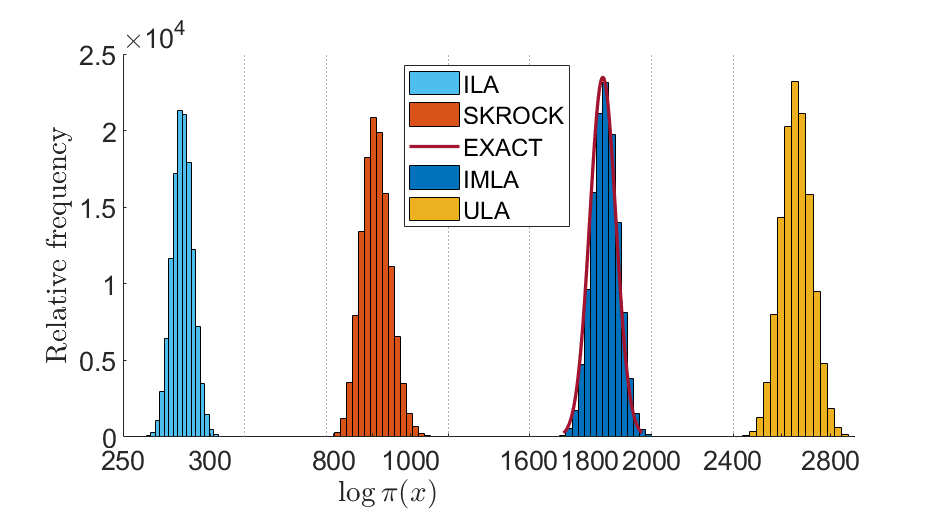}
    \includegraphics[width=0.8\linewidth]{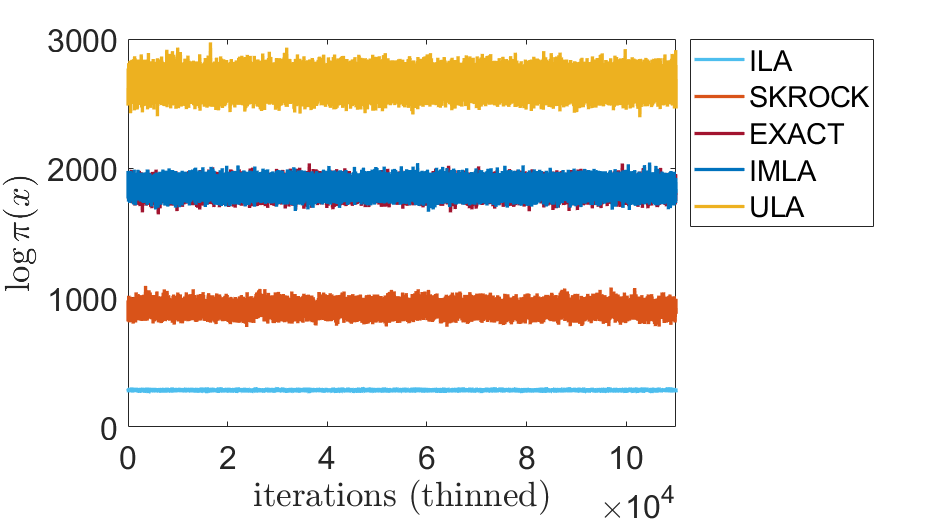}
    \captionof{figure}{\review{Top: Comparison of histograms of the scalar statistic $\log \pi(x)$ for the Gaussian mixture model. This figure shows the different biases induced by the sampling algorithms, and significant agreement between the exact and the IMLA statistic. Bottom: Traces for the same data showing the stationary behavior of the chains.}}
    \label{fig:logpigmm}
\end{figure} 

\begin{table}[t]
\centering
\begin{center}
\begin{tabular}{|c|c|c|}
\hline
Algorithm & \multicolumn{1}{|p{1.5cm}|}{\centering Mean \\$\mathcal{W}_2(\pi, \Tilde{\pi})$} & Std.  Dev.     \\ \hline\hline
EXACT      & 1.4123e-07                       & 1.4448e-07 \\ \hline
IMLA       & 1.4095e-07                    & 1.5540e-07\\ \hline
ILA       & 6.2550e-04                    & 1.7596e-06\\ \hline
ULA       & 6.6183e-05                        & 7.7698e-07 \\ \hline
SKROCK    & 2.0804e-04                        & 1.3499e-06 \\ \hline
\end{tabular}
\captionof{table}{$\mathcal{W}_2(\pi, \Tilde{\pi})$ errors for one arbitrary pixel.}
\label{tab:wasserstein}
\end{center}
\end{table}

\begin{table}[h]
\centering
\begin{center}
\begin{tabular}{|c|c|}
\hline
Algorithm & \multicolumn{1}{|p{1.5cm}|}{\centering $\mathcal{W}_2(\pi, \Tilde{\pi})$}      \\ \hline\hline
EXACT      & 0.2490                        \\ \hline
IMLA       & 0.3991                    \\ \hline
ILA       & 3.3662                    \\ \hline
ULA       & 1.1514                        \\ \hline
SKROCK    & 2.0769                        \\ \hline
\end{tabular}
\captionof{table}{Summed $\mathcal{W}_2(\pi, \Tilde{\pi})$ errors over all pixels computed over $15000$ samples.}
\label{tab:wasserstein-full}
\end{center}
\end{table}

Furthermore,  in our setting the posterior distribution is separable into one-dimensional distributions over the pixel domain, which significantly simplifies the analysis of the estimation error. In particular, for any pixel it is possible to reliably compute the Wasserstein-$2$ error w.r.t. the true pixel's distribution by using \cite[Remark 2.30]{peyre2018}. We conducted this toy experiment by using {a $60 \times 60$ pixel region} of the \texttt{lizard} test image from the BSDS300 data set \cite{MartinFTM01}.

We quantify the bias of each method in the following two complementary ways: (i) by computing the Wasserstein-$2$ error for the marginal distribution of a randomly selected image pixel; and (ii) by computing the distribution of the scalar summary statistic $\log \pi (x)$, which defines the typical set of the posterior distribution and plays a central role in uncertainty quantification analyses \cite{pereyra-map-2017}. The Wasserstein-$2$ error $\mathcal{W}_2(\pi, \Tilde{\pi})$ was computed using $50$ repetitions of $10^{6}$ realisations for each of the algorithms under scrutiny; for completeness, we also included a fully implicit $(\theta=1)$ Langevin algorithm (ILA) in the comparison. Table \ref{tab:wasserstein} reports the mean and the standard deviation of the Wasserstein-$2$ error for each method w.r.t. the truth. To provide an indication of the baseline error stemming from the finite Monte Carlo sample size, we also include the results obtained for a perfect sampler that draws exact samples from $\pi$. We observe that the error of IMLA is almost exactly the same as the baseline error of the exact samples, confirming the previous analysis and the fact that this algorithm has a remarkably low bias. This result is confirmed in Table \ref{tab:wasserstein-full}, showing the summed Wasserstein errors over all pixels and IMLA having the least error compared to all other methods.

Moreover, Figure \ref{fig:logpigmm} depicts the histogram of the  scalar statistic $T(x)=\log\pi(x)$, as computed from {$5\times10^{5}$} iterations of each algorithm. For IMLA ($\theta=1/2$), SKROCK, and ILA ($\theta=1$) we have used the optimal time step $\delta_{*}$ associated with IMLA, while for ULA we used $\delta=1/L$ as recommended in \cite{durmus2017}. We observe that the histogram obtained by using IMLA is in close agreement with the truth, while {all the other methods exhibit some significant bias, with ILA exhibiting the largest amount of bias, and SKROCK and ULA having similar moderate biases. 

\subsection{Illustrative experiment: One dimensional distributions}\label{sec:1d-distributions}

We continue our illustrative experiments by comparing {the algorithms MYULA ($\theta=0$), IMLA ($\theta=1/2$), and ILA ($\theta=1$)} used to sample from \review{four} representative one-dimensional distributions, namely
\begin{itemize}
\item Laplace distribution: $\pi(x)\propto e^{-|x|}$, \item Uniform distribution: $\pi(x)=e^{-\iota_{[0,1]}(x)}$, 
\item Light-tailed distribution: $\pi(x)\propto e^{-x^4}$,
\item \review{Standard Cauchy distribution (heavy-tailed): $\pi(x) \propto (1+x^2)^{-1}$}
\end{itemize} 
We use $\iota_{[0,1]}$ to denote the indicator function which takes the value $0$ on $[0,1]$ and $+\infty$ otherwise (such that $\pi(x)=1$ in $[0,1]$ and $0$ outside that interval). We have chosen these distributions because they exhibit different tail behaviors and are representative situations where ULA can not be applied directly without smoothing the distribution (the distributions are not Lipschitz-differentiable). In contrast, using the formulation \eqref{eq:IMLAprox} of IMLA we can apply IMLA directly for these distributions without smoothing. Moreover, since the proximal operator has an analytic expression in each of these cases, the implementation of IMLA and ILA is straightforward and has the same cost as MYULA. \review{Special attention is needed when computing the proximal operator for the Cauchy distribution, as it is not convex (see Appendix \ref{app:cauchy_prox}).}

\begin{table}[t]
\centering
\begin{tabular}{|c|c|c|c|c|}
\hline
{Distribution} & {$\text{SD IMLA}$} & {$\text{SD ILA}$} & {$\text{SD MYULA}$} &{ $\text{SD} {\text{ EXACT}}$ } \\ \hline\hline
Laplace      & 1.4046  & 1.4005        & 1.4356           & 1.4142                      \\ \hline
Uniform      & 0.2923  & 0.2936        & 0.2949           & 0.2887                    \\ \hline
$\exp(-x^4)$ & 0.5964  & 0.5777        & 0.6590           & 0.5813                      \\ \hline
\end{tabular}
\caption{Estimated vs. the exact standard deviations (SD) for each method (IMLA, ILA and MYULA) for the experiment shown in Figure \ref{fig:one_d_dist}. \review{Results for the Cauchy distribution are not included, as moments are undefined for this distribution.}}
\label{tab:1d-dist}
\end{table}

Figure \ref{fig:one_d_dist} depicts the histograms associated with the Markov chains generated by using MYULA, IMLA, and ILA for the \review{four} test distributions considered. In the case of the Laplace, the light-tailed \review{and the Cauchy distributions,} we have used $\delta=5\times10^{-2}$, while for the uniform distribution, we have used $\delta=10^{-4}$. We have run the algorithms for $15 \times 10^6$ iterations and set the smoothing parameter of MYULA to $\lambda=\delta$, as we found that this gives the best approximation of $\pi$ in terms of the convergence speed and asymptotic bias trade-off. 

As can observed in Figure \ref{fig:one_d_dist}, IMLA and ILA achieve a visibly smaller bias than MYULA for all the distributions, which is in agreement with our analysis in Section \ref{subsec:Gaussian}. \review{For the heavy-tailed distribution, the continuous time OLSDE is not exponentially ergodic as established in \cite{roberts1996}, and therefore, discrete time approximations do not generally converge geometrically fast. We observe that IMLA is more robust to the heavy tails and recovers the shape of the true distribution best, whereas MYULA puts more mass in the tails of the distribution and ILA in the center of the distribution.} For a quantitative comparison, we report the standard deviation (SD) estimated by each method in Table \ref{tab:1d-dist} \review{except for the Cauchy distribution, for which moments do not exist}. It can be observed that MYULA overestimates the SD in all cases, IMLA and ILA are significantly closer to the true value, and IMLA is the most accurate for the Laplace and the Uniform distributions. 

\begin{figure}[!htbp]
\centering
\begin{subfigure}{.245\textwidth}
    \centering
    \includegraphics[width=\linewidth]{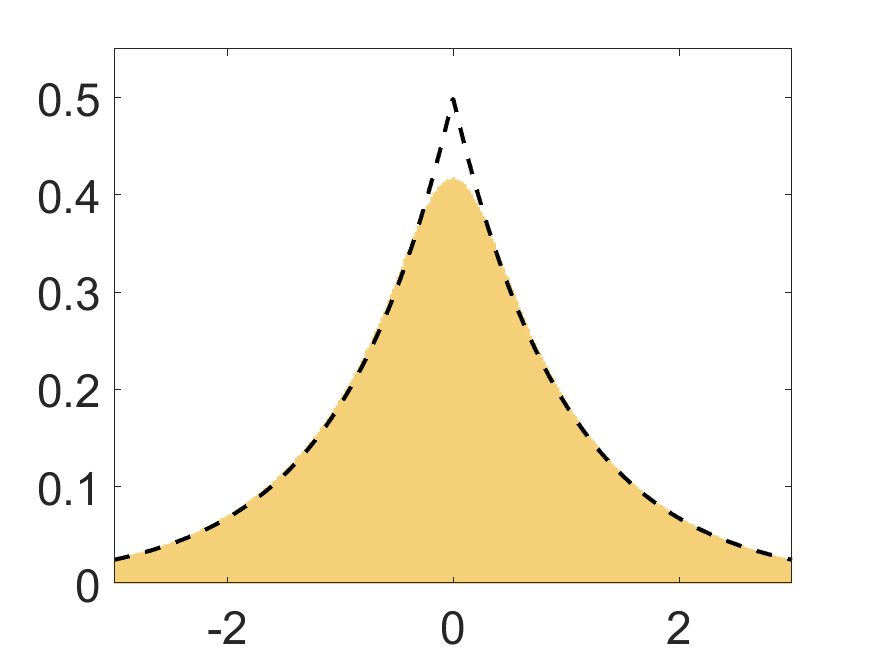}
\end{subfigure}%
\begin{subfigure}{.245\textwidth}
    \centering
    \includegraphics[width=\linewidth]{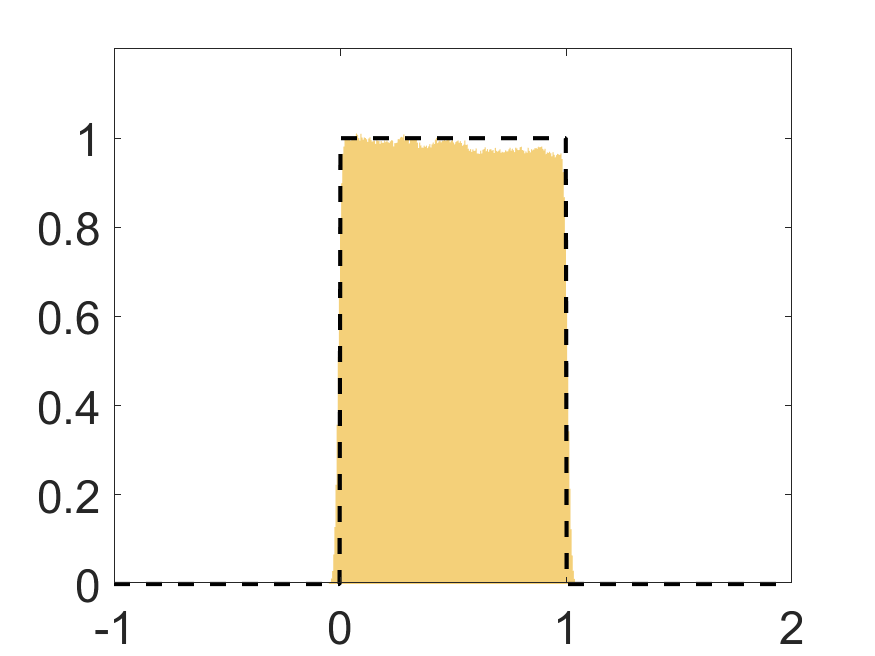}
\end{subfigure}%
\begin{subfigure}{.245\textwidth}
    \centering
    \includegraphics[width=\linewidth]{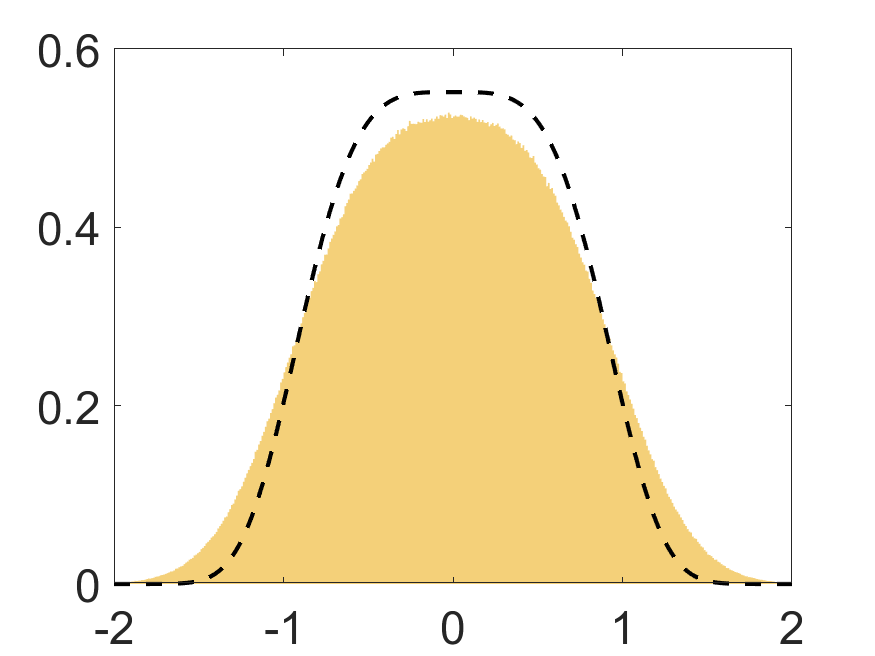}
\end{subfigure}
\begin{subfigure}{.245\textwidth}
    \centering
    \includegraphics[width=\linewidth]{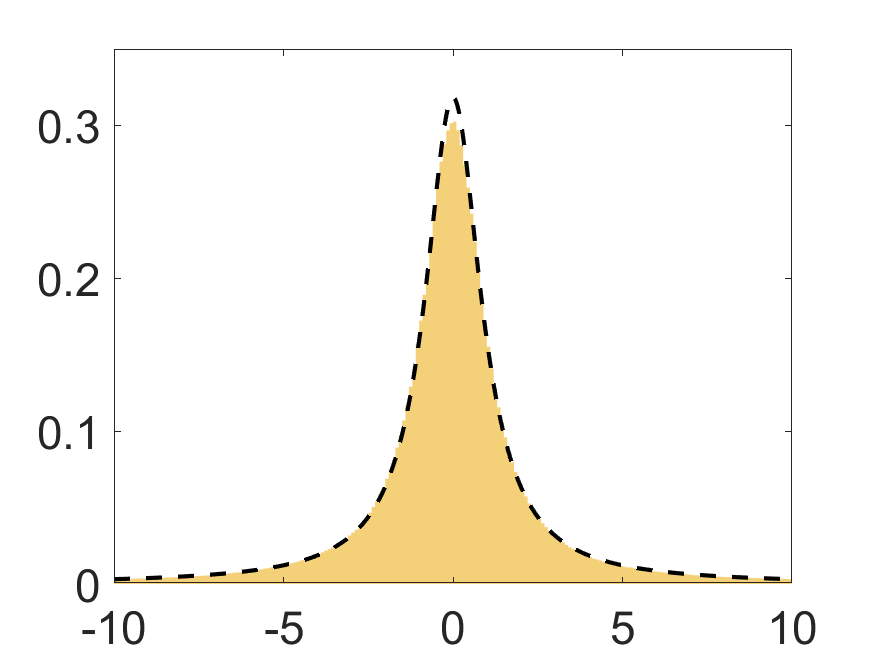}
\end{subfigure}

\begin{subfigure}{.245\textwidth}
    \centering
    \includegraphics[width=\linewidth]{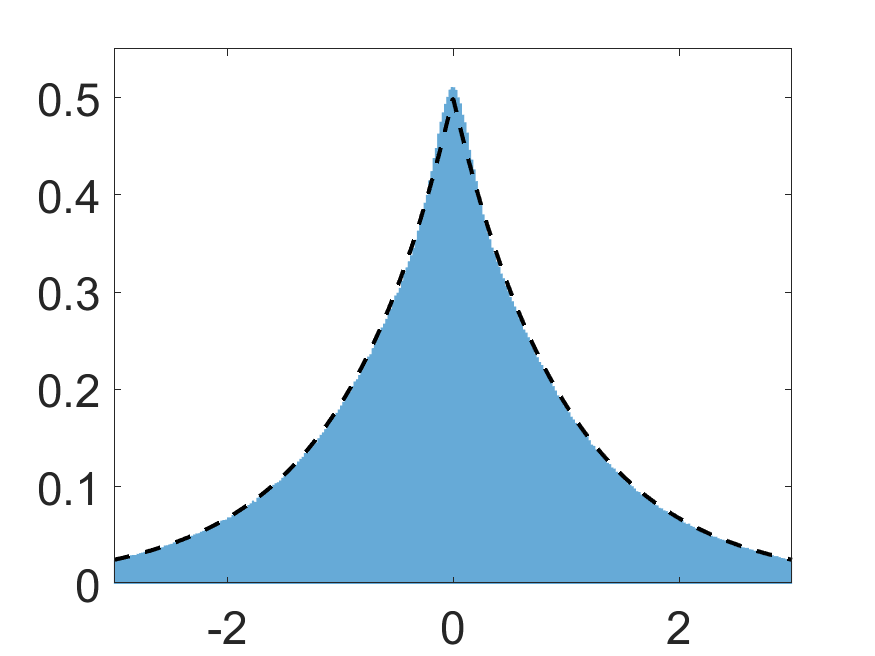}
\end{subfigure}%
\begin{subfigure}{.245\textwidth}
    \centering
    \includegraphics[width=\linewidth]{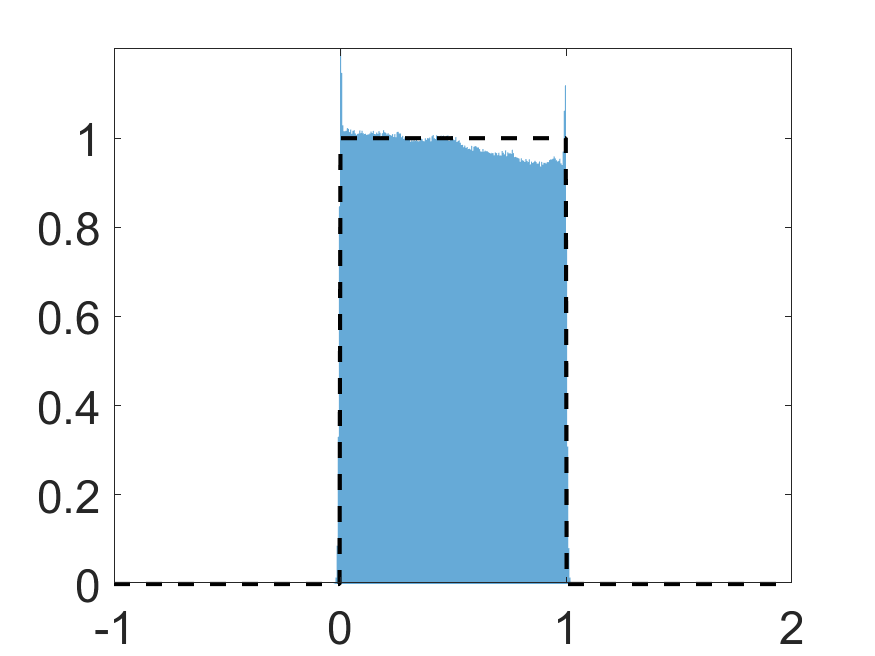}
\end{subfigure}%
\begin{subfigure}{.245\textwidth}
    \centering
    \includegraphics[width=\linewidth]{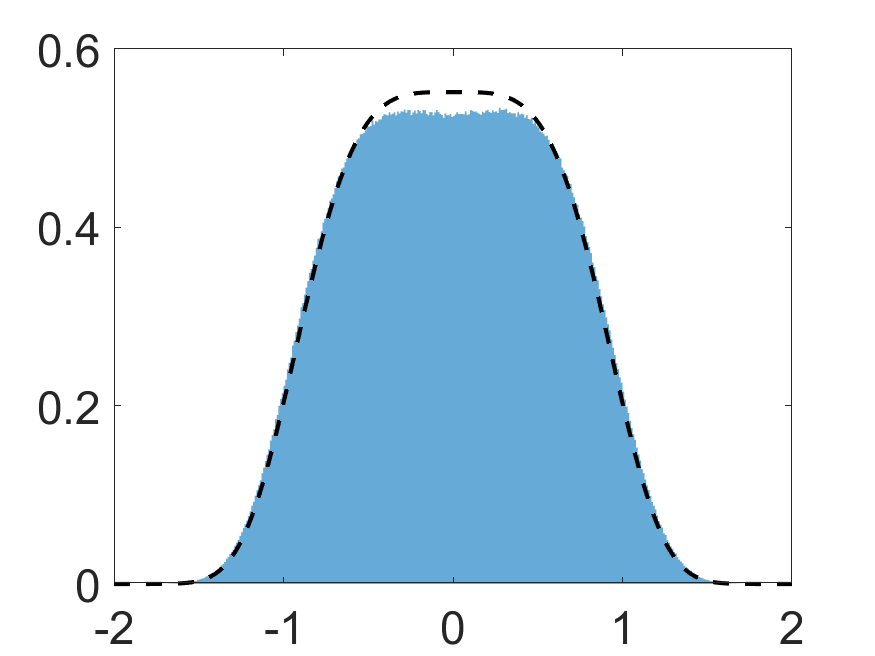}
\end{subfigure}%
\begin{subfigure}{.245\textwidth}
    \centering
    \includegraphics[width=\linewidth]{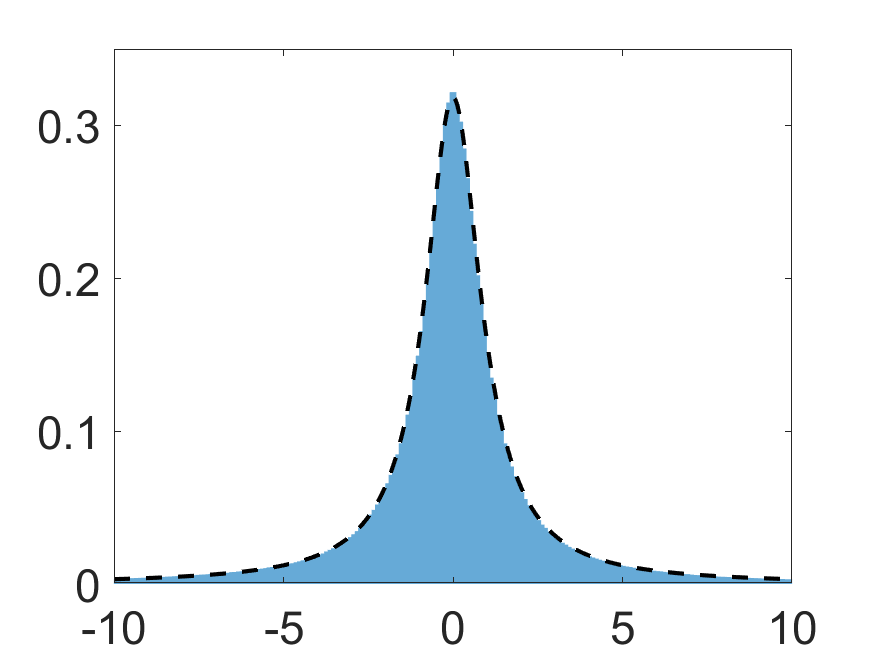}
\end{subfigure}%

\begin{subfigure}{.245\textwidth}
    \centering
    \includegraphics[width=\linewidth]{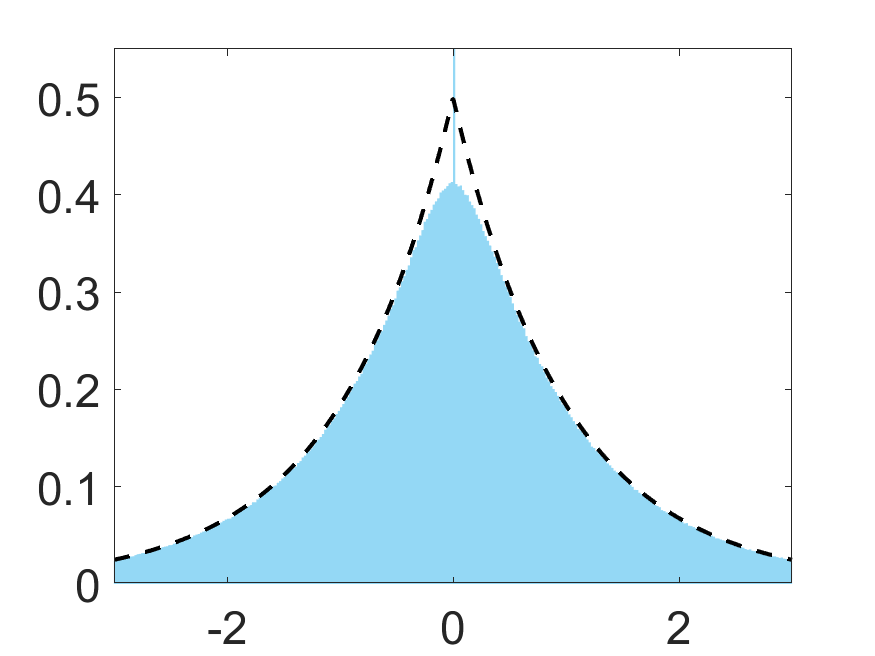}
\end{subfigure}%
\begin{subfigure}{.245\textwidth}
    \centering
    \includegraphics[width=\linewidth]{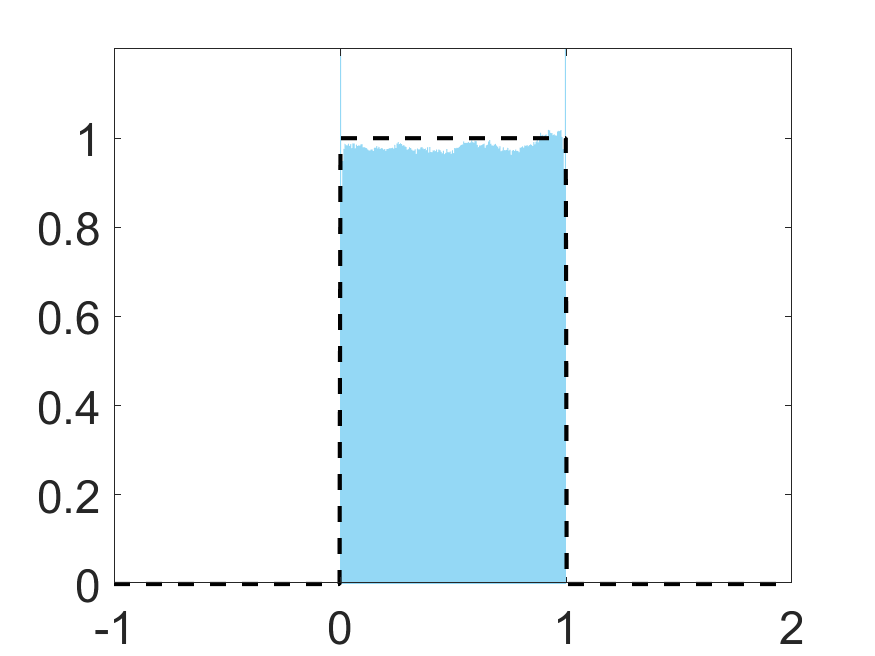}
\end{subfigure}%
\begin{subfigure}{.245\textwidth}
    \centering
    \includegraphics[width=\linewidth]{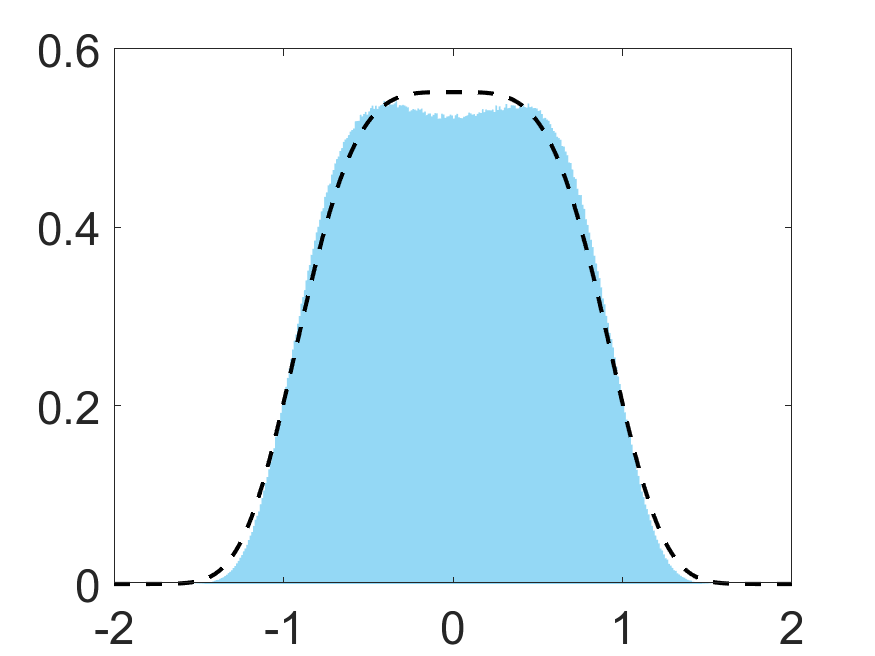}
\end{subfigure}%
\begin{subfigure}{.245\textwidth}
    \centering
    \includegraphics[width=\linewidth]{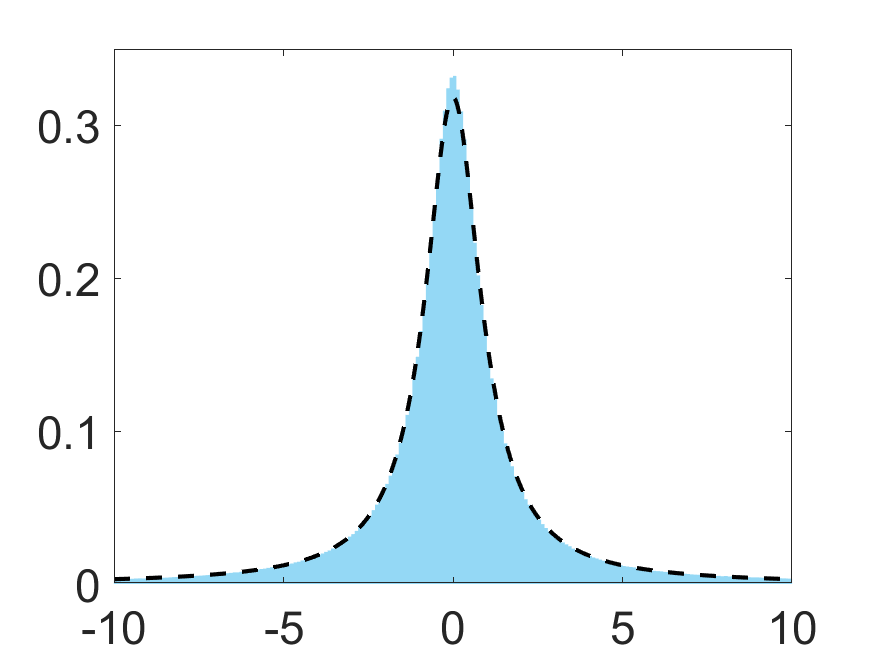}
\end{subfigure}%

\begin{subfigure}{.245\textwidth}
    \centering
    \includegraphics[width=\linewidth]{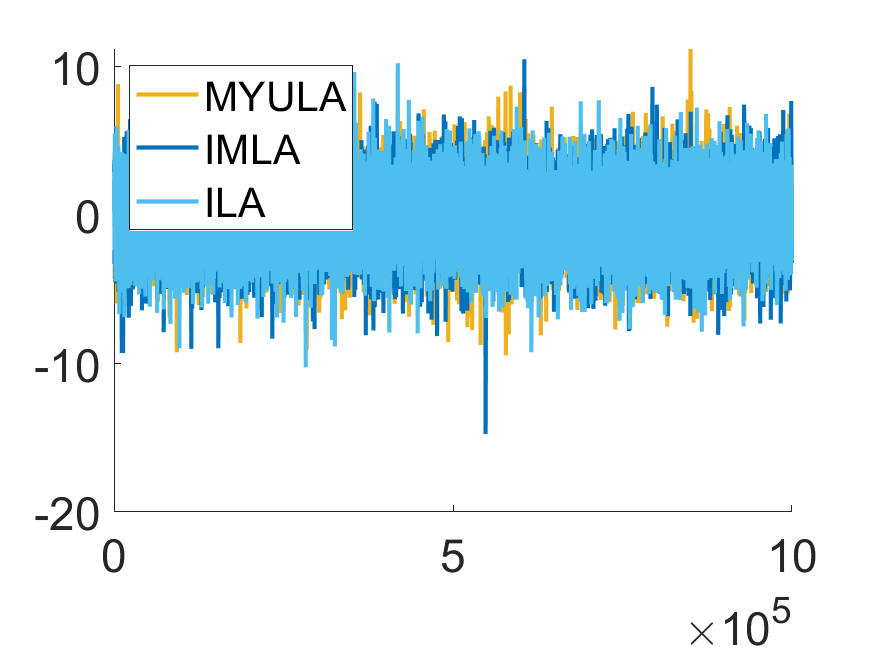}
    \caption{Laplace}
        \label{fig:one_d_laplace}
\end{subfigure}%
\begin{subfigure}{.245\textwidth}
    \centering
    \includegraphics[width=\linewidth]{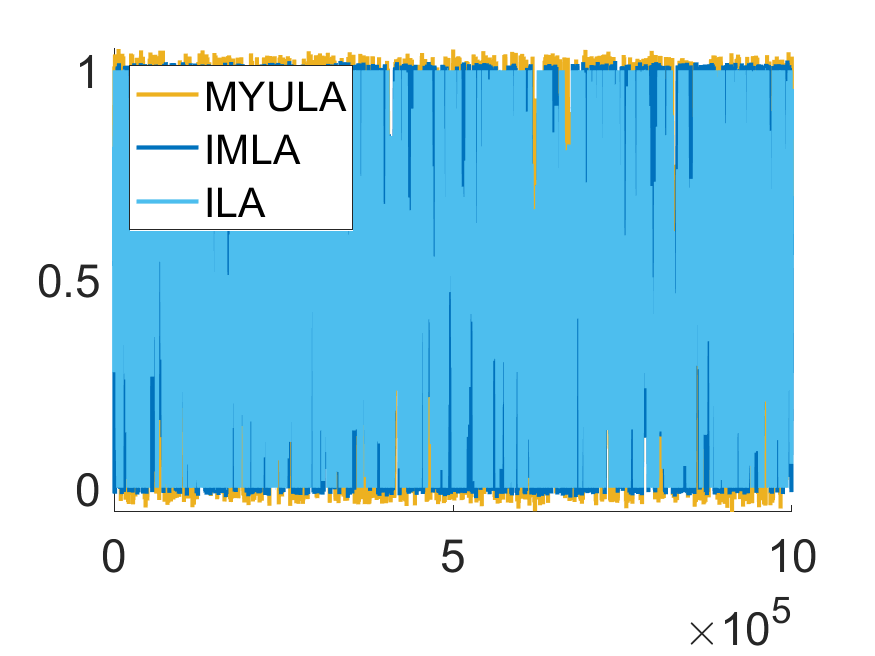}
    \caption{Uniform}
     \label{fig:one_d_uniform}
\end{subfigure}%
\begin{subfigure}{.245\textwidth}
    \centering
    \includegraphics[width=\linewidth]{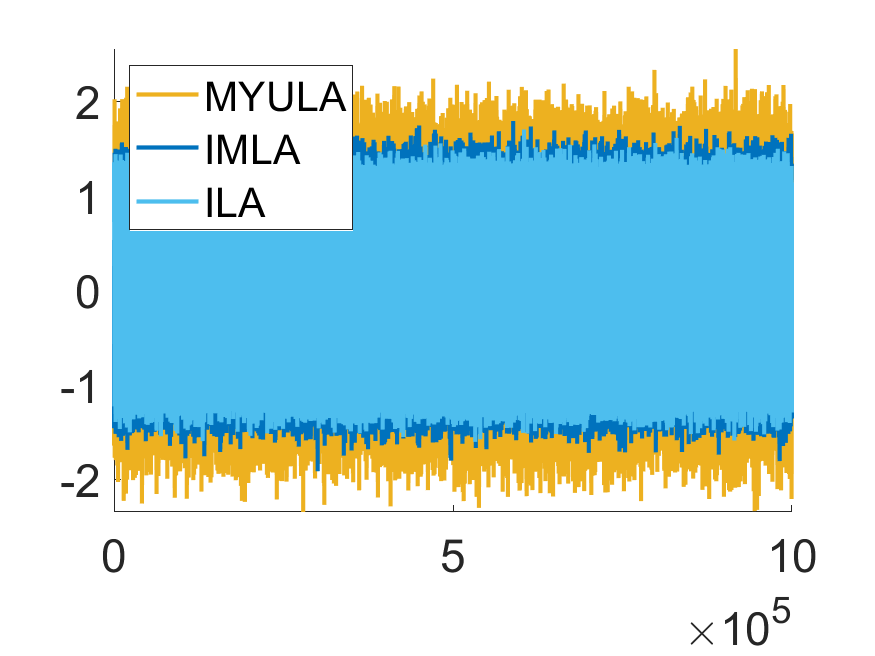}
    \caption{Light-tailed}
     \label{fig:one_d_x4}
\end{subfigure}%
\begin{subfigure}{.245\textwidth}
    \centering
    \includegraphics[width=\linewidth]{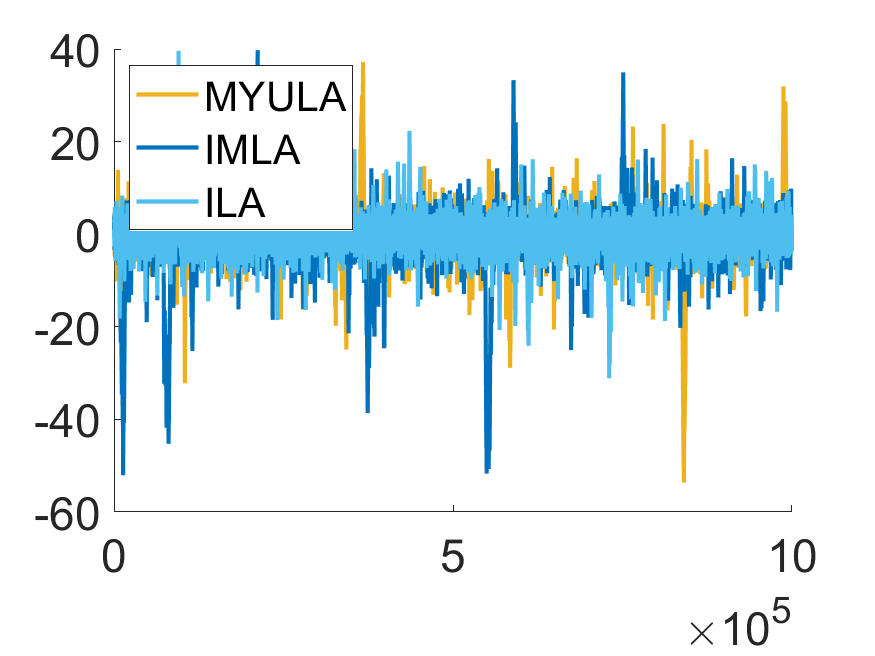}
    \caption{Heavy-tailed}
     \label{fig:one_d_cauchy}
\end{subfigure}%
\caption{\review{Histograms  of  the  Laplace  distribution (first column), the uniform distribution (second column), a light-tailed distribution with $\pi(x)\propto e^{-x^4}$ (third column), and a heavy-tailed Cauchy distribution (last column) computed with $15 \times 10^6$ iterations generated by MYULA (first row) and IMLA (second row) and ILA with $\theta=1$ (third row), for $\delta=0.05$ (left, right) and $\delta=10^{-4}$ (middle). The dashed black curve represents the probability density function of the respective distribution. Last row: Trace plots of the first $10^6$ iterations to show the behaviour of the chains.}}
\label{fig:one_d_dist}
\end{figure}

\begin{algorithm}[h!]
\caption{IMLA}\label{alg:IMLA}
\begin{algorithmic}
\Require $N \geq 0$, $\delta>0$ and $X_0\in \R^d$.
\For{n=0 : N-1}\\
\textbf{Draw} $$\xi_n\sim \mathcal{N}(0,I_d)$$
\textbf{Set}
$$X_{n+1}\gets \arg\min_{x\in \R^d} 2\pot\left(\frac{1}{2} x+\frac{1}{2} X_n\right) +\frac{1}{2\delta}\lVert x-X_n-\sqrt{2\delta}\xi_n\rVert^2$$
\EndFor
\end{algorithmic}
\end{algorithm}

\subsection{Algorithmic recommendations: IMLA}\label{sec:recommend}
{The analysis in Sections \ref{subsec:Gaussian} indicates that to converge as quickly as possible, one should choose $\theta=1$ with arbitrary large step size $\delta$. However, the bias in this case would be very large as seen in Figures \ref{fig:conv-plot}(\subref*{figure1}) and \ref{fig:conv-plot}(\subref*{figure2}). In contrast, using $\theta=1/2$ in the Gaussian case is bias-free. Moreover, it can be observed in Figure \ref{fig:conv-plot}(\subref*{figure2}) that it also provides a better non-asymptotic convergence behavior than $\theta=1$. The illustrative numerical experiments in Section \ref{sec:gmm-model} reveal that this behavior persists outside the Gaussian regime.
Furthermore, the numerical experiments in Section \ref{sec:1d-distributions}, even though not  covered by our theoretical analysis, reveal no substantial benefit of using $\theta=1$ over $\theta=1/2$ for the values of $\delta$ considered. Hence, from now on we concentrate on the choice $\theta=1/2$, which we summarise in Algorithm \ref{alg:IMLA}. In the case when $U$ is strongly convex, we set $\delta=\delta_*=2/\sqrt{Lm}$ which gives the optimal contraction rate. Algorithm \ref{alg:IMLA} does not incorporate hard constraints, but it can be modified if needed, as will be discussed in Section \ref{sec:poiss-deconv}.

\section{Imaging experiments}\label{sec:Numerical}

In this section, we demonstrate the proposed Bayesian methodology on a range of
experiments related to image deconvolution. We consider problems with an underlying convex geometry (i.e., with a posterior distribution that is
log-concave). In Section \ref{sec:nn-deconv} we consider a model with a data-driven convex prior and a Gaussian noise characteristic, and in Section \ref{sec:poiss-deconv} we consider a model with an assumption-driven convex prior and a Poisson noise characteristic.

\subsection{Image deconvolution using a CRR-NN prior}\label{sec:nn-deconv}
We consider a non-blind motion deblurring problem where we seek to recover
$x \in \R^d$ from a blurry and noisy observation $y=Ax+\epsilon$, where the linear operator $A$ represents a known motion blur kernel and $\epsilon \sim \mathcal{N}(0,\sigma^2 I_{d})$. We conduct our experiments with a range of motion blurs from \cite{Levin09}; in all cases, $A^\top A$ is highly ill-conditioned but full rank. To regularise the estimation problem, we leverage a state-of-the-art data-driven prior \cite{goujon2022crrnn}, which is by construction log-concave. Combining this prior with the likelihood gives us a posterior distribution with the following density
\begin{equation} \label{eq:post1}
 \pi(x) \propto \exp\left(-\frac{\norm{Ax-y}^{2}}{2\sigma^{2}} - \frac{\alpha}{\mu} R_{\Theta}(\mu x) \right)\,,   
\end{equation}
where $R_{\Theta}(\cdot)$ is the convex-ridge regularizer neural network (CRR-NN) \cite{goujon2022crrnn}, with $\mu>0$ and $\alpha>0$ scale parameters. In our numerical experiments, we use three different images from the BSD300 data set \cite{MartinFTM01} (\texttt{castle, lizard}, and \texttt{person}) depicted in Figure \ref{fig:obs-motion}. To degrade each image, we use a different blur kernel\footnote{We use the kernels from \cite{Levin09} with indexes $\{2,3,4\}$.} and generate the observation $y$ by adding Gaussian noise with variance $\sigma^2$ chosen to achieve a blurred signal-to-noise  ratio of 30dB. The corresponding blurred and noisy images are shown in Figure \ref{fig:obs-motion}, while the values of the Lipschitz and strong convexity constants $(L, m)$ for each blur kernel are reported in Table \ref{tab:params-motion}.
\begin{figure}[htbp]
\centering
\begin{subfigure}{.45\textwidth}
    \centering
    \includegraphics[trim=50 0 50 30, clip=true, width=\linewidth]{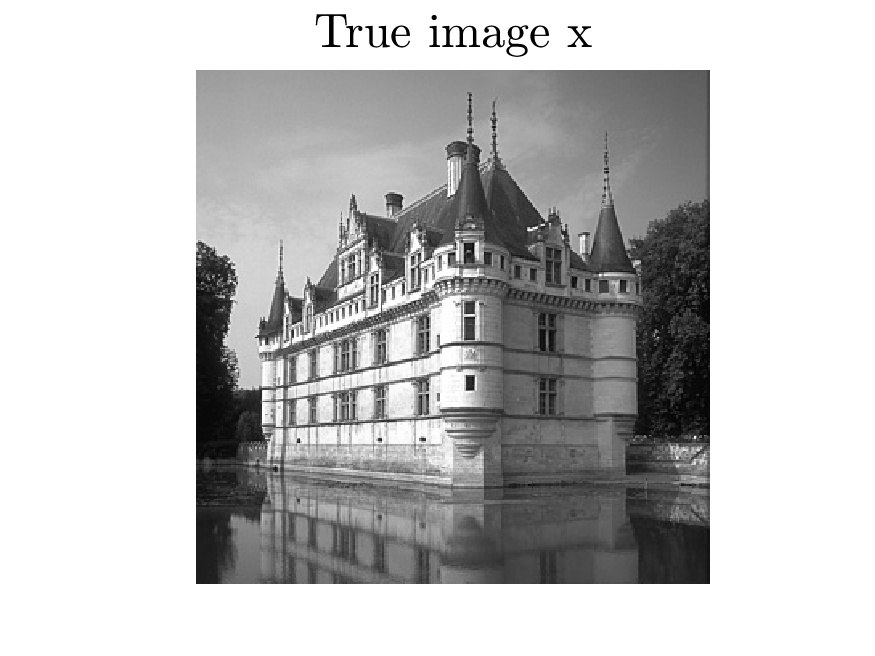}
\end{subfigure}%
\begin{subfigure}{.45\textwidth}
    \centering
    \includegraphics[trim=50 0 50 30, clip=true,width=\linewidth]
    {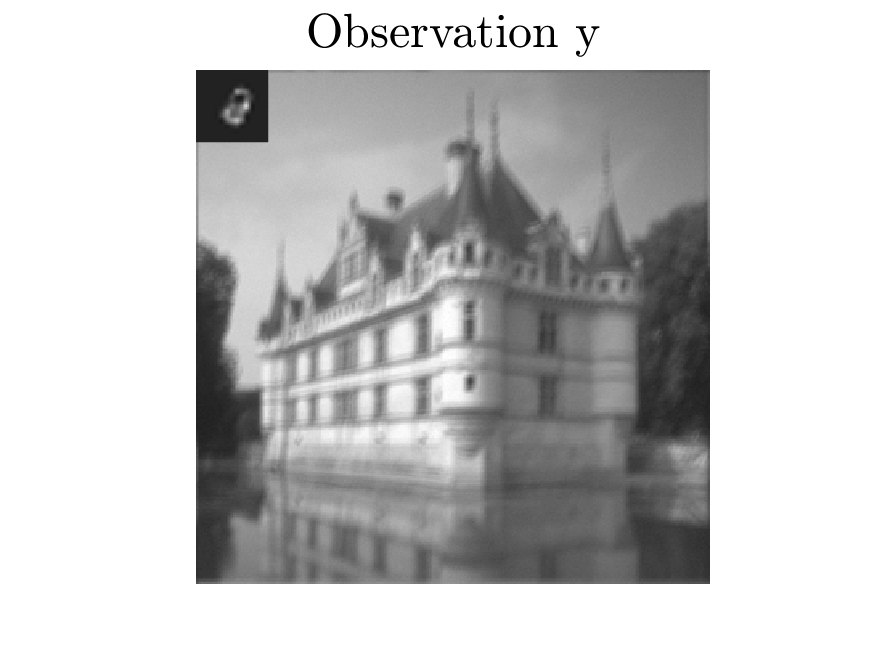}
\end{subfigure}

\begin{subfigure}{.45\textwidth}
    \centering
    \includegraphics[trim=50 0 50 30, clip=true,width=\linewidth]{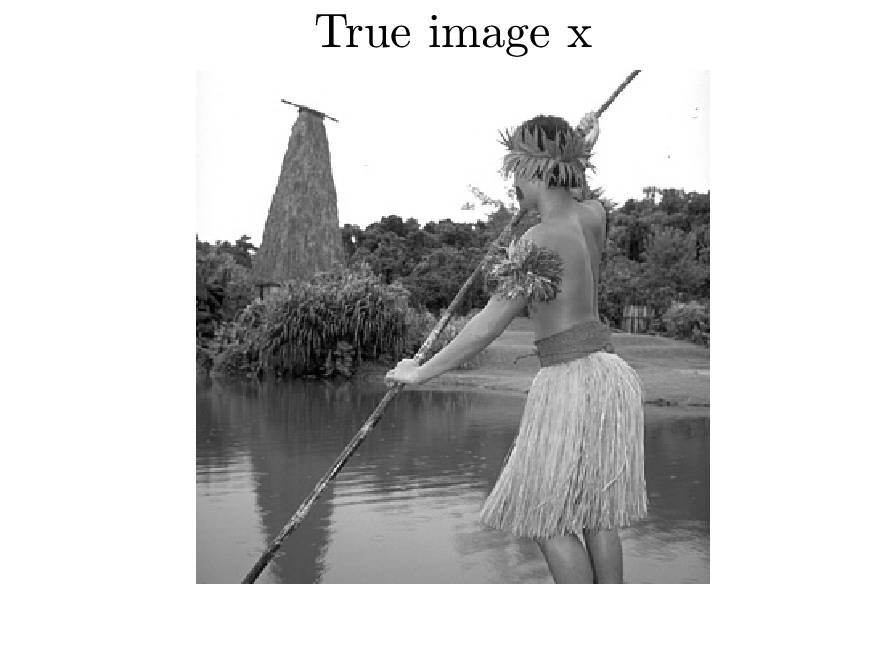}
\end{subfigure}%
\begin{subfigure}{.45\textwidth}
    \centering
    \includegraphics[trim=50 0 50 30, clip=true, width=\linewidth]{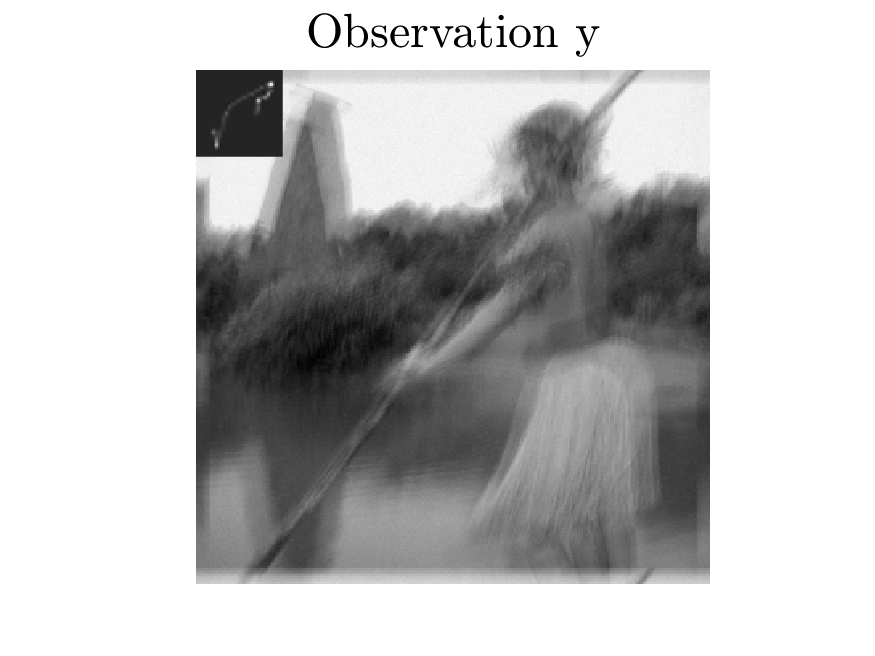}
\end{subfigure}%

\begin{subfigure}{.45\textwidth}
    \centering
    \includegraphics[trim=50 0 50 30, clip=true,width=\linewidth]
    {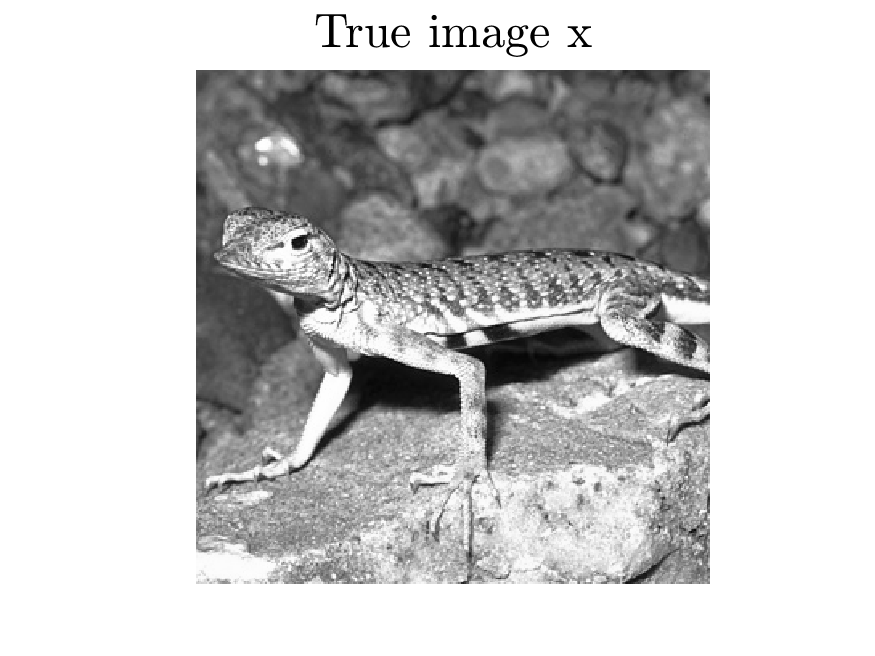}
\end{subfigure}%
\begin{subfigure}{.45\textwidth}
    \centering
    \includegraphics[trim=50 0 50 30, clip=true,width=\linewidth]{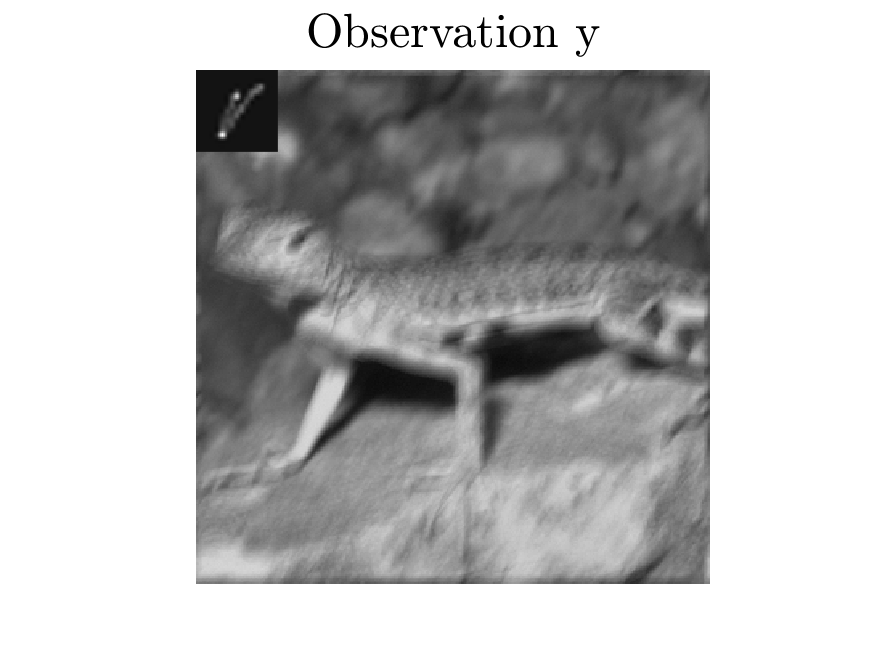}
\end{subfigure}%
\caption{Motion deconvolution experiment: Ground truth images $x$ (left) of size $320\times320$ with observations $y$ (right) for \texttt{castle, person} and \texttt{lizard} images (top to bottom) and their respective blur kernels in the top left corner of $y$.}
\label{fig:obs-motion}
\end{figure}
In addition, as recommended in \cite{goujon2022crrnn}, the parameters $\alpha$ and $\mu$ of the neural network regularizer are adjusted for each experiment by using a grid search method\footnote{\url{https://github.com/axgoujon/convex_ridge_regularizers} }, which optimizes the parameters to yield the MAP solution with optimal peak-signal-to-noise ratio (PSNR).

\begin{table}[t]
\centering
\begin{tabular}{|c|c|c|c|c|c|c|c|}
\hline
Experiment & $L$     & $m$ & $h_{\mathrm{IMLA}}$ & stages & $h_{\mathrm{ULA}}$ & $\alpha $ &$\mu$  \\ \hline\hline 
\texttt{castle}     & 43521 & 17.39 & 0.002     & 8      & 2.3e-05  & 2828.43 & 8.00      \\ \hline
\texttt{lizard}    & 50885 & 17.25 & 0.002     & 8      & 2.0e-05 &  3363.59 & 8.00   \\ \hline
\texttt{person}    & 41896 & 2.41 & 0.006     & 12     & 2.4e-05  & 2828.43 & 11.31   \\ \hline
\end{tabular}
\caption{Algorithm parameters for motion deconvolution experiment.}
\label{tab:params-motion}
\end{table}

For each image $y$, we use IMLA, SKROCK and ULA to draw Monte Carlo samples from the posterior distribution $\pi$. {In the case of ULA, the step size is set to half the stability barrier, $h_{\mathrm{ULA}}=1/L$ as recommended in \cite{durmus2017}, which is of the order of $10^{-5}$. In contrast, for IMLA we use the optimal step size $h_{\mathrm{IMLA}}=2/\sqrt{Lm}$ as suggested by the analysis in Section \ref{subsec:Gaussian} (note that $h_{\mathrm{IMLA}}$ is two orders of magnitude larger than $h_{\mathrm{ULA}}$, see Table \ref{tab:params-motion}).} We use the same step size for SKROCK, and automatically derive the appropriate number of internal stages, $s$, for each experiment (see Table \ref{tab:params-motion}). To make the comparison fair between SKROCK and ULA we use $4\times 10^{4}$ iterations for SKROCK and $4s\times 10^{4}$ iterations for ULA. For IMLA, we use the same number of iterations as for SKROCK ($4\times 10^{4}$), as we are initially interested in comparing the performance of the two methods without taking into account the average computational cost per iteration (since this is highly implementation dependent). We defer the analysis of the computational cost to Section \ref{sec:poiss-deconv-cost}. 

To demonstrate the effectiveness of our method in an uncertainty quantification context, we summarise $\pi$ by computing two different Bayesian estimators that summarise $\pi$ optimally w.r.t. to complementary loss functions and by computing uncertainty visualisation plots based on (marginal) second-order moments of $\pi$. More precisely, we compute: (i) the minimum mean square estimator (MMSE) solution given by the posterior mean, (ii) the MAP solution by solving \eqref{eq:post1} as an optimisation problem\footnote{The MAP solution is computed by using the gradient method \cite{goujon2022crrnn}.} (see \cite{pereyra-map} for details about the optimality of these estimators), (iii) the marginal standard deviations of $x$ at the pixel-wise scale, and (iv) the joint marginal standard deviations for groups of pixels of size $2\times2$, $4\times4$, $8\times8$, and $16\times16$ pixels, which display the uncertainty about structures of different size and in different regions of $x$. The MMSE and MAP solutions are depicted in Figure \ref{fig:mean-motion-comp}, whereas the uncertainty visualisation plots are displayed in Figure \ref{fig:std-motion-comp} and Figure \ref{fig:castle-uq-scales}. For comparison, we also include the results obtained with a Bayesian model using a non-convex neural network prior \cite{RyuLWCWY19} used within the PnP-ULA \cite{Laumont2021} method to perform computations.  


\begin{figure}[p]
\centering
\begin{minipage}{.32\textwidth}
    \centering
    \includegraphics[trim=30 0 30 0, clip=true,width=\linewidth]{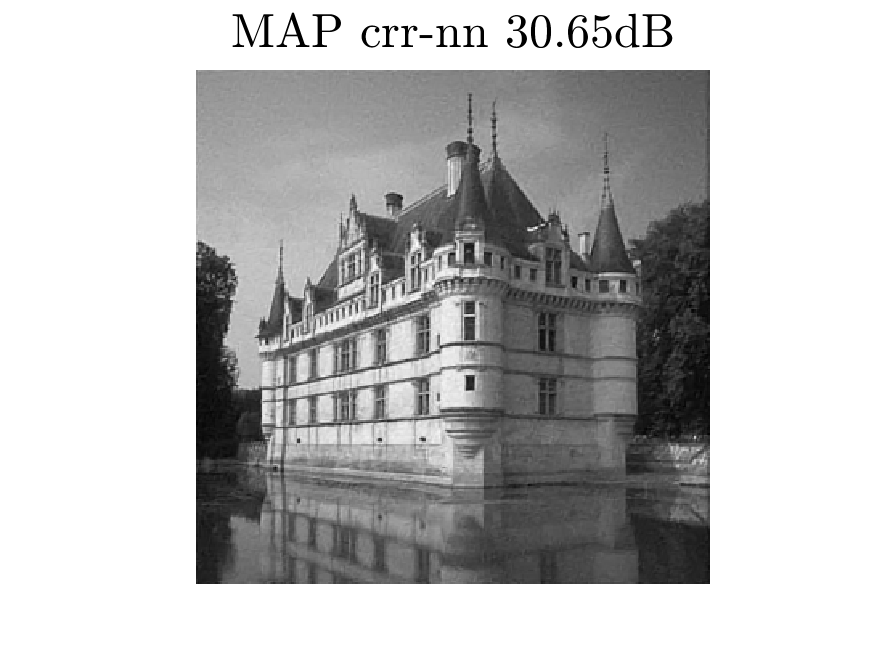}
\end{minipage}%
\begin{minipage}{.32\textwidth}
    \centering
    \includegraphics[trim=30 0 30 0, clip=true,width=\linewidth]{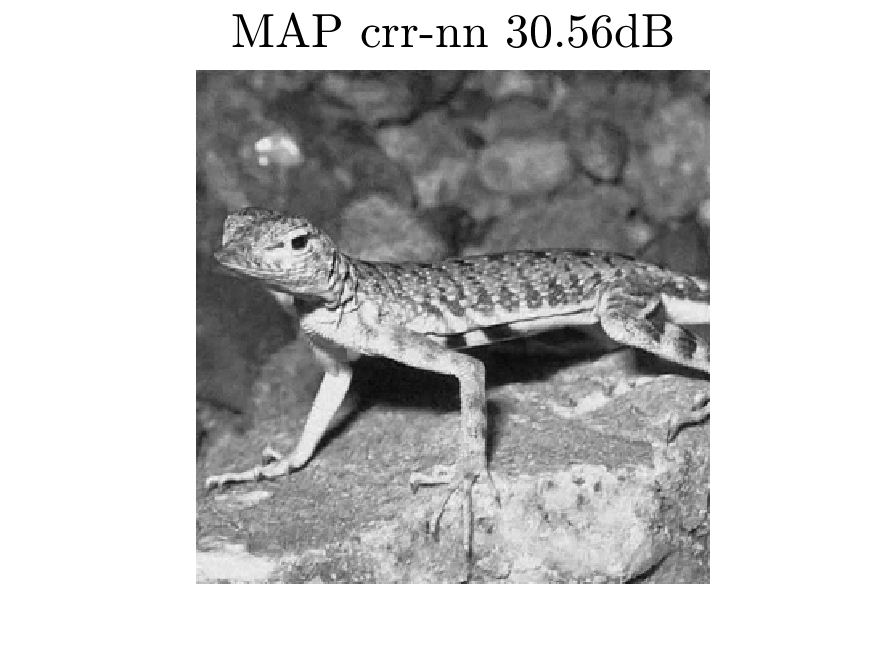}
\end{minipage}%
\begin{minipage}{.32\textwidth}
    \centering
    \includegraphics[trim=30 0 30 0, clip=true,width=\linewidth]{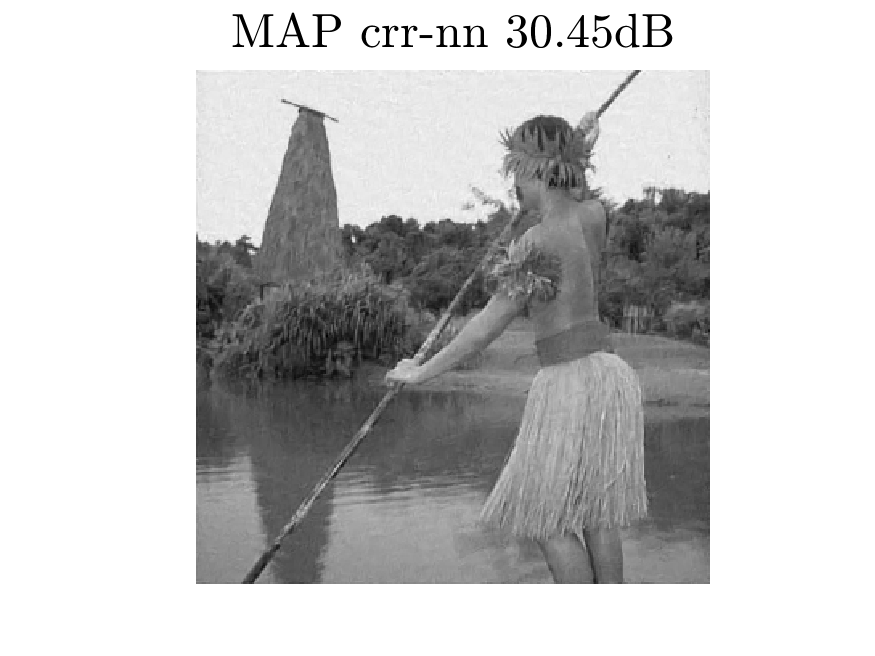}
\end{minipage}%

\begin{minipage}{.32\textwidth}
    \centering
    \includegraphics[trim=30 0 30 0, clip=true,width=\linewidth]{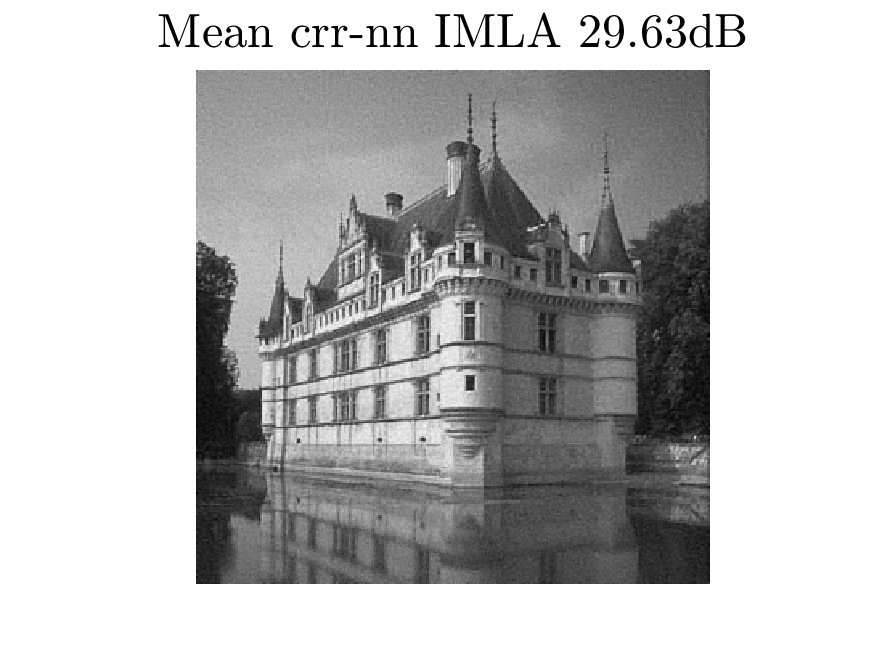}
\end{minipage}%
\begin{minipage}{.32\textwidth}
    \centering
    \includegraphics[trim=30 0 30 0, clip=true,width=\linewidth]{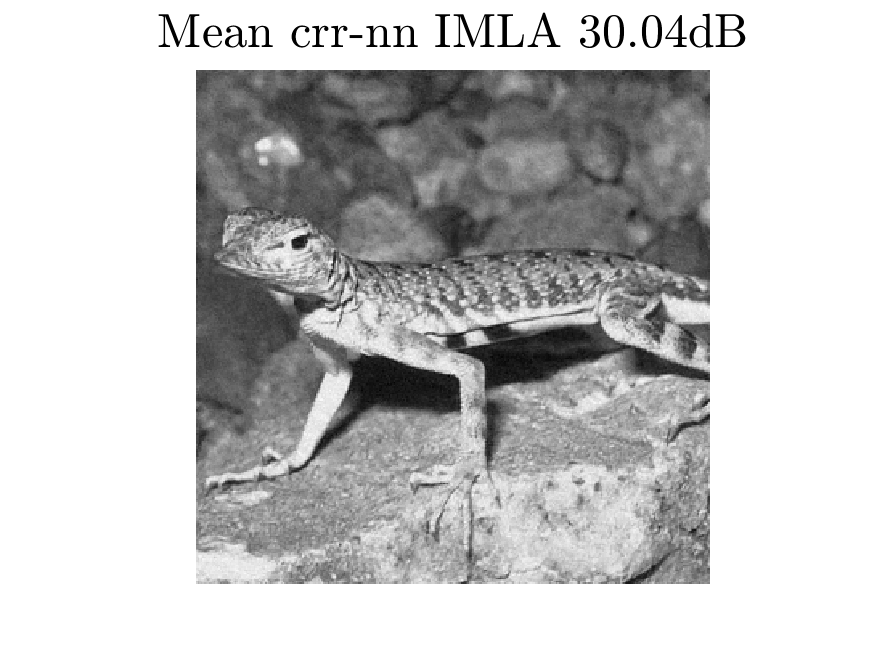}
\end{minipage}%
\begin{minipage}{.32\textwidth}
    \centering
    \includegraphics[trim=30 0 30 0, clip=true,width=\linewidth]{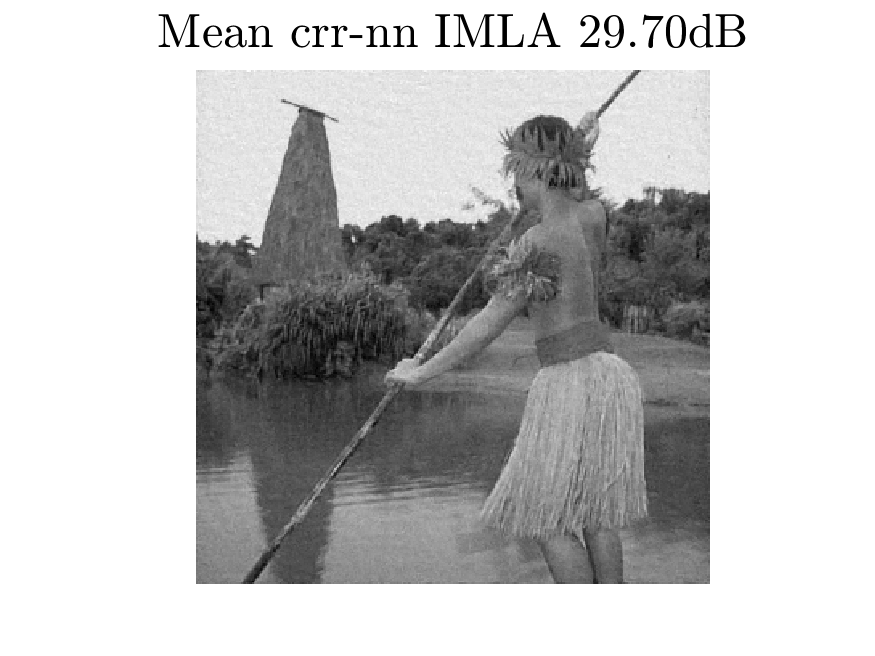}
\end{minipage}%

\begin{minipage}{.32\textwidth}
    \centering
    \includegraphics[trim=30 0 30 0, clip=true,width=\linewidth]{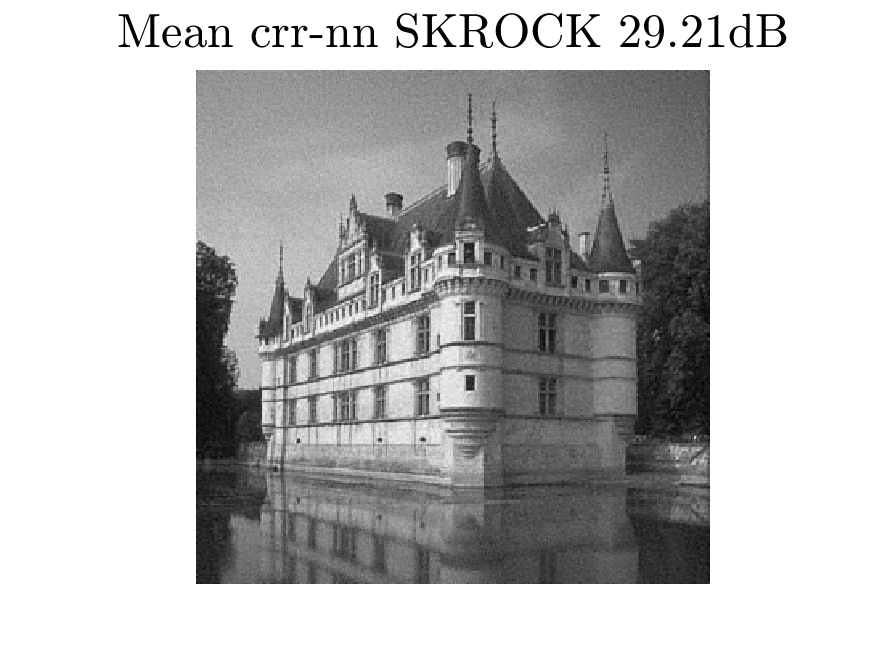}
\end{minipage}%
\begin{minipage}{.32\textwidth}
    \centering
    \includegraphics[trim=30 0 30 0, clip=true,width=\linewidth]{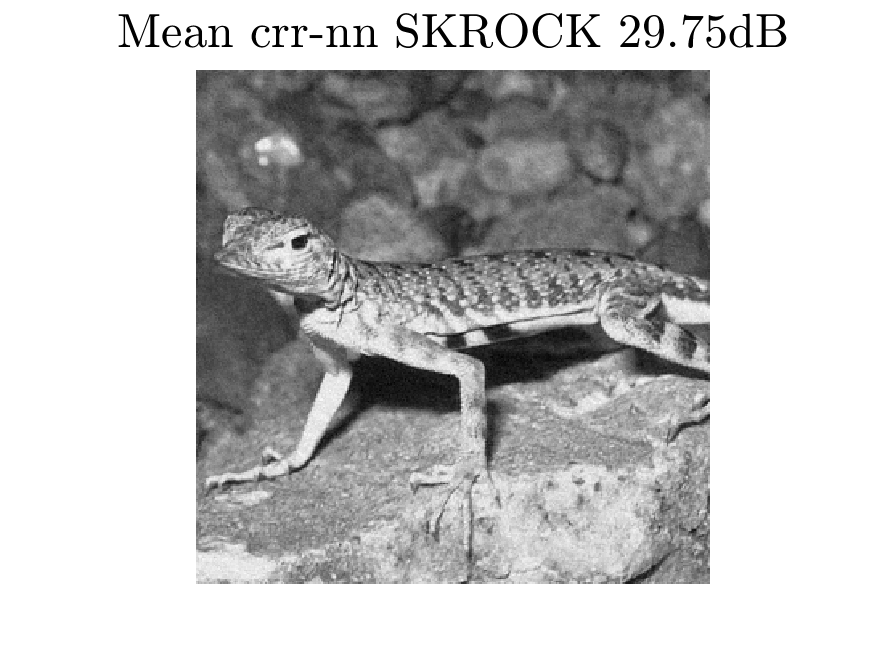}
\end{minipage}%
\begin{minipage}{.32\textwidth}
    \centering
    \includegraphics[trim=30 0 30 0, clip=true,width=\linewidth]{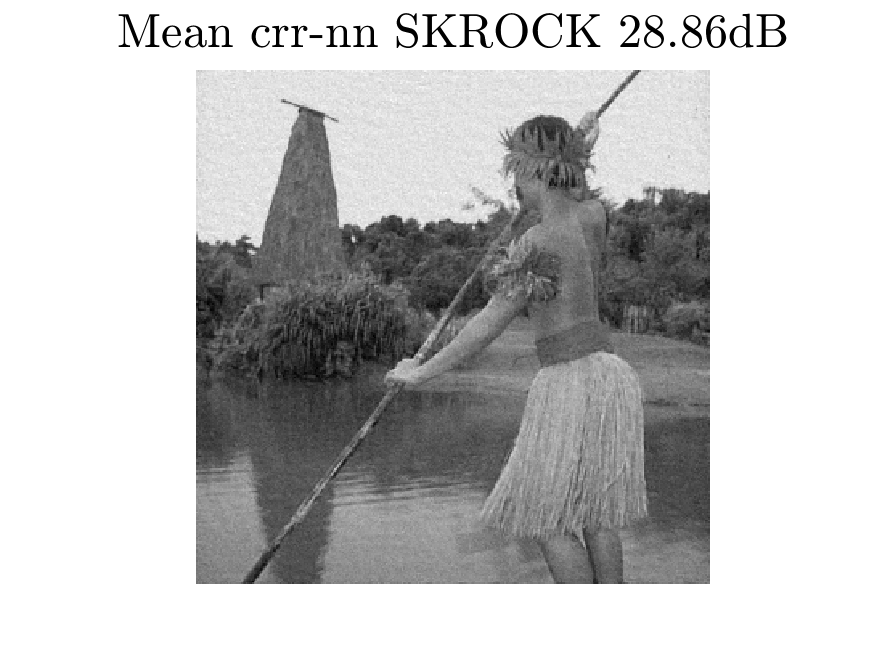}
\end{minipage}%

\begin{minipage}{.32\textwidth}
    \centering
    \includegraphics[trim=30 0 30 0, clip=true,width=\linewidth]{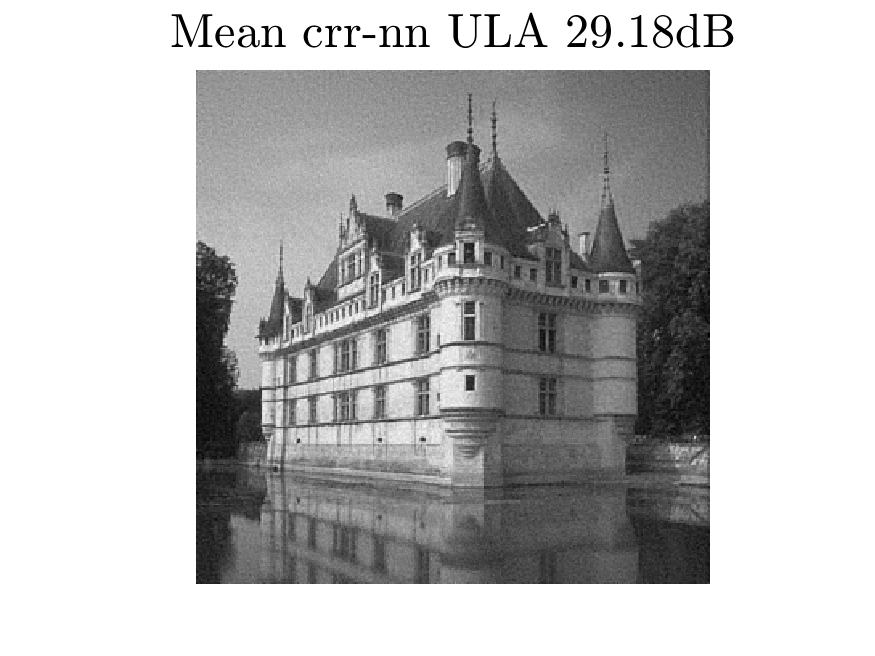}
\end{minipage}%
\begin{minipage}{.32\textwidth}
    \centering
    \includegraphics[trim=30 0 30 0, clip=true,width=\linewidth]{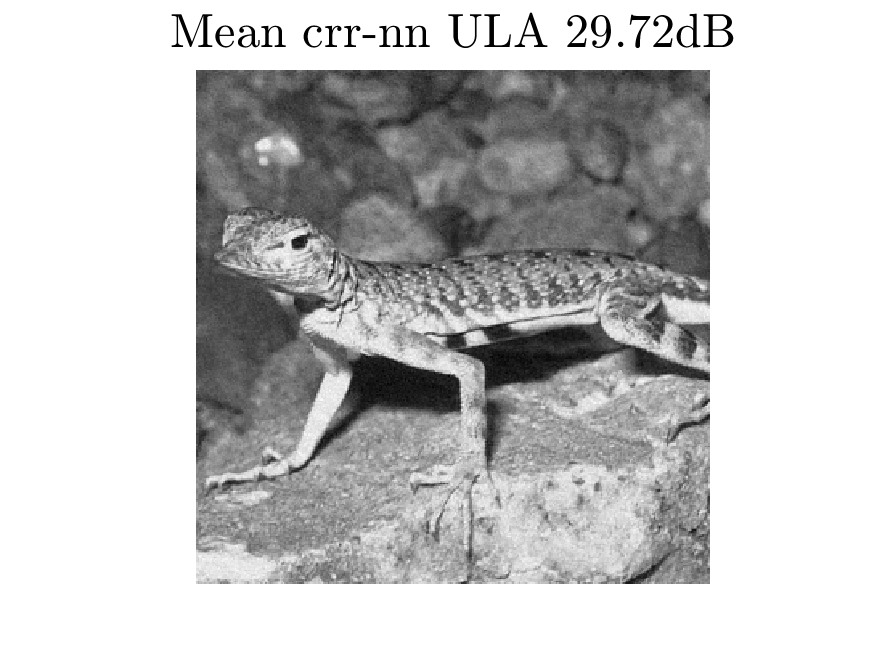}
\end{minipage}%
\begin{minipage}{.32\textwidth}
    \centering
    \includegraphics[trim=30 0 30 0, clip=true,width=\linewidth]{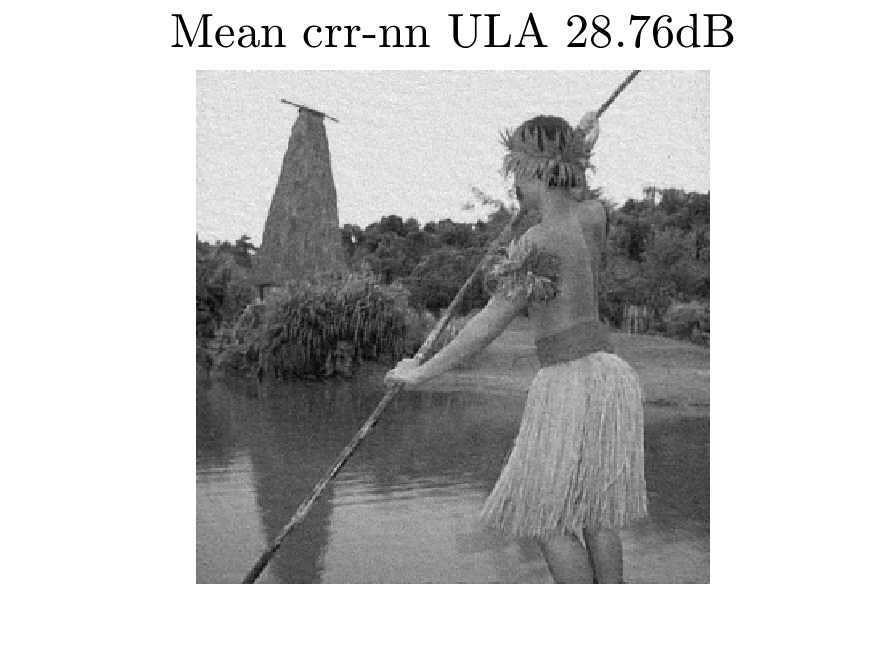}
\end{minipage}%

\begin{minipage}{.32\textwidth}
    \centering
    \includegraphics[trim=30 0 30 0, clip=true,width=\linewidth]{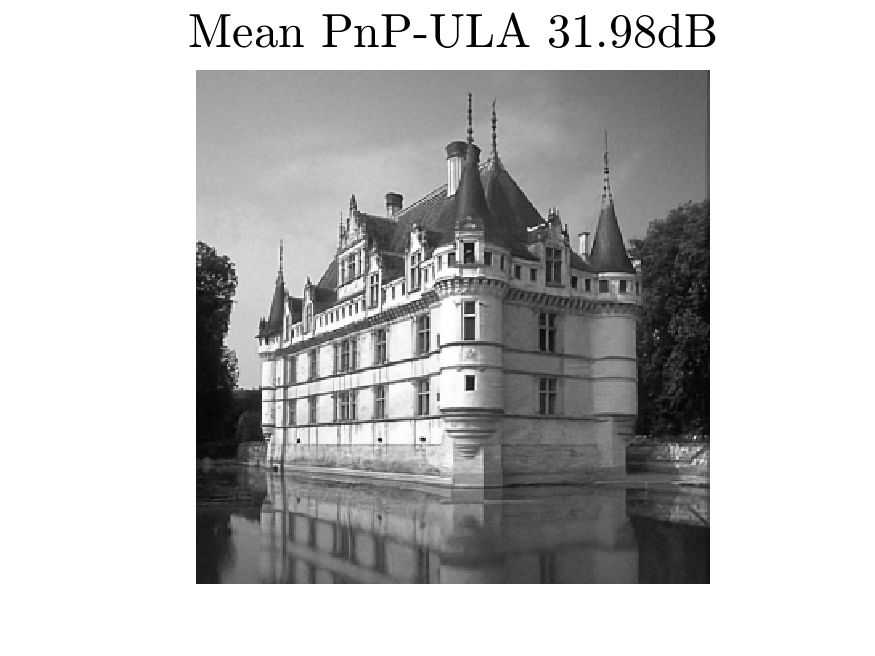}
\end{minipage}%
\begin{minipage}{.32\textwidth}
    \centering
    \includegraphics[trim=30 0 30 0, clip=true,width=\linewidth]{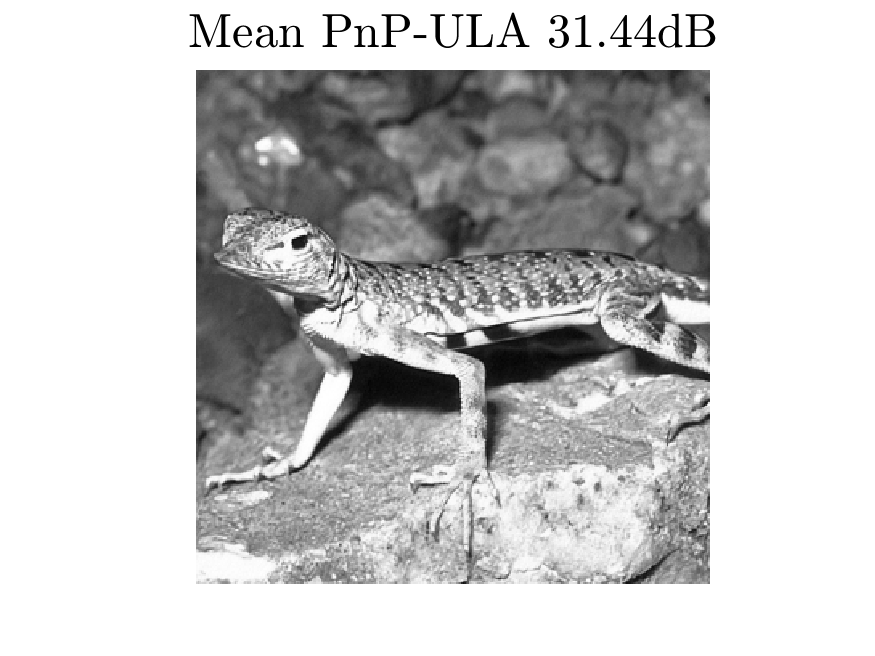}
\end{minipage}%
\begin{minipage}{.32\textwidth}
    \centering
    \includegraphics[trim=30 0 30 0, clip=true,width=\linewidth]{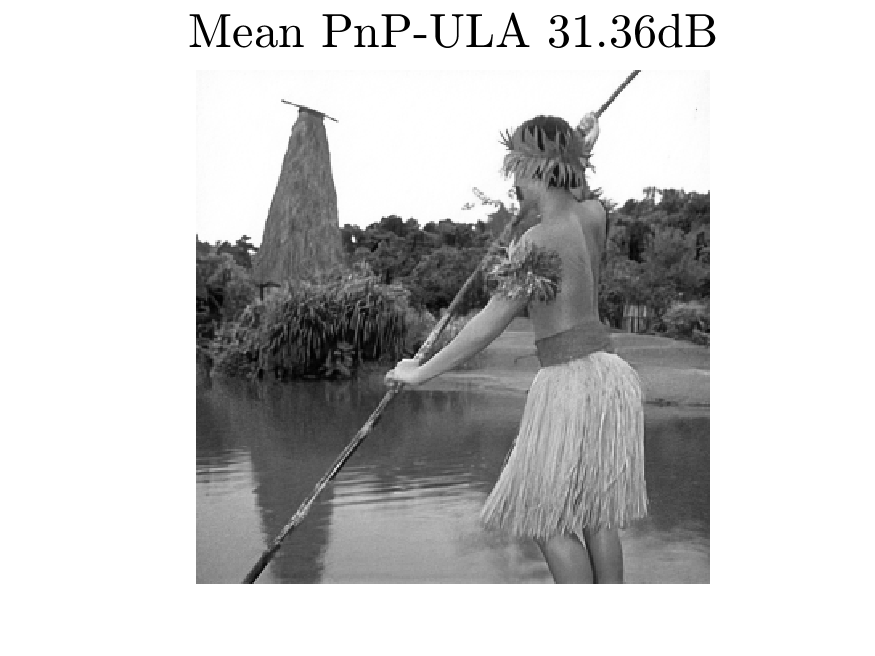}
\end{minipage}%

\caption{MAP solution for the motion deconvolution problem using a gradient method \cite{goujon2022crrnn} for the CRR-NN model (first line) and posterior means for IMLA, SKROCK, ULA and PnP-ULA (line 2-5) for \texttt{castle, lizard} and \texttt{person} images (left to right) including the corresponding PSNR values in $dB$.}
\label{fig:mean-motion-comp}
\end{figure}

\begin{figure}[h!]
\centering


\begin{subfigure}{.32\textwidth}
    \centering
    \includegraphics[trim=20 0 20 0, clip=true,width=\linewidth]{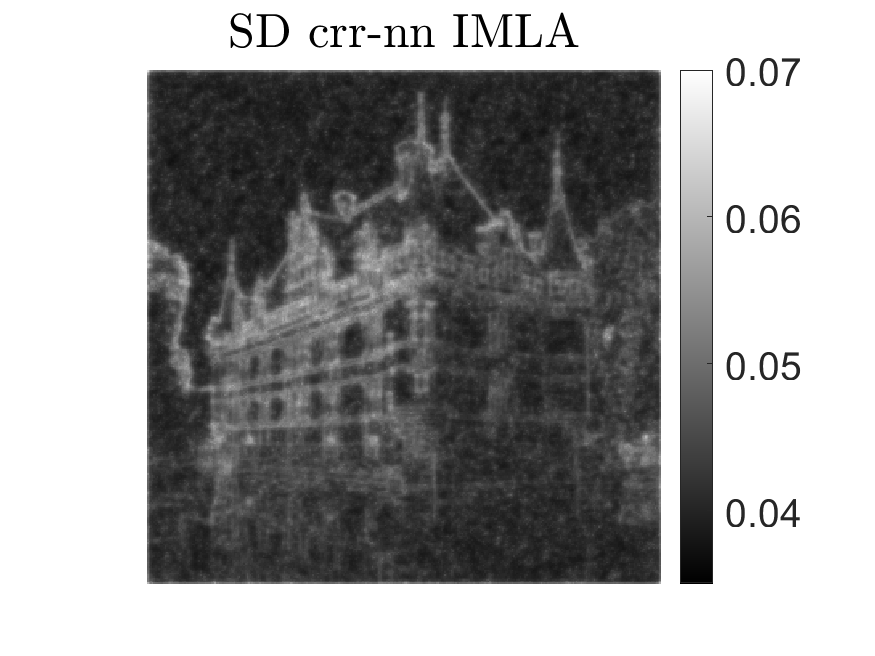}
\end{subfigure}%
\begin{subfigure}{.32\textwidth}
    \centering
    \includegraphics[trim=20 0 20 0, clip=true,width=\linewidth]{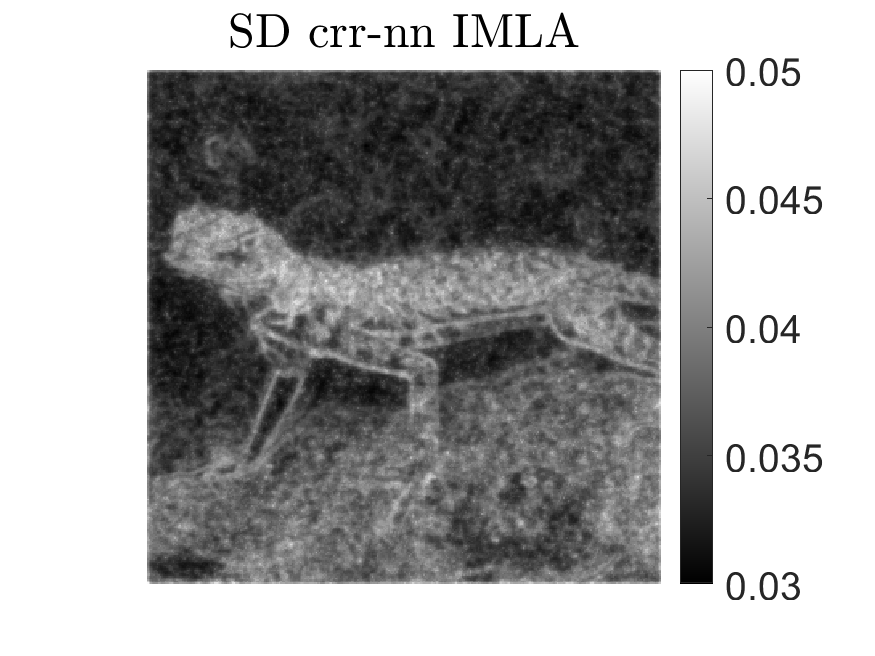}
\end{subfigure}%
\begin{subfigure}{.32\textwidth}
    \centering
    \includegraphics[trim=20 0 20 0, clip=true,width=\linewidth]{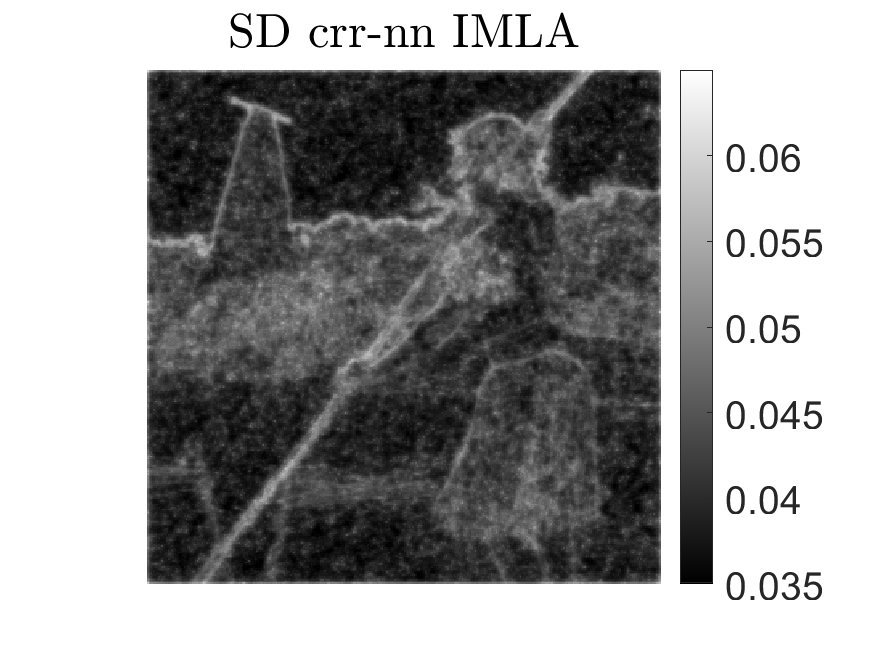}
\end{subfigure}%

\begin{subfigure}{.32\textwidth}
    \centering
    \includegraphics[trim=20 0 20 0, clip=true,width=\linewidth]{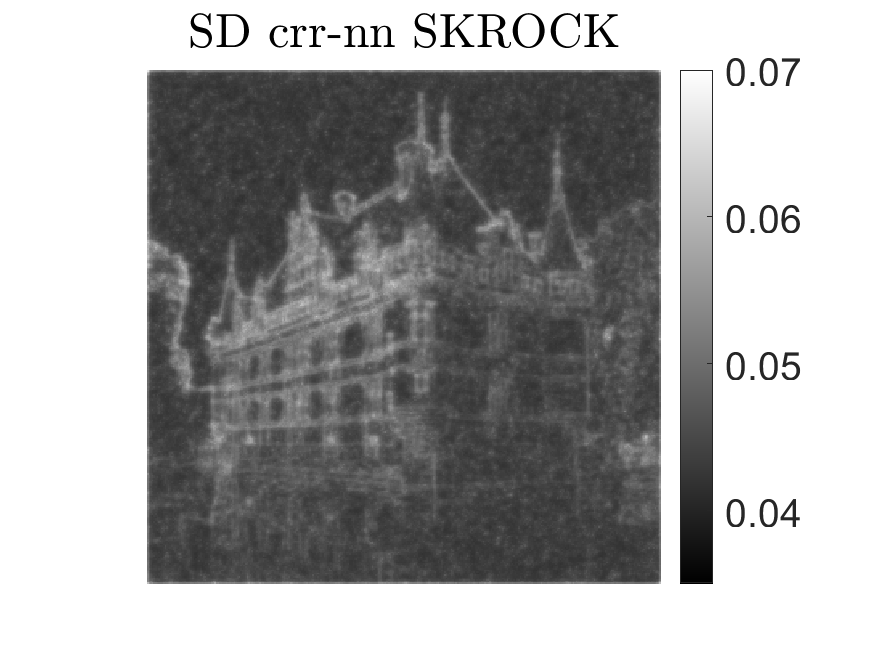}
\end{subfigure}%
\begin{subfigure}{.32\textwidth}
    \centering
    \includegraphics[trim=20 0 20 0, clip=true,width=\linewidth]{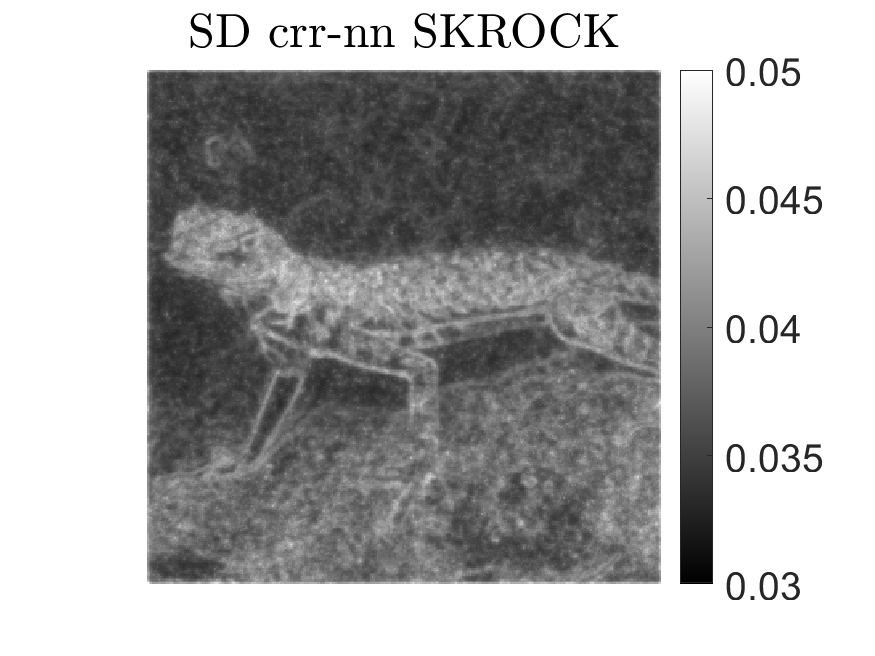}
\end{subfigure}%
\begin{subfigure}{.32\textwidth}
    \centering
    \includegraphics[trim=20 0 20 0, clip=true,width=\linewidth]{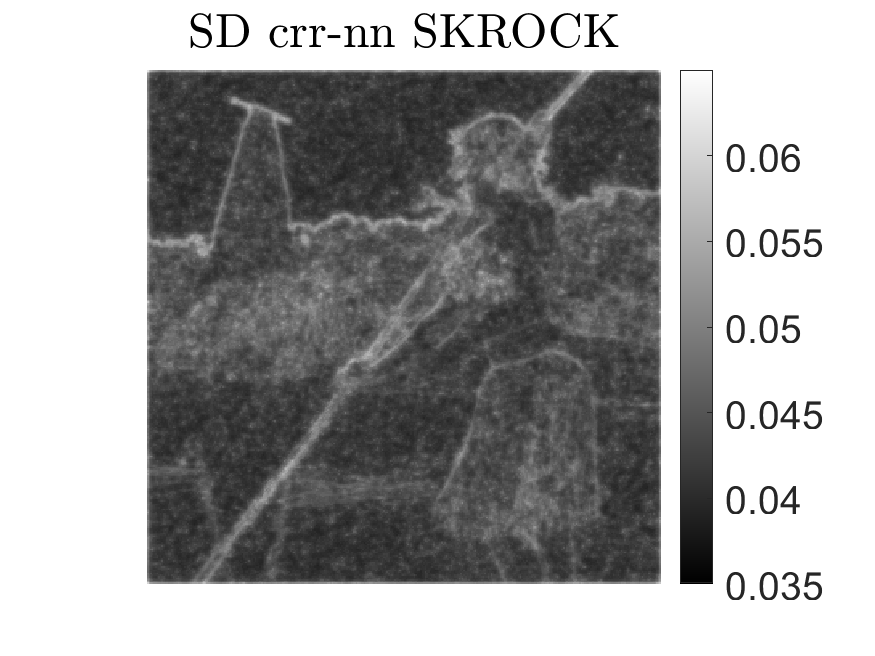}
\end{subfigure}%

\begin{subfigure}{.32\textwidth}
    \centering
    \includegraphics[trim=20 0 20 0, clip=true,width=\linewidth]{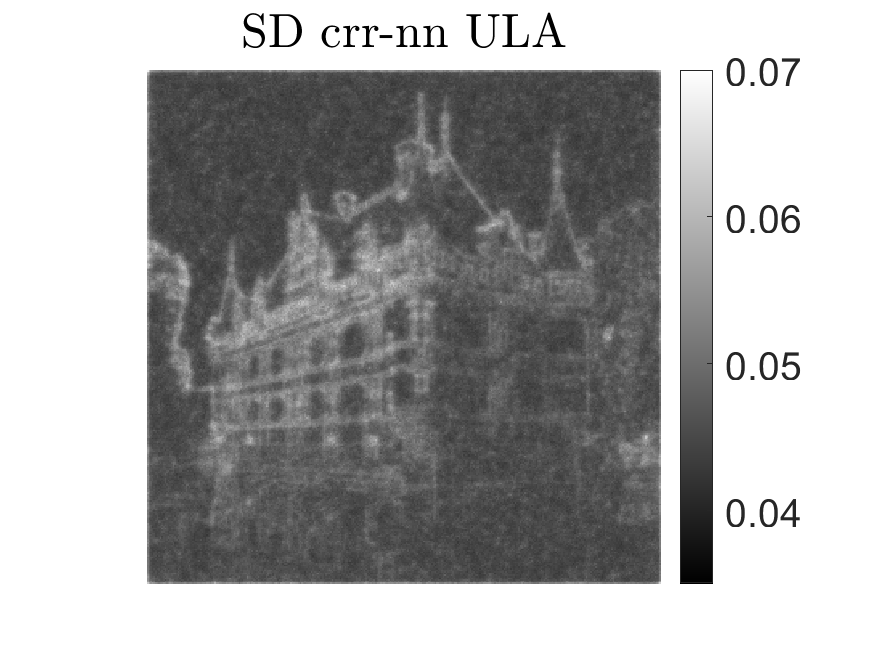}
\end{subfigure}%
\begin{subfigure}{.32\textwidth}
    \centering
    \includegraphics[trim=20 0 20 0, clip=true,width=\linewidth]{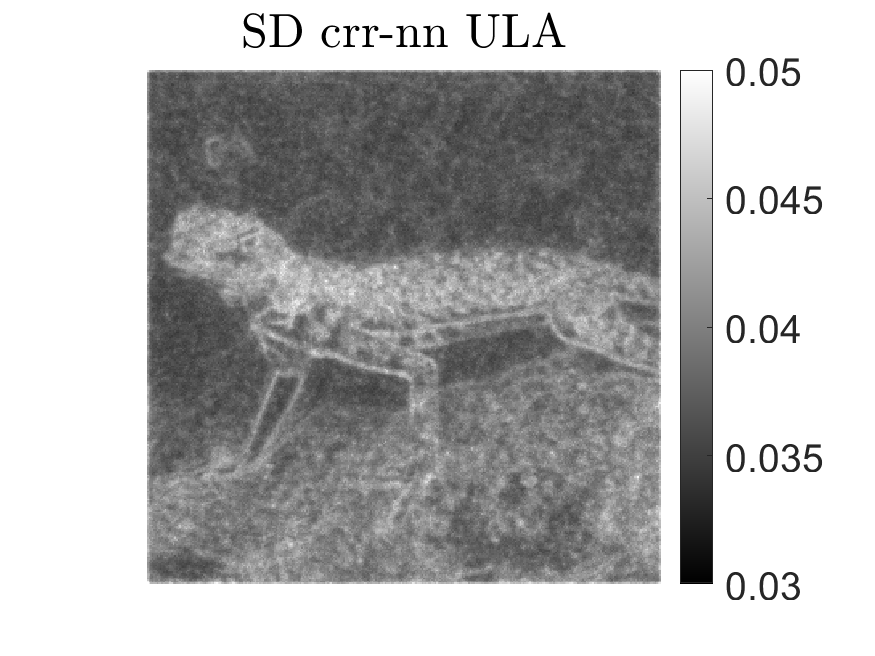}
\end{subfigure}%
\begin{subfigure}{.32\textwidth}
    \centering
    \includegraphics[trim=20 0 20 0, clip=true,width=\linewidth]{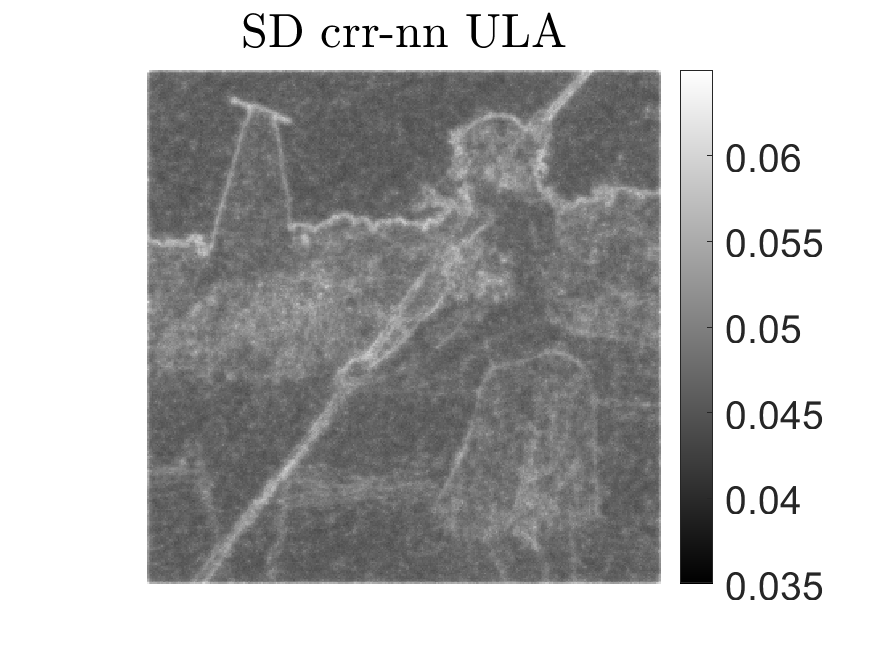}
\end{subfigure}%

\begin{subfigure}{.32\textwidth}
    \centering
    \includegraphics[trim=20 0 20 0, clip=true,width=\linewidth]{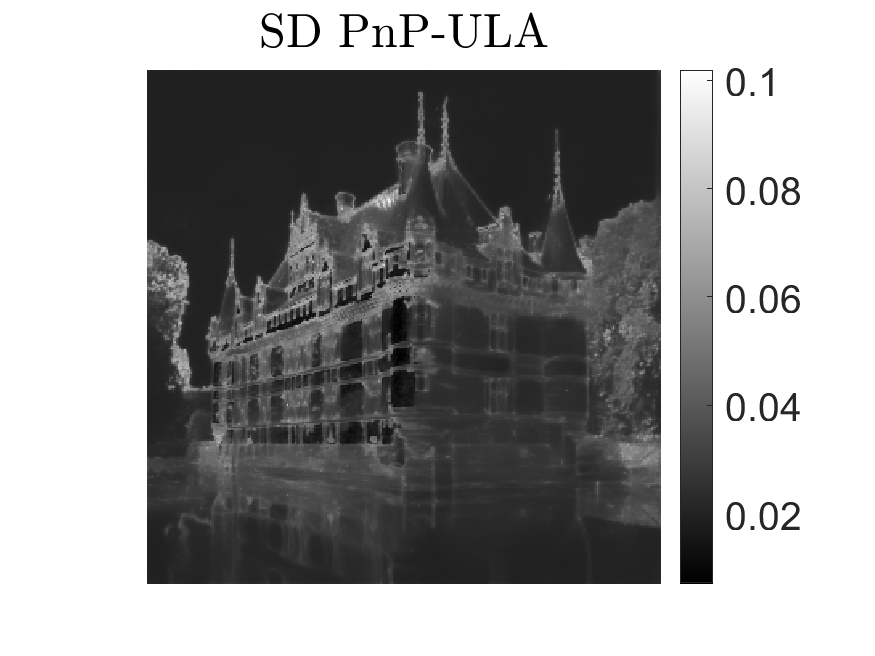}
\end{subfigure}%
\begin{subfigure}{.32\textwidth}
    \centering
    \includegraphics[trim=20 0 20 0, clip=true,width=\linewidth]{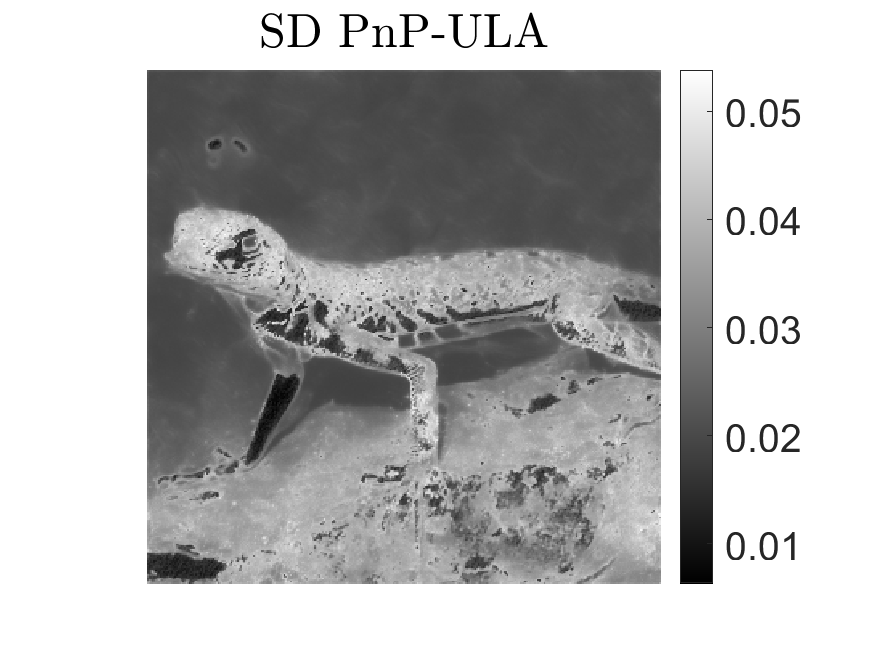}
\end{subfigure}%
\begin{subfigure}{.32\textwidth}
    \centering
    \includegraphics[trim=20 0 20 0, clip=true,width=\linewidth]{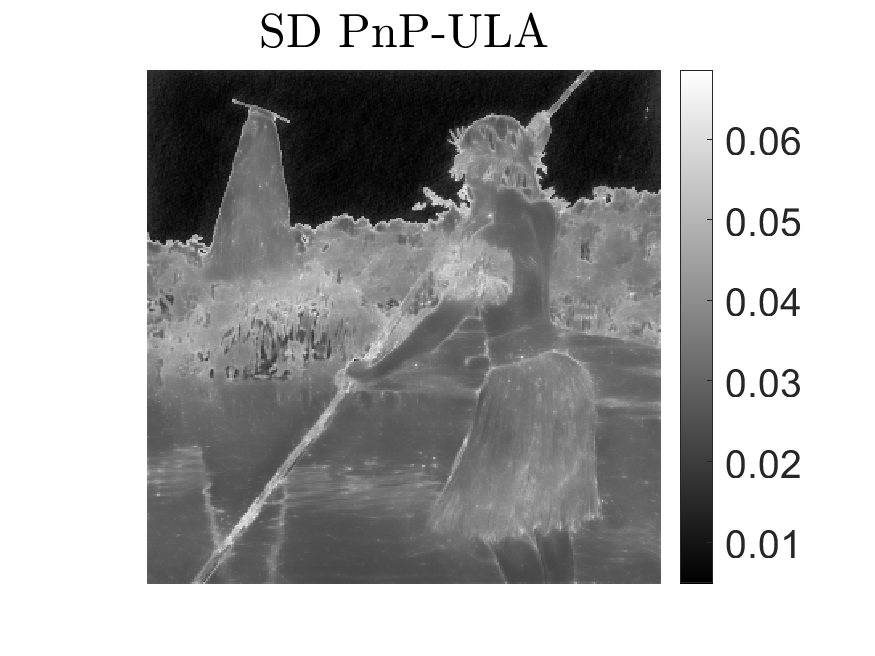}
\end{subfigure}%
\caption{Marginal posterior standard deviation (SD) for the motion deconvolution problem for IMLA, SKROCK, ULA and PnP-ULA for \texttt{castle, lizard} and \texttt{person} images (left to right).}
\label{fig:std-motion-comp}
\end{figure}

\begin{figure}[h!]
    \centering
\begin{minipage}{\textwidth}   


\begin{minipage}{0.24\textwidth}
    \includegraphics[width=\linewidth]{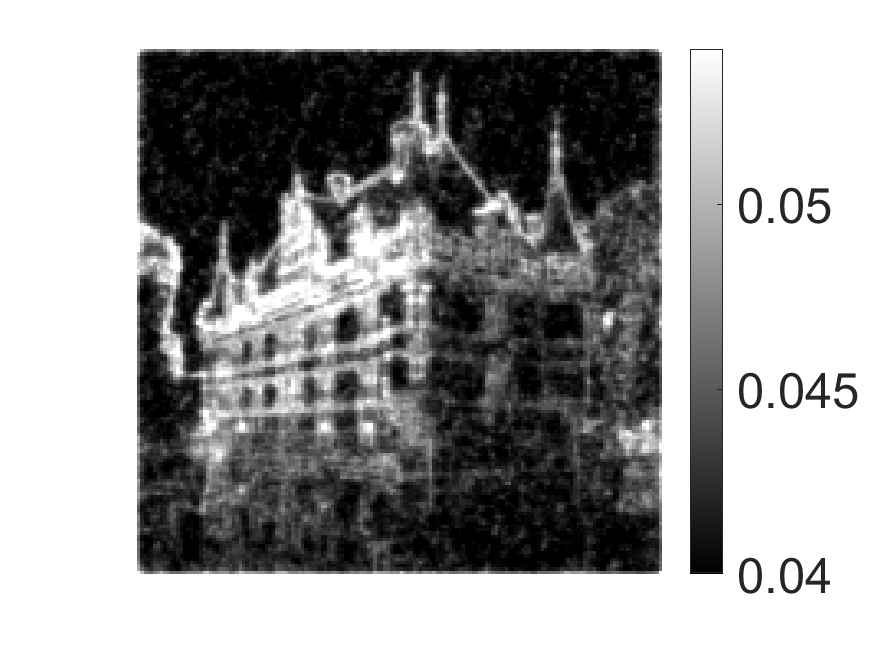}
\end{minipage}
\begin{minipage}{0.24\textwidth}
    \includegraphics[width=\linewidth]{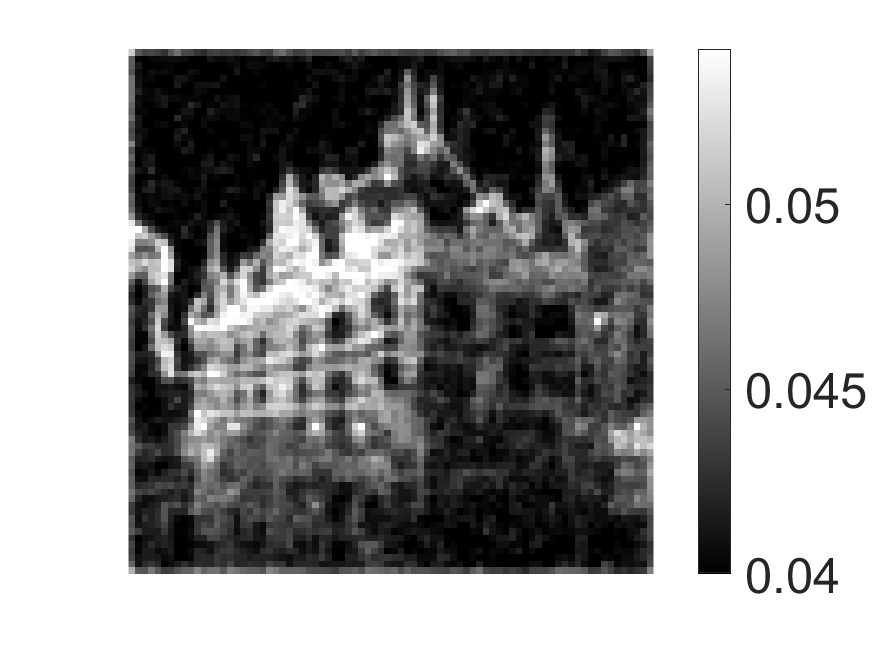}
\end{minipage}
\begin{minipage}{0.24\textwidth}
    \includegraphics[width=\linewidth]{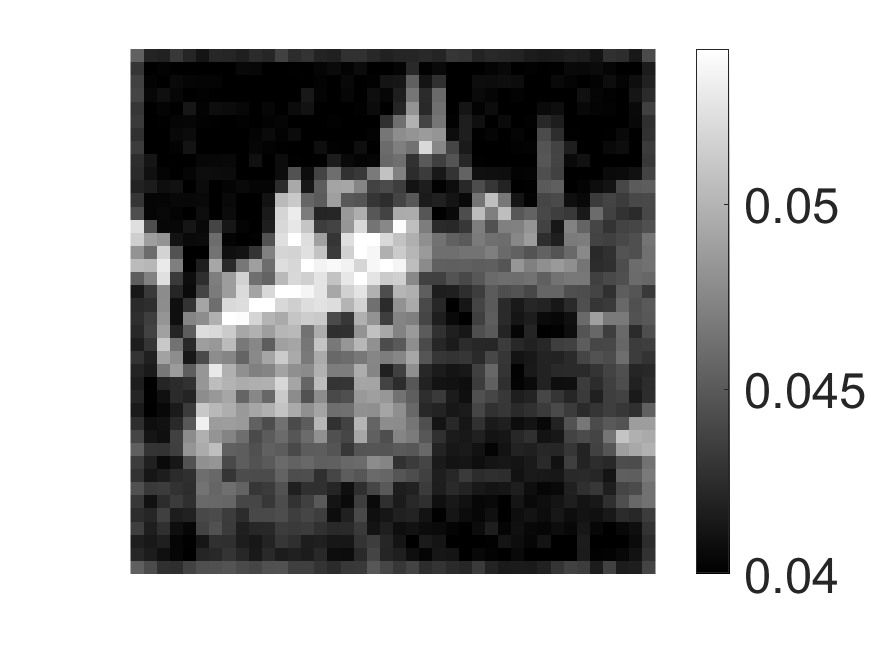}
\end{minipage}
\begin{minipage}{0.24\textwidth}
    \includegraphics[width=\linewidth]{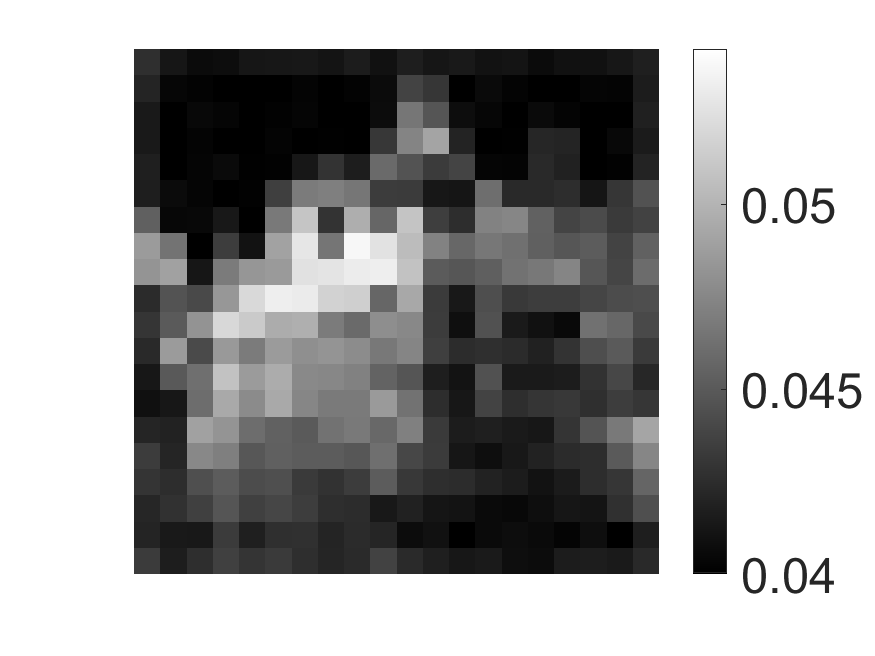}
\end{minipage}

\begin{minipage}{0.24\textwidth}
    \includegraphics[width=\linewidth]{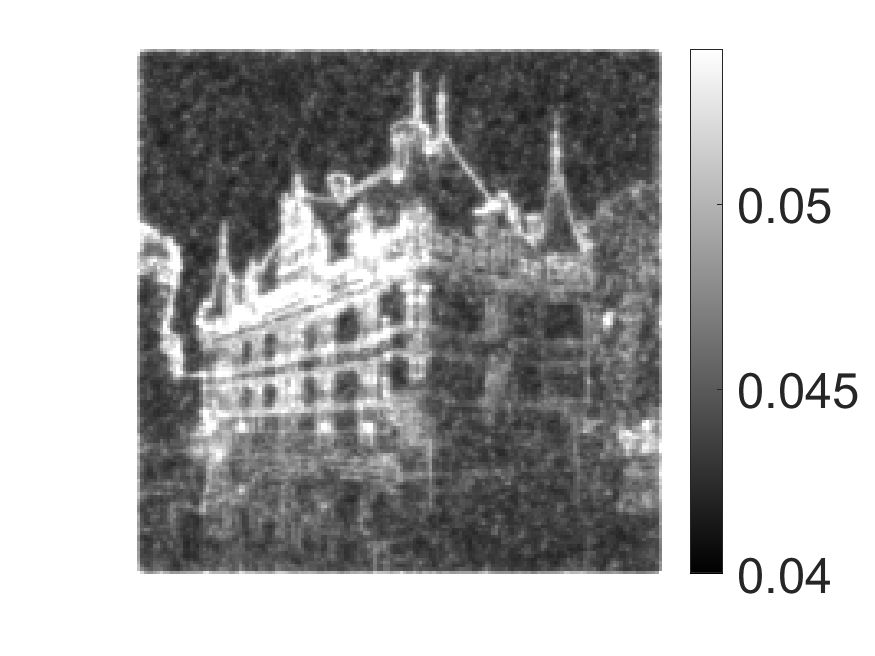}
\end{minipage}
\begin{minipage}{0.24\textwidth}
    \includegraphics[width=\linewidth]{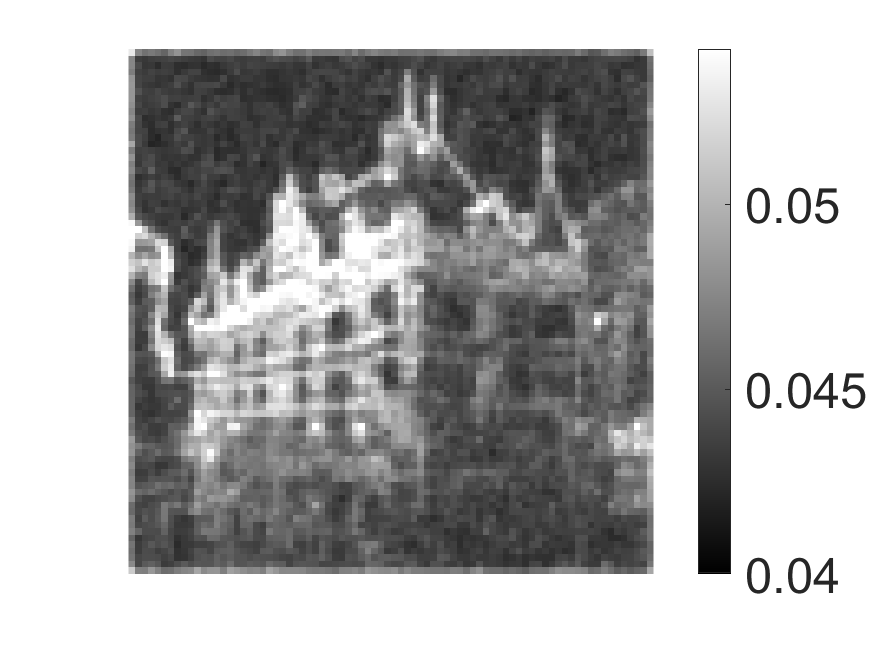}
\end{minipage}
\begin{minipage}{0.24\textwidth}
    \includegraphics[width=\linewidth]{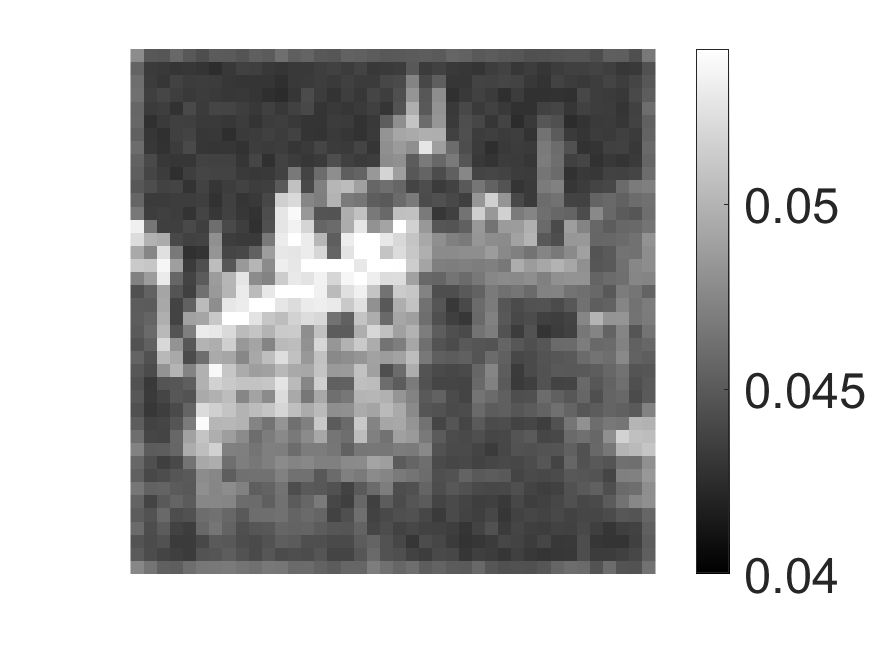}
\end{minipage}
\begin{minipage}{0.24\textwidth}
    \includegraphics[width=\linewidth]{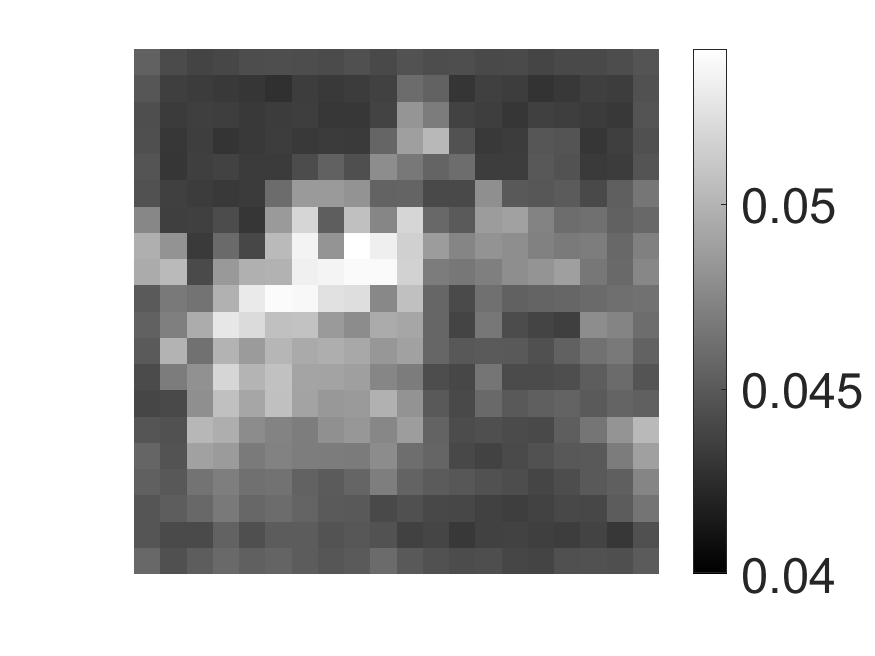}
\end{minipage}

\begin{minipage}{0.24\textwidth}
    \includegraphics[width=\linewidth]{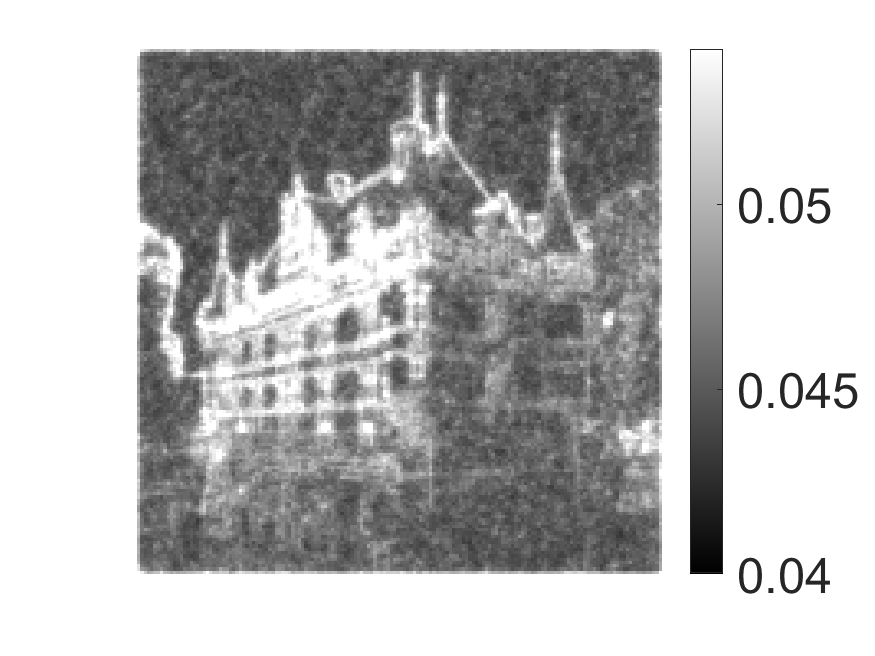}
\end{minipage}
\begin{minipage}{0.24\textwidth}
    \includegraphics[width=\linewidth]{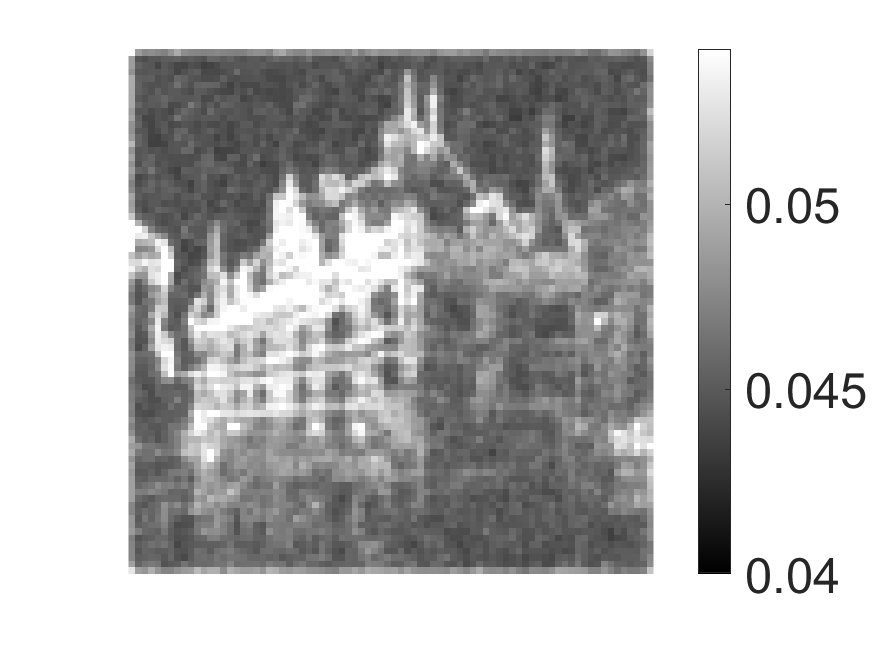}
\end{minipage}
\begin{minipage}{0.24\textwidth}
    \includegraphics[width=\linewidth]{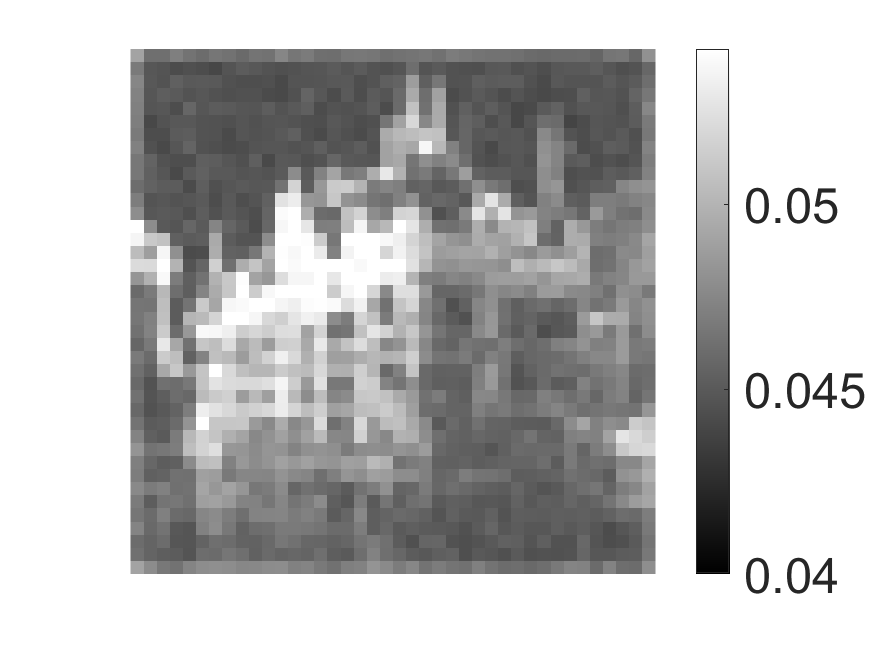}
\end{minipage}
\begin{minipage}{0.24\textwidth}
    \includegraphics[width=\linewidth]{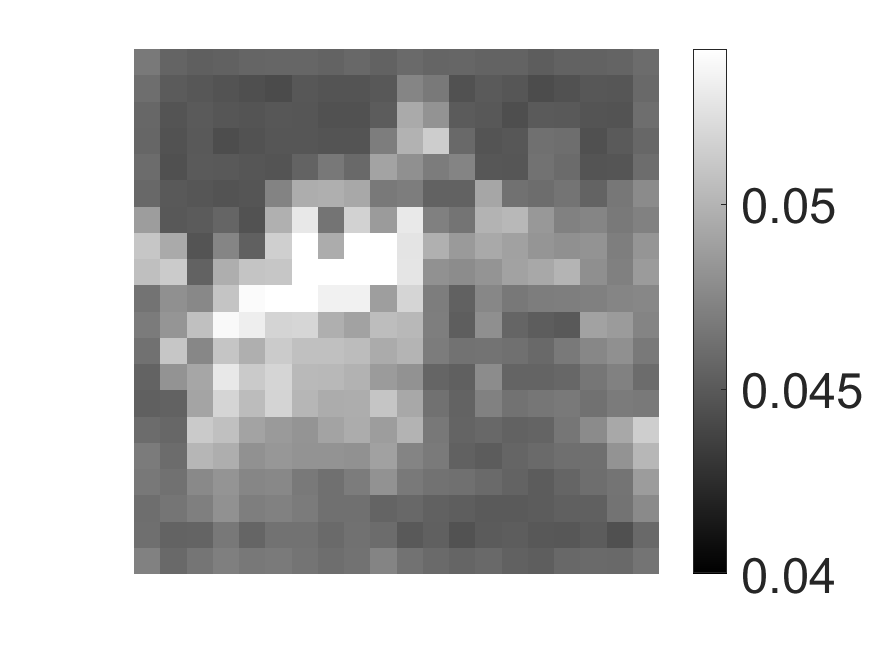}
\end{minipage}


\begin{minipage}{0.24\textwidth}
    \includegraphics[width=\linewidth]{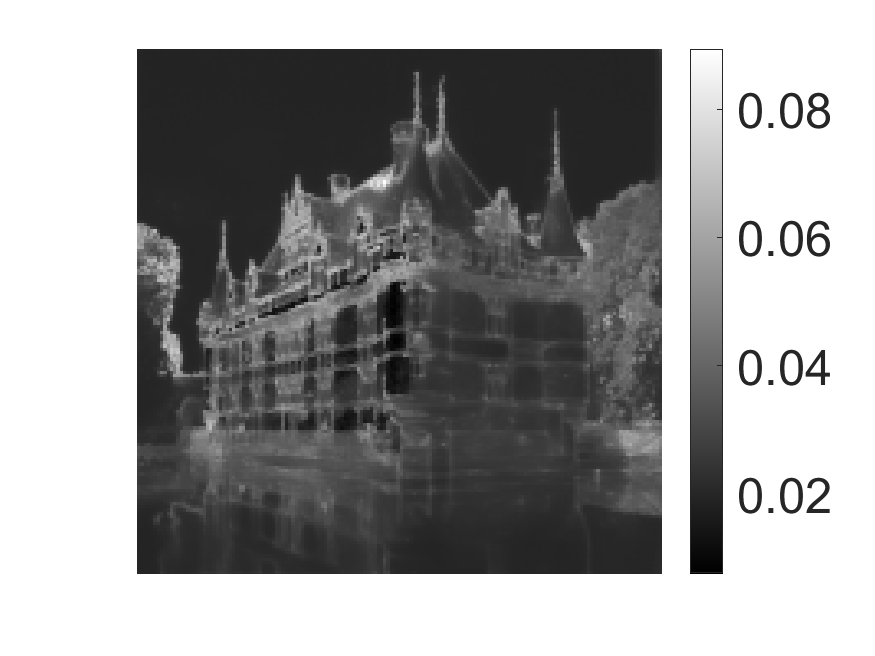}
\end{minipage}
\begin{minipage}{0.24\textwidth}
    \includegraphics[width=\linewidth]{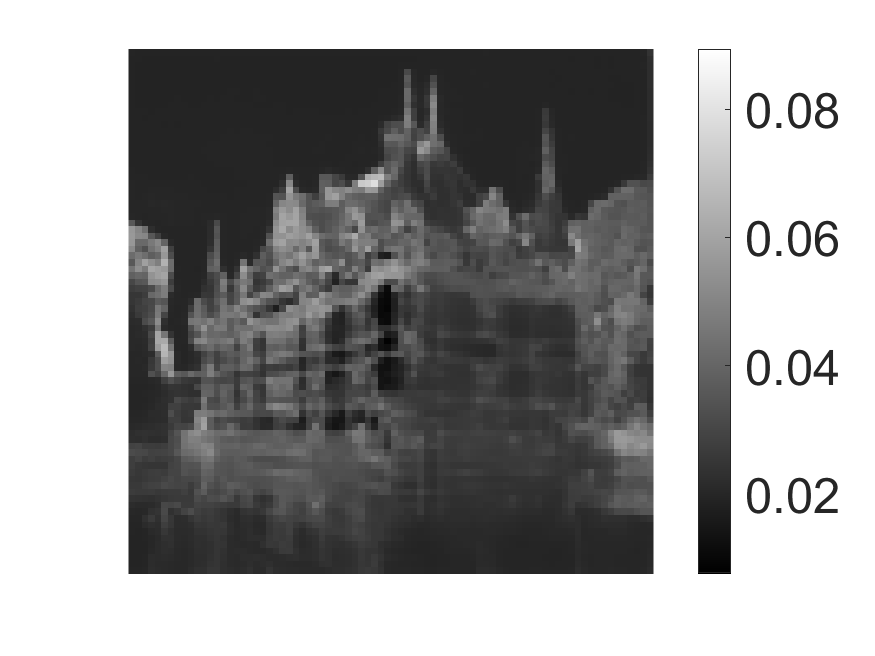}
\end{minipage}
\begin{minipage}{0.24\textwidth}
    \includegraphics[width=\linewidth]{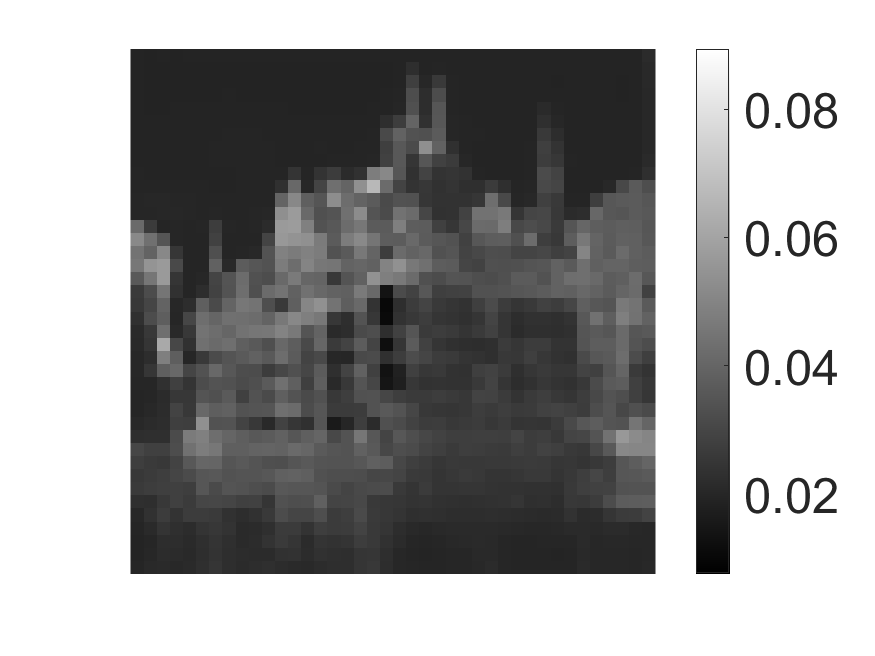}
\end{minipage}
\begin{minipage}{0.24\textwidth}
    \includegraphics[width=\linewidth]{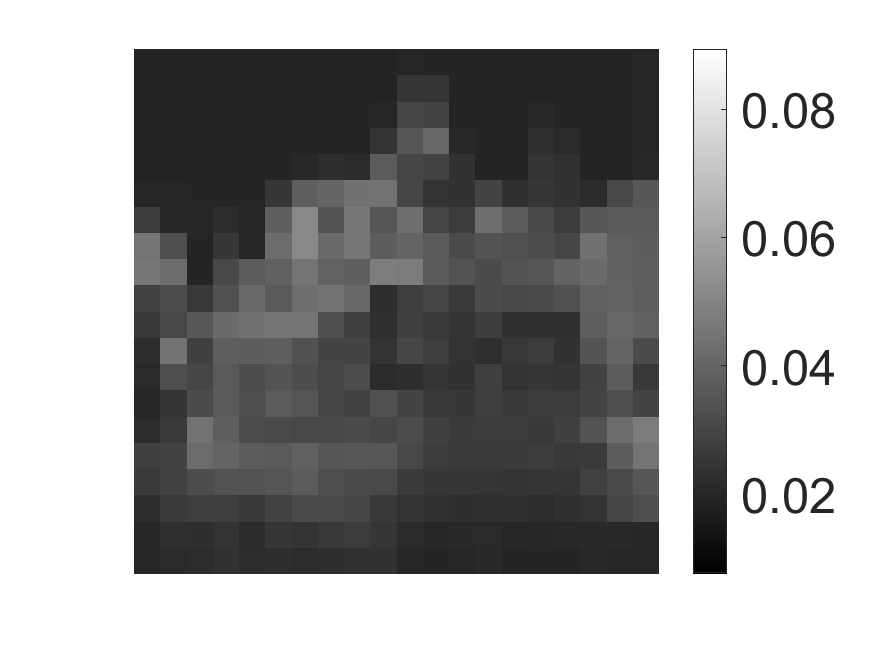}
\end{minipage}
\captionof{figure}{Visualisation of uncertainty by displaying the marginal standard deviations at scales $2\times2$, $4\times4$, $8\times8$, and $16\times16$ pixels (left to right), for IMLA, SKROCK, ULA, and PnP-ULA (top to bottom) for the \texttt{castle} image.}
\label{fig:castle-uq-scales}
\end{minipage}
\end{figure}

{From Figures \ref{fig:mean-motion-comp}-\ref{fig:castle-uq-scales}, we observe that IMLA achieves a higher PSNR than ULA and SKROCK (see Figure \ref{fig:mean-motion-comp}) and that it also delivers more stable estimates of the second-order moments. In particular, ULA and SKROCK have noticeably higher Monte Carlo error noise in the background of Figure~\ref{fig:std-motion-comp}, which indicates that IMLA is exploring the solution space more quickly. The MAP estimator delivers better point estimates in terms of PSNR, which is common in models that are log-concave, especially when the hyperparameters of \eqref{eq:post1} are tuned to optimise the PSNR of the MAP solution. In addition, the MMSE calculated by PnP-ULA outperforms the MAP and MMSE solutions provided by the log-concave model \eqref{eq:post1} due to the higher accuracy of the non-convex prior. This additional accuracy comes at a potentially significantly higher computational cost, as navigating non-convex landscapes requires additional iterations and evaluating the neural network \cite{RyuLWCWY19} is significantly more expensive than evaluating \cite{goujon2022crrnn}. 

In order to further illustrate the convergence properties of the IMLA, ULA and SKROCK algorithms, Figure \ref{fig:psnr-acf-motion-comp} depicts the PSNR of the running mean, the trace of the statistic $\log \pi$, and the autocorrelation function (ACF) of the slowest and fastest mixing components of the Markov chains for the \texttt{castle} image (the results for the other two images were quantitatively equivalent). Note that the slow/fast components have been identified by performing an approximate singular value decomposition of the posterior covariance and then computing the ACF for the projections of the chain on the eigenvectors with largest (fastest) and smallest (slowest) eigenvalues.

In this figure, one can observe that IMLA and SKROCK behave very similarly both in terms of their transient regimes (evolution of PSNR and of $\log \pi$), and in terms of their stationary regimes (described by the ACFs), whereas ULA converges quickly to the typical set of $\pi$ (as illustrated by the evolution of $\log \pi$) but explores the space more slowly (see the ACF plots for ULA). This is consistent with the theory for ULA and IMLA, and in agreement with our understanding of SKROCK based on its analysis for Gaussian models.

\begin{figure}[p]
\begin{minipage}{.49\textwidth}
    \centering
    \includegraphics[width=\linewidth]{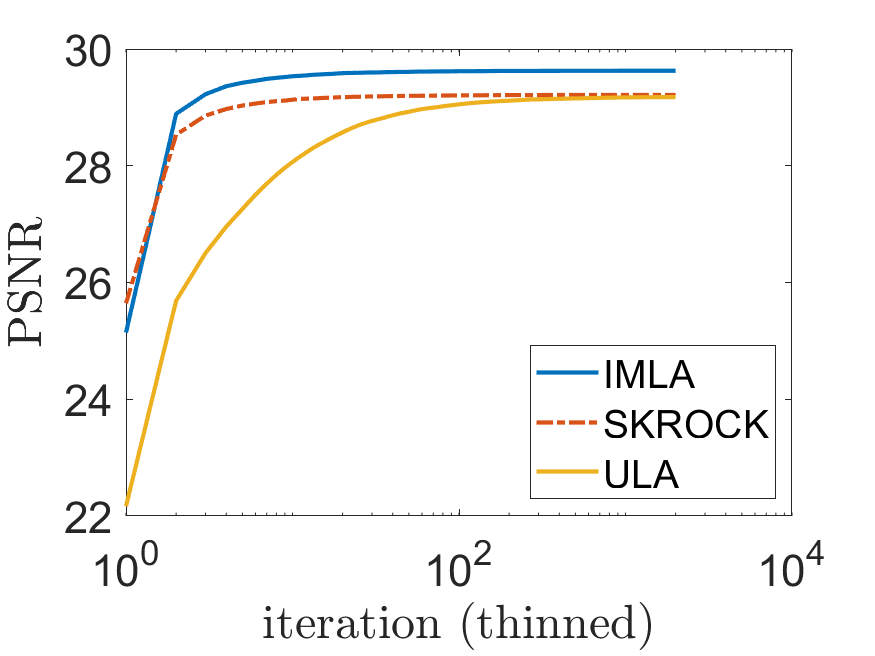}
\end{minipage}%
\begin{minipage}{.49\textwidth}
    \centering
    \includegraphics[width=\linewidth]{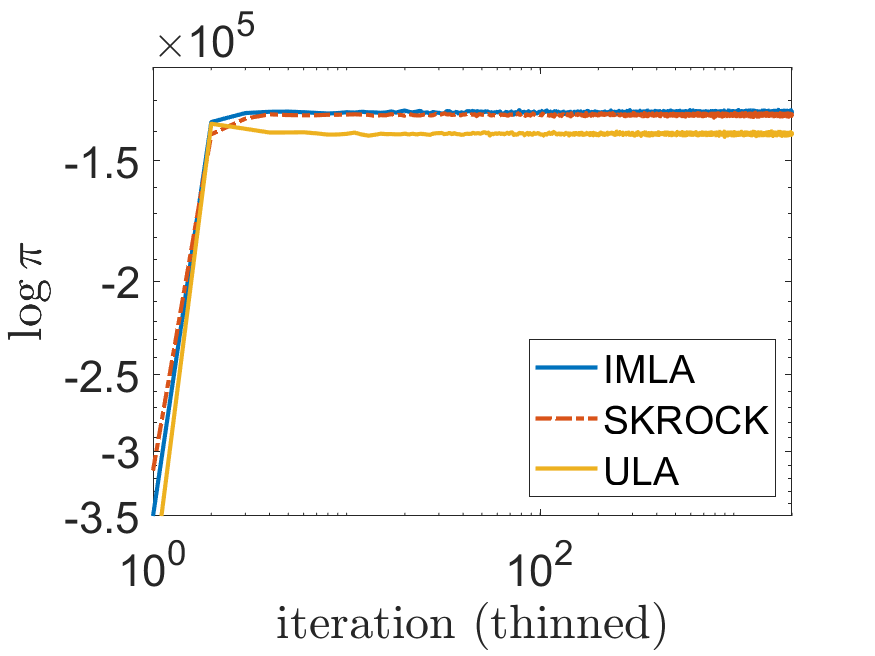}
\end{minipage}%

\begin{minipage}{.49\textwidth}
    \centering
    \includegraphics[width=\linewidth]{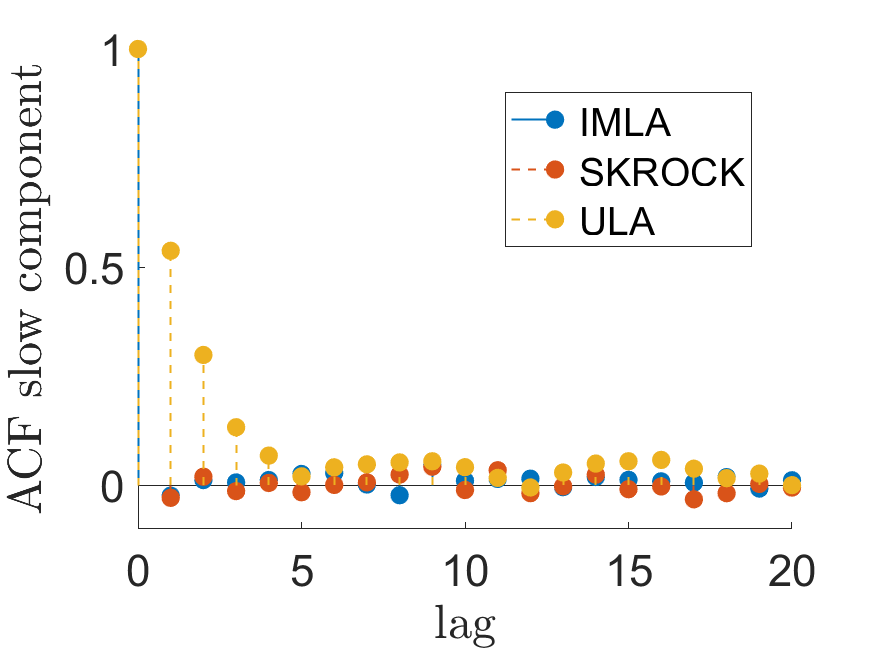}
\end{minipage}%
\begin{minipage}{.49\textwidth}
    \centering
    \includegraphics[width=\linewidth]{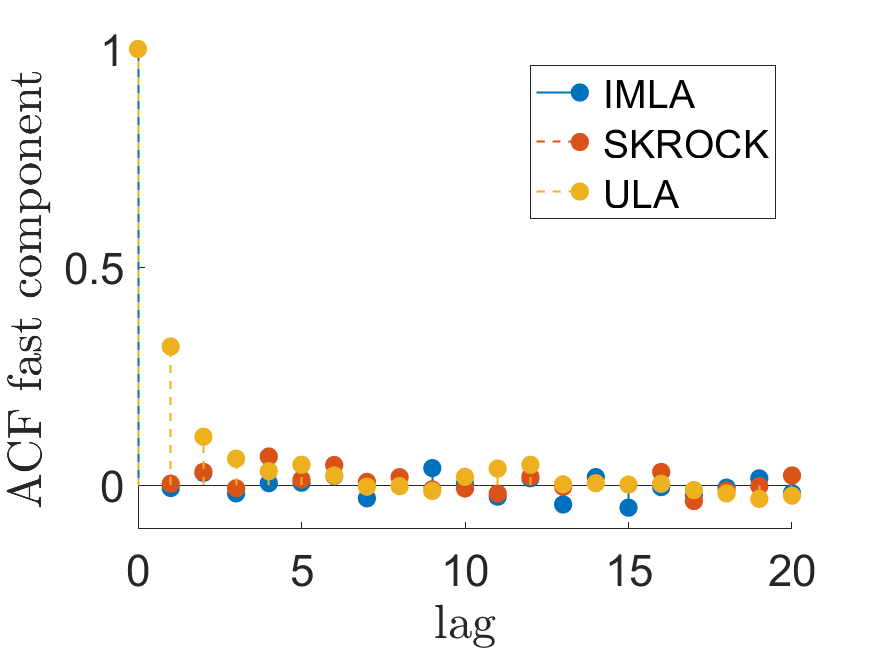}
\end{minipage}%

\begin{minipage}{.49\textwidth}
    \centering
    \includegraphics[width=\linewidth]{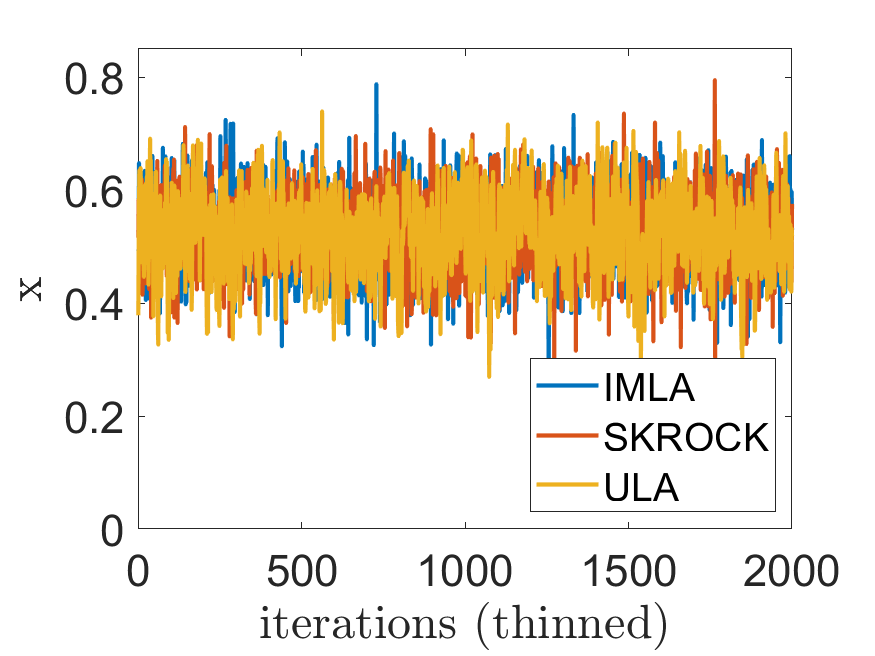}
\end{minipage}%
\begin{minipage}{.49\textwidth}
    \centering
    \includegraphics[width=\linewidth]{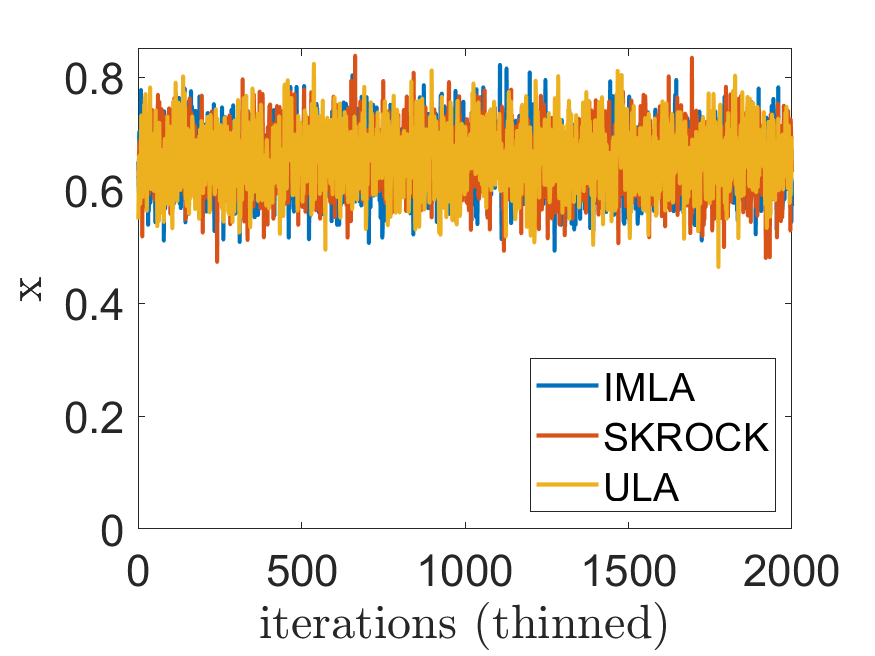}
\end{minipage}%
\caption{Motion deconvolution experiment. PSNR of the running mean (top left) and $\log \pi(x)$ (top right) for $N$ iterations, thinned fairly to obtain 2000 data points, \review{ ACF of the slowest and fastest component (middle left and right), and pixel traces for the respective slow/fast components (bottom left and right) }for IMLA, SKROCK and ULA for the \texttt{castle} image.}
\label{fig:psnr-acf-motion-comp}
\end{figure}
\begin{figure}[!htb]
    \centering
\begin{minipage}{\textwidth}   
\begin{minipage}{0.49\textwidth}
    \includegraphics[width=\linewidth]{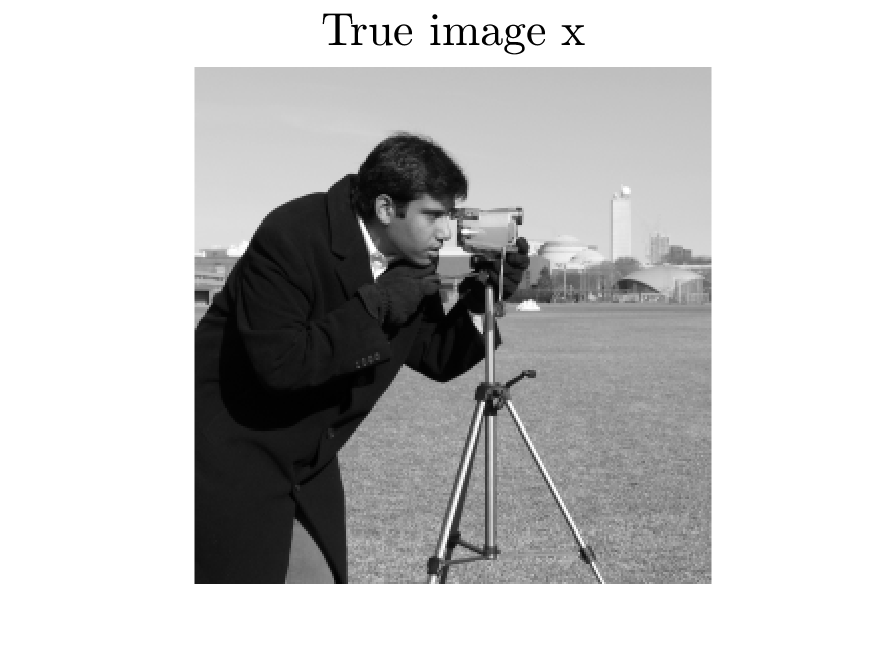}
\end{minipage}
\begin{minipage}{0.49\textwidth}
    \includegraphics[width=\linewidth]{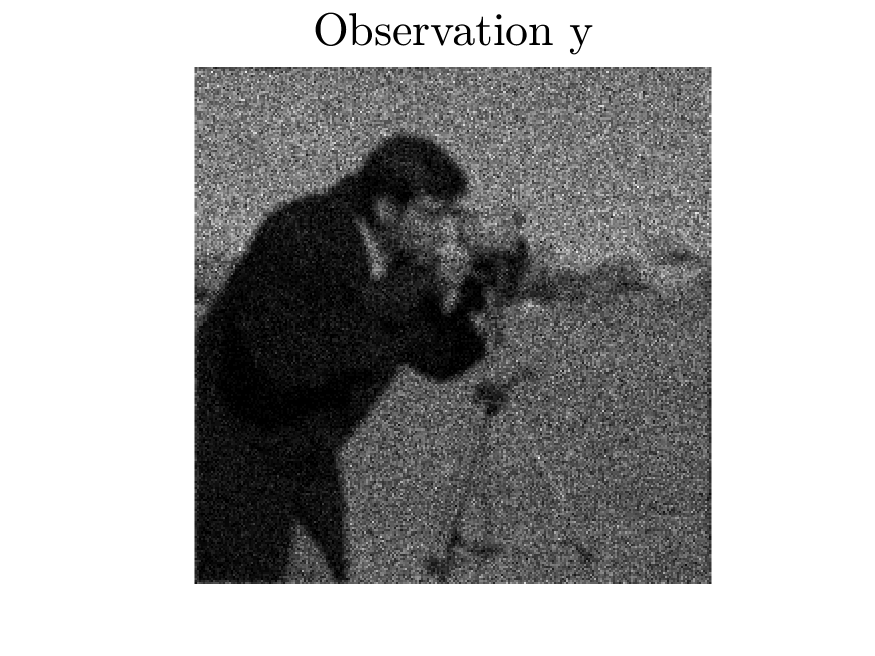}
\end{minipage}
\captionof{figure}{\texttt{cameraman} image $x$ (left) and observed image $y$ (right) of size $256\times256$ for the Poisson deconvolution experiment.}
\label{fig:poisson-data-cman}
\end{minipage}
\end{figure}
\begin{figure}[!htb]
    \centering
\begin{minipage}{\textwidth}   
\begin{minipage}{0.32\textwidth}
    \includegraphics[width=\linewidth]{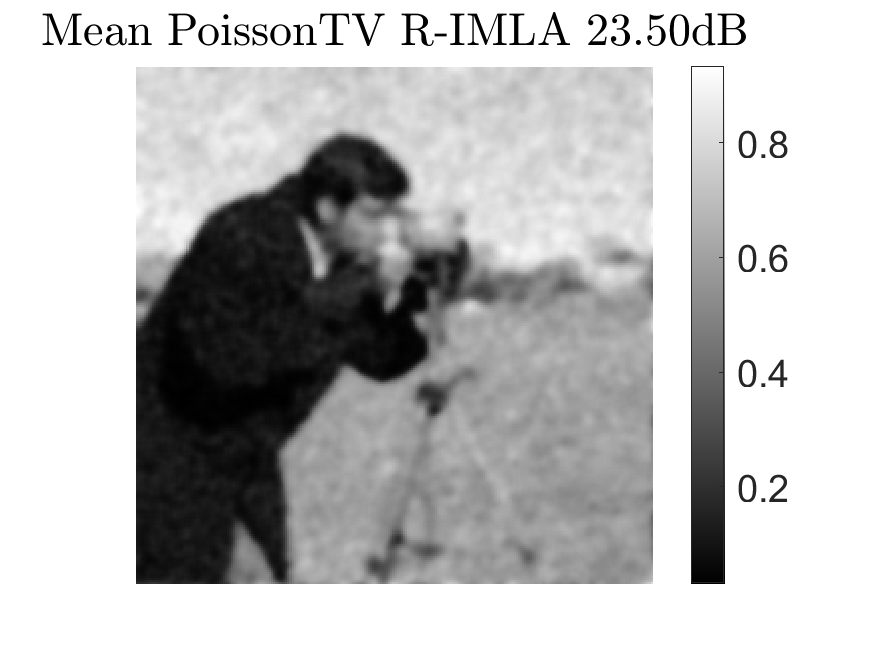}
\end{minipage}
\begin{minipage}{0.32\textwidth}
    \includegraphics[width=\linewidth]{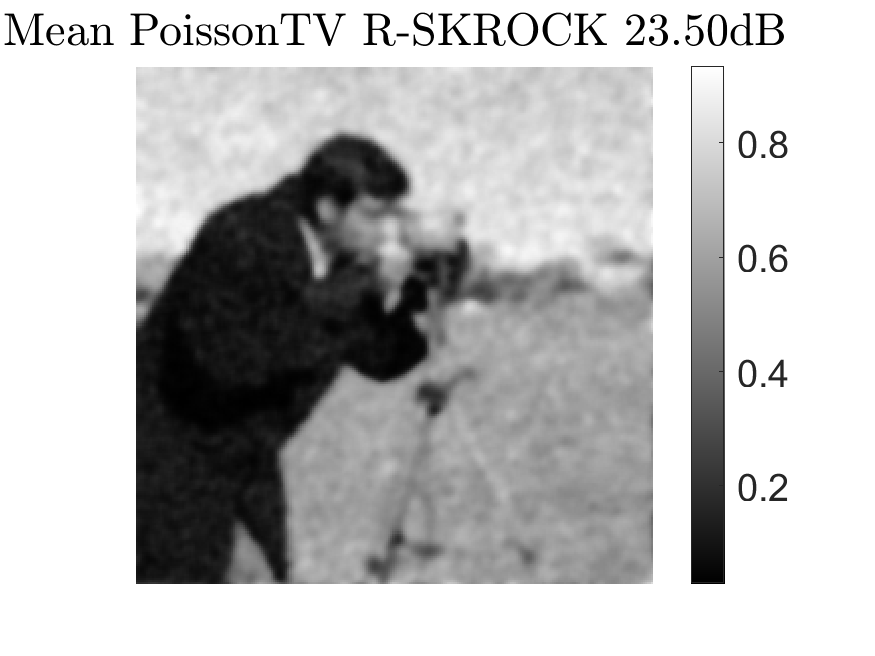}
\end{minipage}
\begin{minipage}{0.32\textwidth}
    \includegraphics[width=\linewidth]{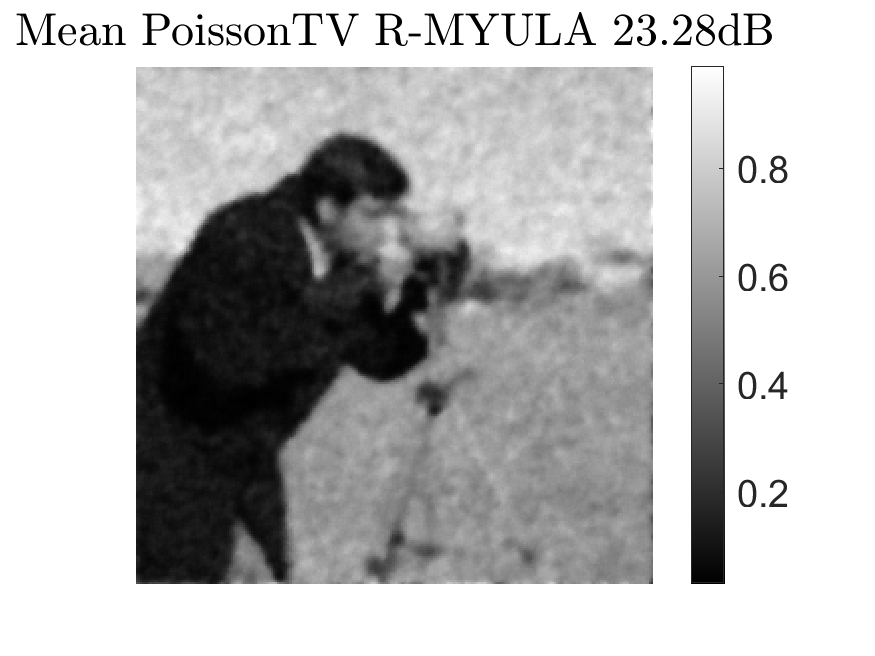}
\end{minipage}
\end{minipage}

\begin{minipage}{\textwidth}   
\begin{minipage}{0.32\textwidth}
    \includegraphics[width=\linewidth]{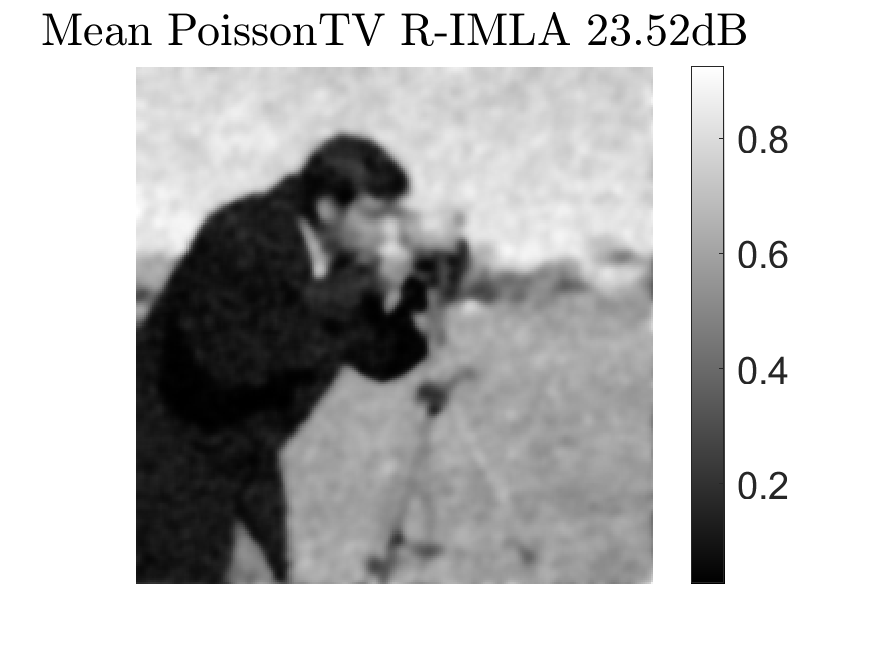}
\end{minipage}
\begin{minipage}{0.32\textwidth}
    \includegraphics[width=\linewidth]{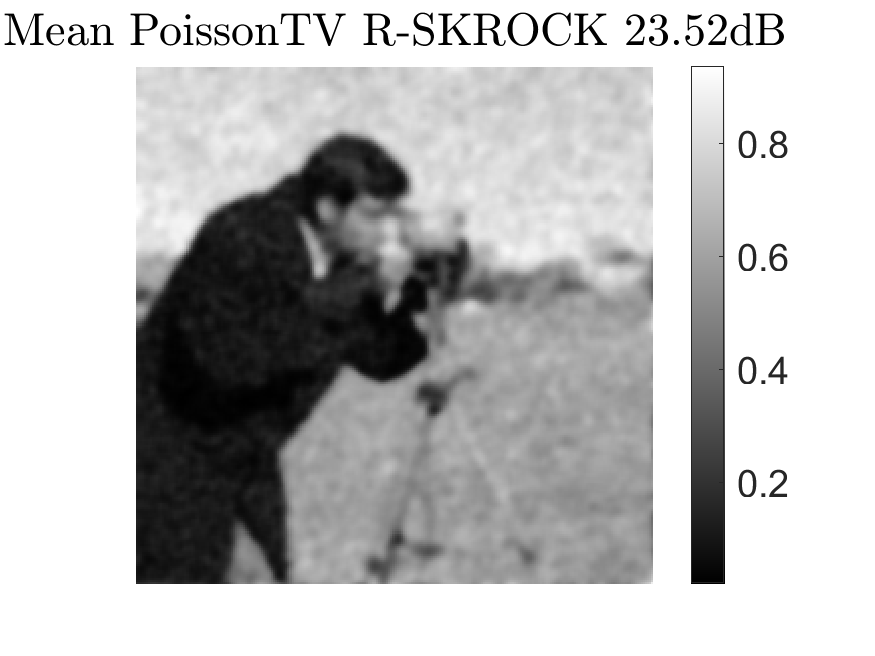}
\end{minipage}
\begin{minipage}{0.32\textwidth}
    \includegraphics[width=\linewidth]{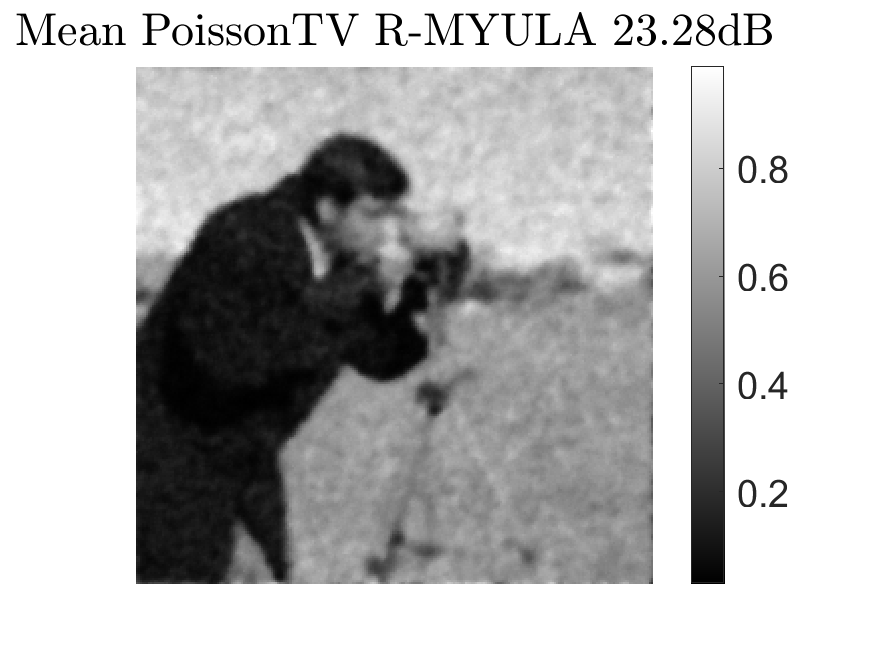}
\end{minipage}
\captionof{figure}{Posterior mean for Poisson deconvolution experiment. R-IMLA was run using $h=6.65\times 10^{-5}$ (top), and $h=1.16\times 10^{-3}$ (bottom) which is equivalent to the effective step size of R-SKROCK when $s=10$ (top) or $s=40$ (bottom); R-MYULA was run using $h=1/L$.}
\label{fig:post-mean-integration-time-cman}
\end{minipage}
\end{figure}
\begin{figure}[t]
    \centering
\begin{minipage}{\textwidth}   
\begin{minipage}{0.32\textwidth}
    \includegraphics[width=\linewidth]{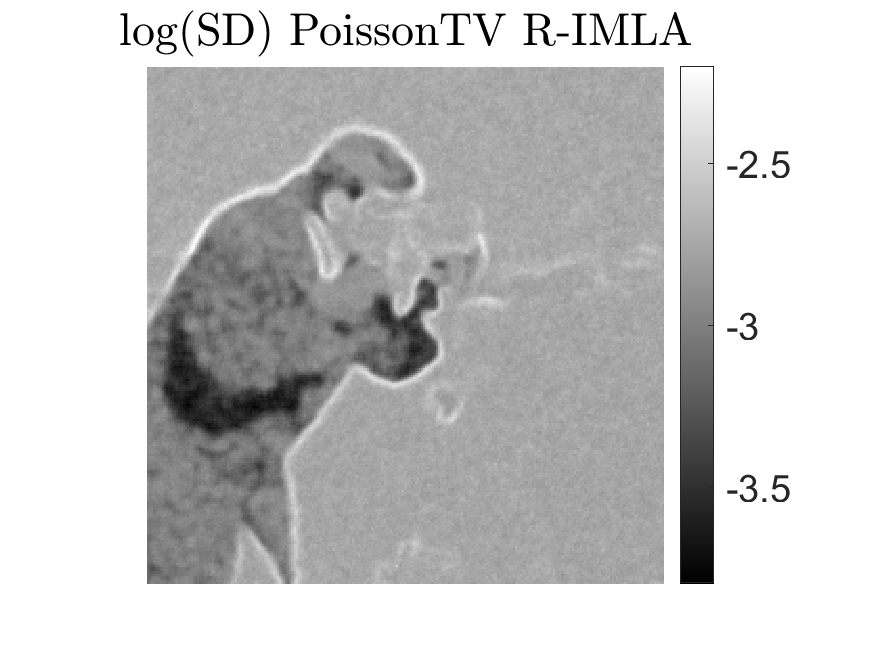}
\end{minipage}
\begin{minipage}{0.32\textwidth}
    \includegraphics[width=\linewidth]{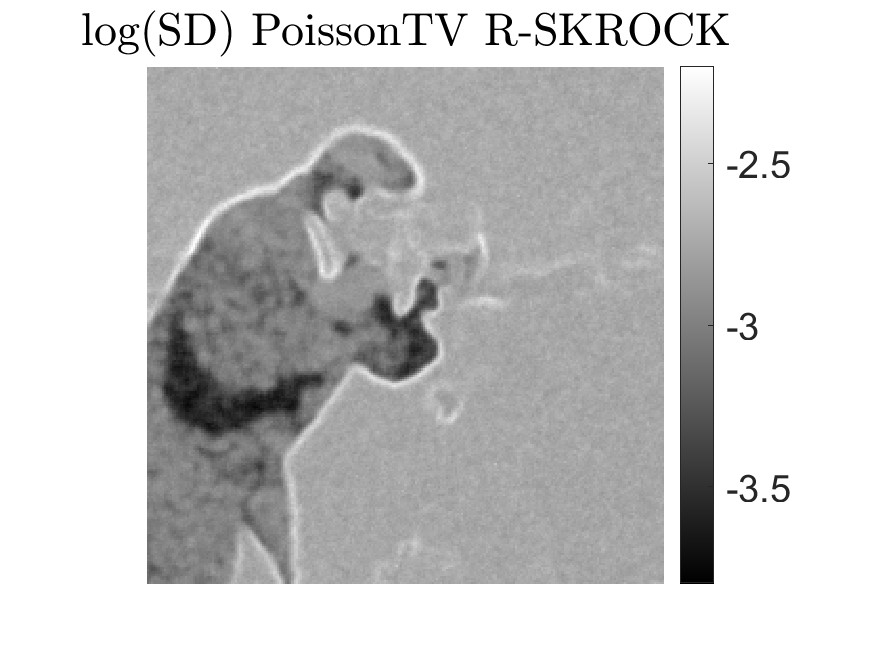}
\end{minipage}
\begin{minipage}{0.32\textwidth}
    \includegraphics[width=\linewidth]{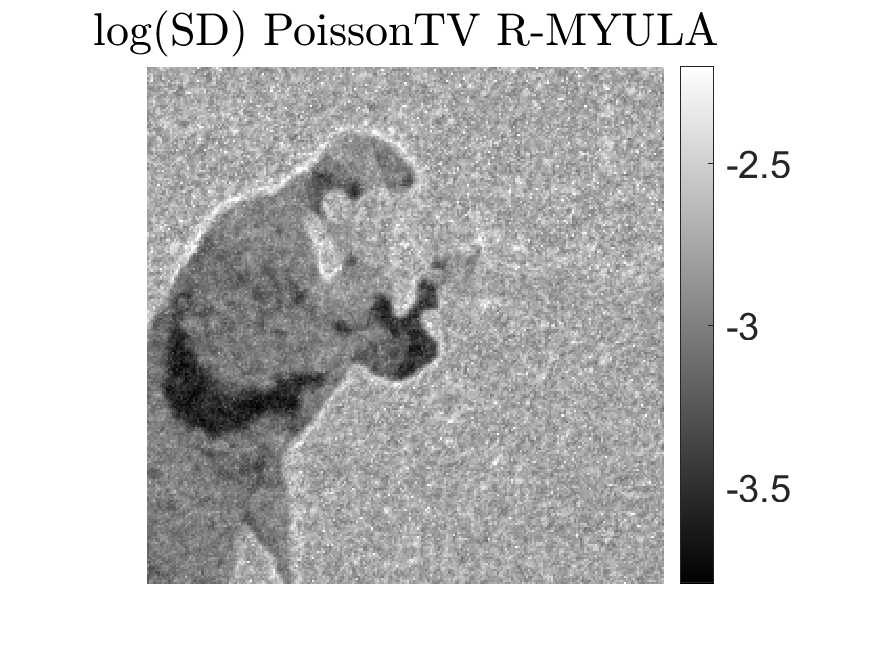}
\end{minipage}
\end{minipage}

\begin{minipage}{\textwidth}   
\begin{minipage}{0.32\textwidth}
    \includegraphics[width=\linewidth]{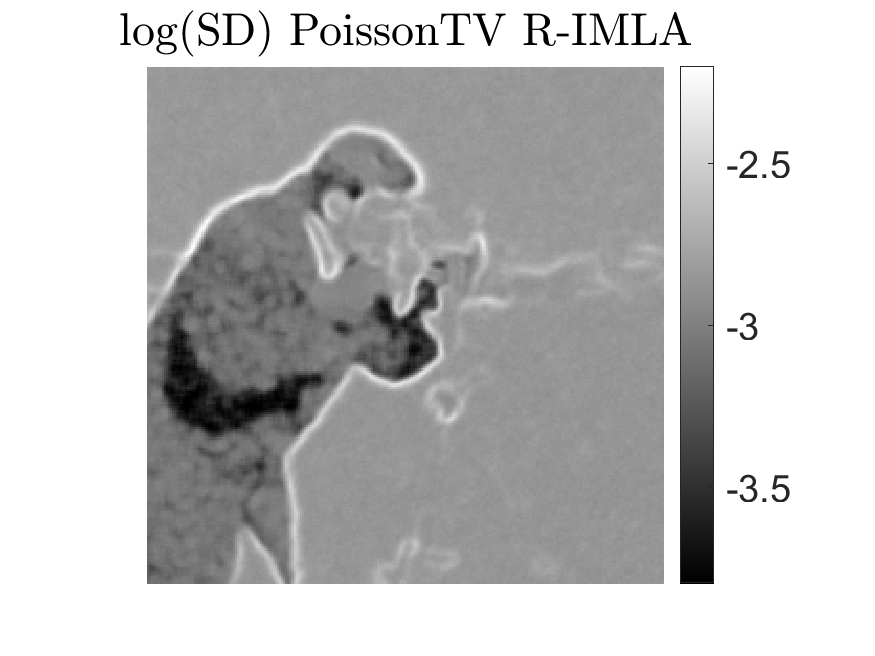}
\end{minipage}
\begin{minipage}{0.32\textwidth}
    \includegraphics[width=\linewidth]{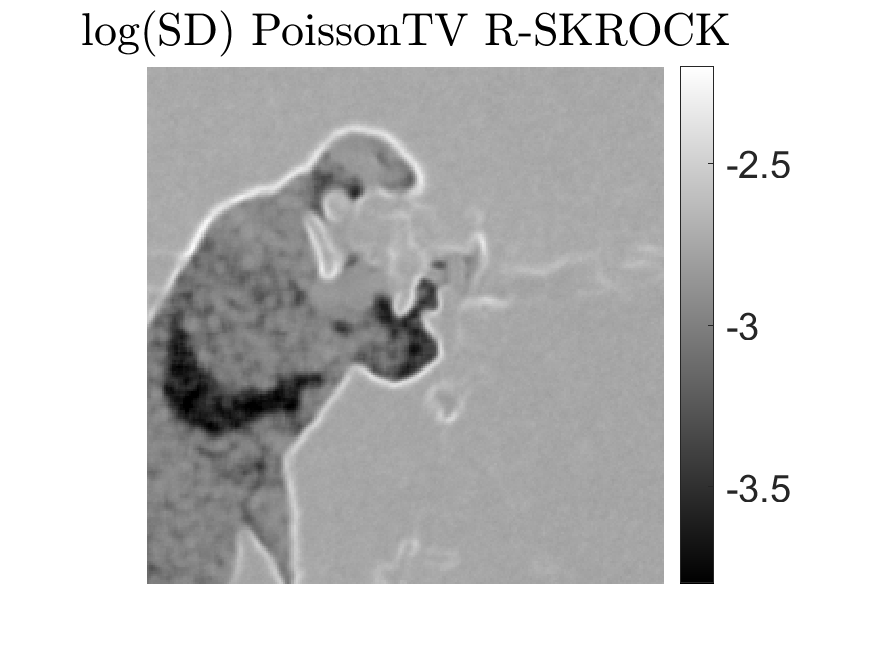}
\end{minipage}
\begin{minipage}{0.32\textwidth}
    \includegraphics[width=\linewidth]{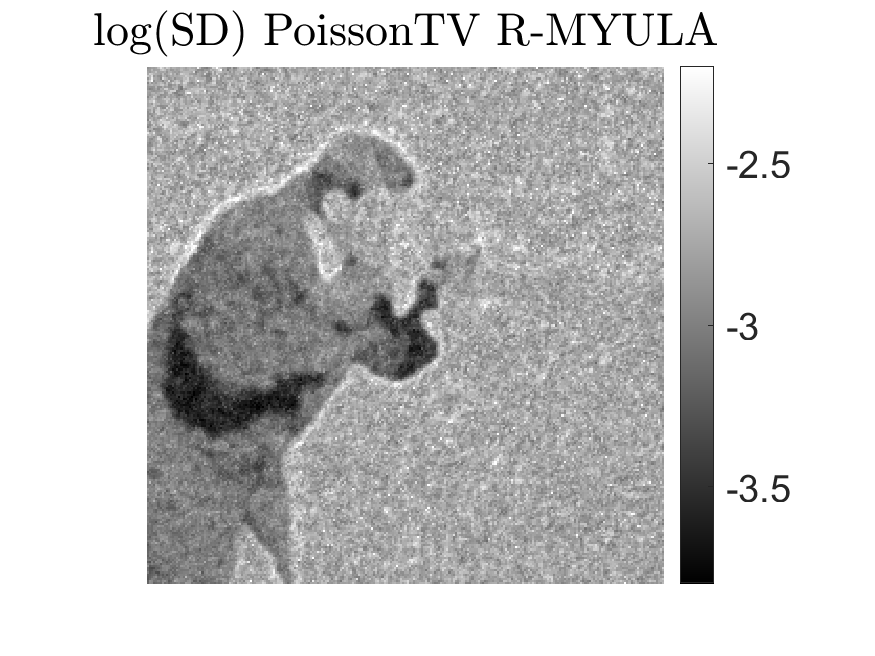}
\end{minipage}
\captionof{figure}{Marginal posterior standard deviation for the Poisson deconvolution experiment, in log-scale. R-IMLA was run using $h=6.65\times 10^{-5}$ (top), and $h=1.16\times 10^{-3}$ (bottom) which is equivalent to the effective step size of R-SKROCK when $s=10$ (top) or $s=40$ (bottom); R-MYULA was run using $h=1/L$.}
\label{fig:std-integration-time-cman}
\end{minipage}
\end{figure}

\subsection{Deconvolution with Poisson noise}\label{sec:poiss-deconv}
We now consider a total-variation (non-blind) image deconvolution problem with Poisson noise \cite{MDAMZ}, which allows us to demonstrate our methodology in a strongly non-Gaussian and non-smooth situation. The posterior distribution that we consider is given by 
\begin{equation} \label{eq:post2}
\pi(x) \propto \exp\left(-\sum_{i=1}^d[(Ax)_i + \beta -y_i\log((Ax)_i + \beta) +\iota_{(Ax)_i\geq 0}] -\theta_{\mathrm{TV}}\mathrm{TV}(x)\right),
\end{equation}
where $\beta>0$ is an arbitrary constant that can be interpreted as a known background level.

This model has two forms of non-smoothness that represent a challenge to the application of strategies derived from the OLSDE \eqref{eq:OLSDE}: (i) the model contains a (hard) positivity constraint; (ii) the total-variation pseudo-norm $\mathrm{TV}(\cdot)$ is not Lipschitz smooth. Also note that \eqref{eq:post2} is weakly log-concave, hence this allows us to illustrate our approach in a situation where the analysis of Sections \ref{subsec:Gaussian} and \ref{subsec:Convex} does not apply directly. In particular, because of the Poisson likelihood and the positivity constraint, $\pi$ is very different from a Gaussian model, so one could expect IMLA to exhibit a noticeable bias when the step size is sufficiently large.

\begin{algorithm}[t]
\caption{R-IMLA}\label{alg:RIMLA}
\begin{algorithmic}
\Require $N \geq 0$, $\delta>0$ and $X_0\in \R^d$.
\For{n=0 : N-1}\\
\textbf{Draw} $$\xi_n\sim \mathcal{N}(0,I_d)$$
\textbf{Set}
$$X_{n+1}^i=\lvert \tilde{X}_{n+1}^i\rvert, \qquad \text{ for all } i=1,\dots, d,$$
\begin{gather}\label{eq:subproblem}
    \tilde{X}_{n+1}\gets \arg\min_{x\in \R^d}2\pot^{\lambda}\left(\frac{1}{2} x+\frac{1}{2}X_{n}\right) +\frac{1}{2\delta}\lVert x-X_{n}-\sqrt{2\delta}\xi_n\rVert^2
\end{gather}
\EndFor
\end{algorithmic}
\end{algorithm}

Following \cite{durmus2017}, to address the lack of smoothness of the total-variation prior we replace $\mathrm{TV}(\cdot)$ by its Moreau-Yosida envelope 
$$
\mathrm{TV}^\lambda(x) = \min_{u\in\R^d} \left\{TV(u)+\frac{1}{2\lambda}\lVert x-u\rVert^2\right\}\, ,
$$
and construct the regularised posterior density
$$
\pi^{\lambda}(x) \propto\exp\left(-\sum_{i=1}^d[(Ax)_i + \beta -y_i\log((Ax)_i + \beta) +\iota_{(Ax)_i\geq 0}] -\theta_{\mathrm{TV}^\lambda}\mathrm{TV}^\lambda(x)\right)\, .
$$
Applying gradient-based sampling strategies to $\pi^{\lambda}$
 remains challenging due to the positivity constraint. Following \cite{MDAMZ}, we address this issue by using a reflected version of the overdamped Langevin diffusion, which naturally incorporates the positivity constraint. This diffusion underpins the R-MYULA and R-SKROCK methods recently introduced in \cite{MDAMZ}. Analogously, we construct a reflected version of IMLA (R-IMLA) to approximate a reflected Langevin diffusion targeting $\pi^\lambda$. The resulting scheme is summarised in Algorithm \ref{alg:RIMLA}\footnote{One can also apply IMLA without replacing $g$ by $g^\lambda$, but this would make the comparisons with MYULA and SKROCK less clear as the algorithms would be targeting different models.}.

\begin{figure}[p]
    \centering
\begin{minipage}{\textwidth}   
\begin{minipage}{0.49\textwidth}
    \includegraphics[width=\linewidth]{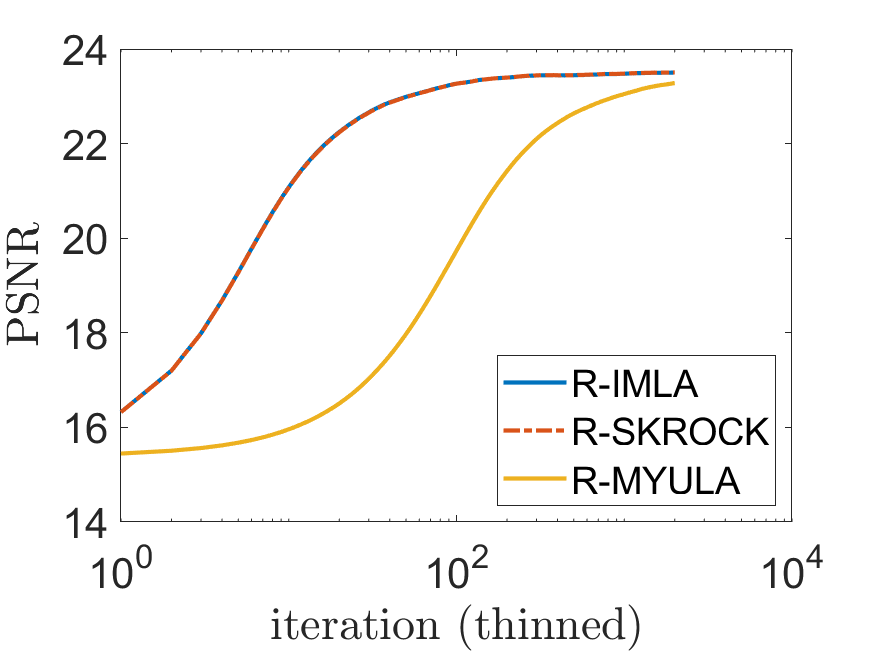}
\end{minipage}
    \begin{minipage}{0.49\textwidth}
        \includegraphics[width=\linewidth]{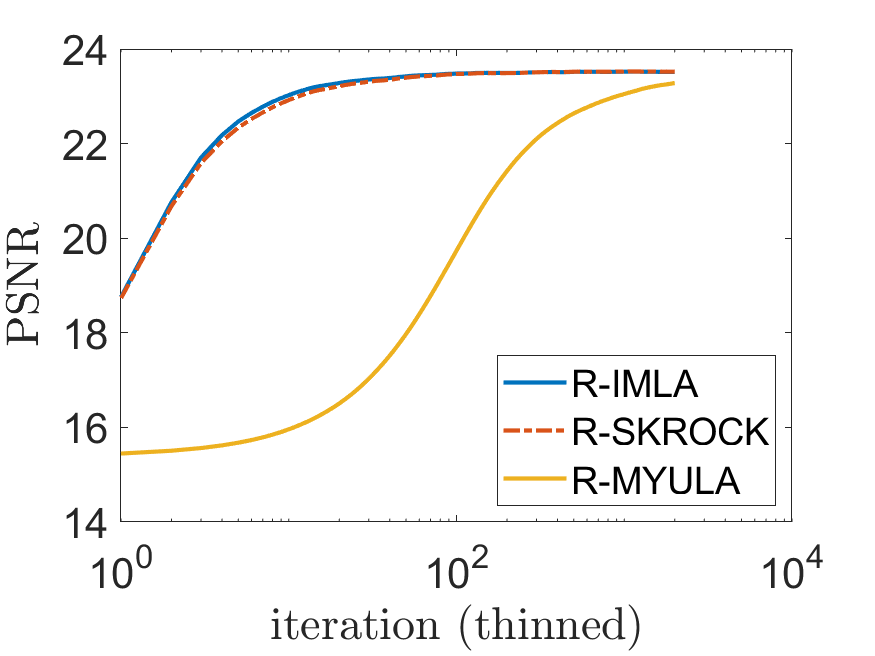}
    \end{minipage}
    \end{minipage}

\begin{minipage}{\textwidth}   
\begin{minipage}{0.49\textwidth}
    \includegraphics[width=\linewidth]{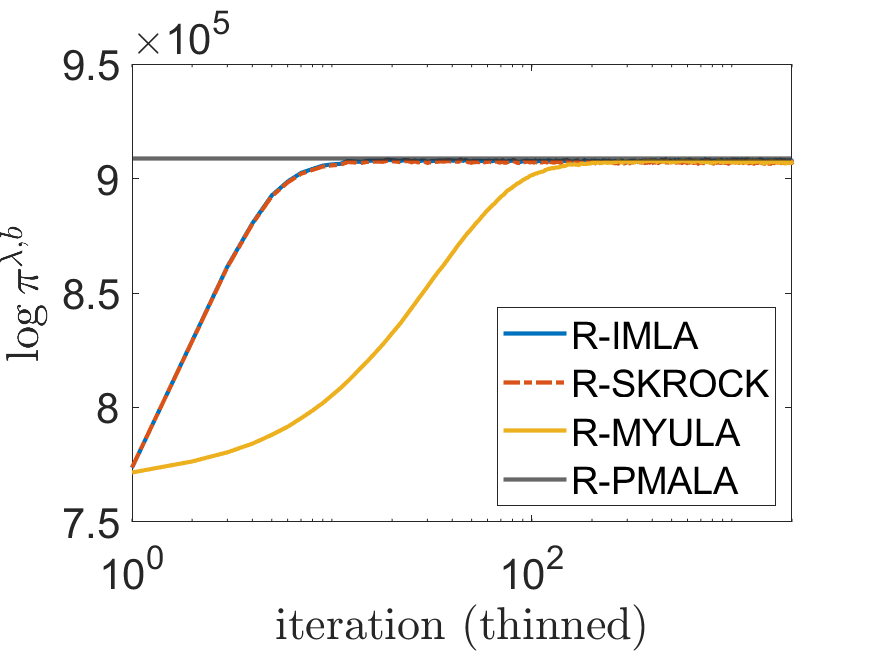}
 \end{minipage}
     \begin{minipage}{0.49\textwidth}
        \includegraphics[width=\linewidth]{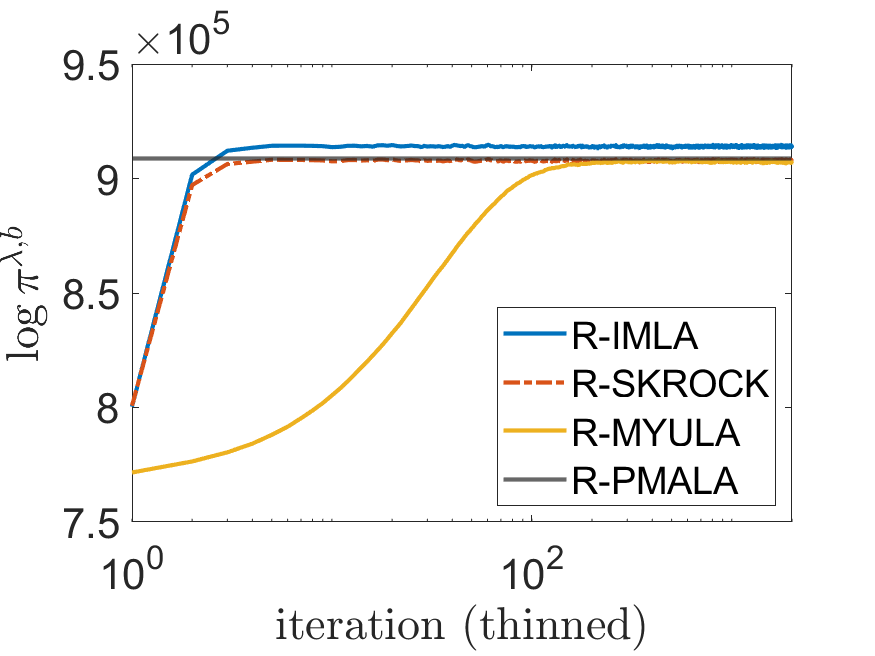}
    \end{minipage}

  \end{minipage}

\begin{minipage}{\textwidth}   
\begin{minipage}{0.49\textwidth}
    \includegraphics[width=\linewidth]{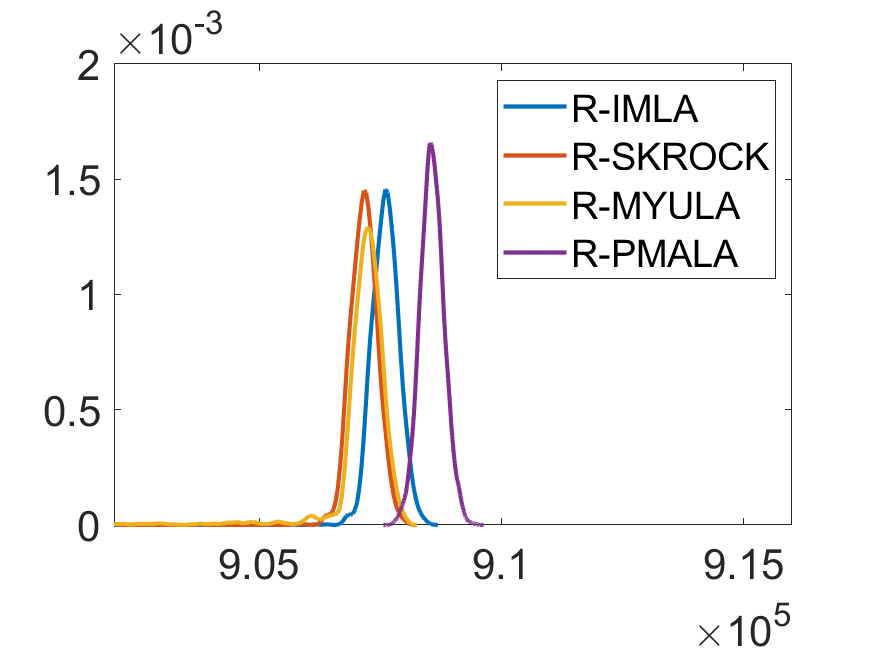}    
\end{minipage}
    \begin{minipage}{0.49\textwidth}
    \includegraphics[width=\linewidth]{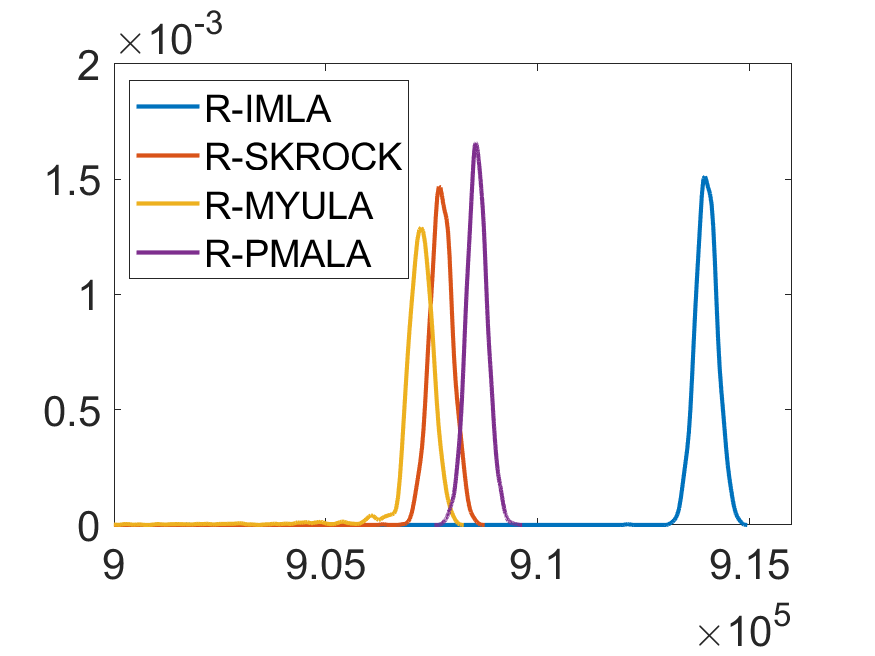}    
    \end{minipage}
    \captionof{figure}{Poisson deconvolution experiment. PSNR of the running mean (top), $\log \pi^{\lambda,b}(x)$ (middle), and non-parametric density estimates of $\log \pi^{\lambda,b}(x)$ at stationarity (bottom). R-IMLA was run using $h=6.65\times 10^{-5}$ (left), and $h=1.16\times 10^{-3}$ (right) which is equivalent to the effective step size of R-SKROCK when $s=10$ (left) or $s=40$ (right), R-MYULA was run using $h=1/L$, R-PMALA is shown as an unbiased reference.}
    \label{fig:psnr-int}
    \end{minipage}
\end{figure}

\begin{figure}[t]
    \centering
\begin{minipage}{\textwidth}   
    \begin{minipage}{0.49\textwidth}
        \includegraphics[width=\linewidth]{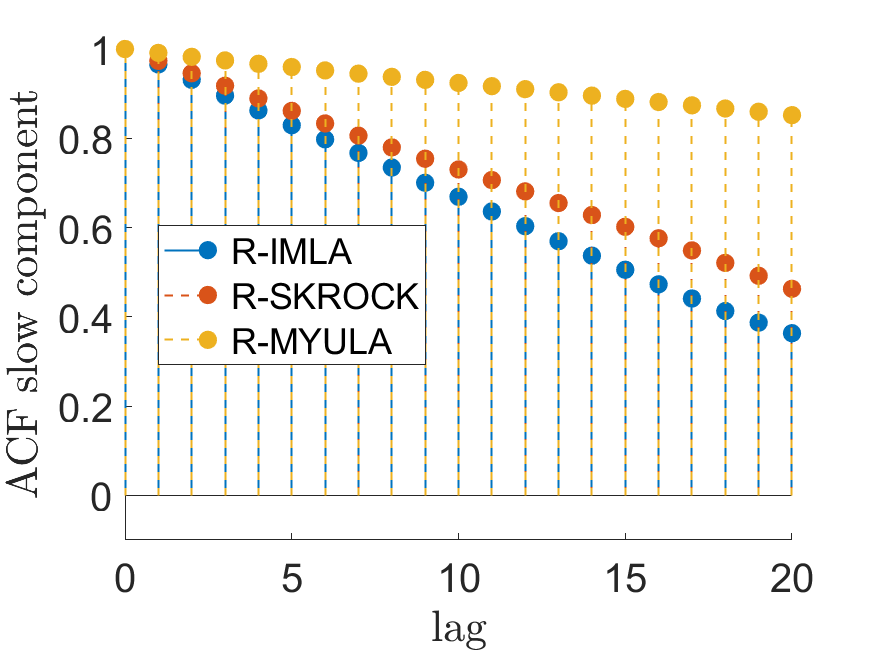}
    \end{minipage}
    \begin{minipage}{0.49\textwidth}
        \includegraphics[width=\linewidth]{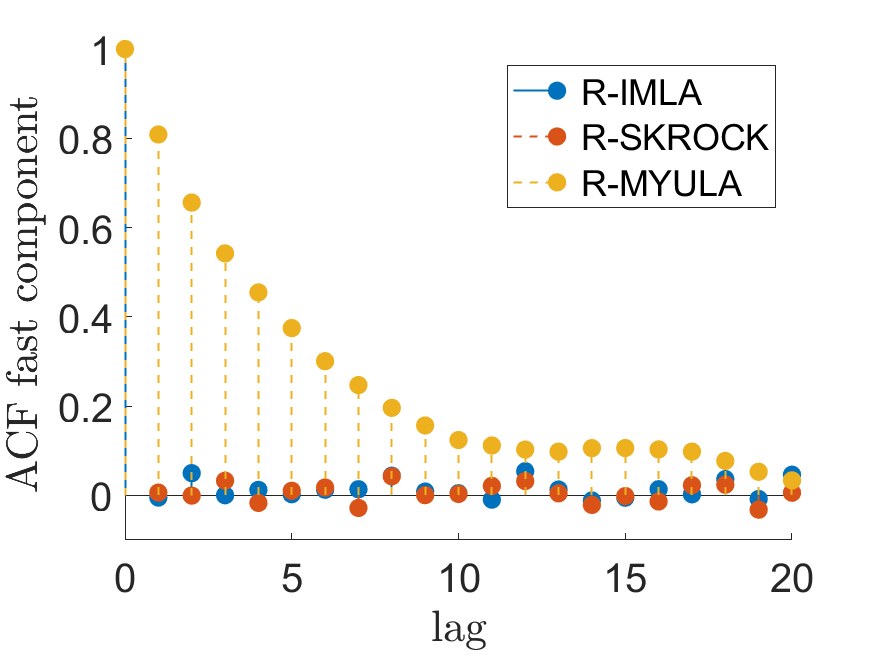}
    \end{minipage}
    \end{minipage}

\begin{minipage}{\textwidth}   
    \begin{minipage}{0.49\textwidth}
        \includegraphics[width=\linewidth]{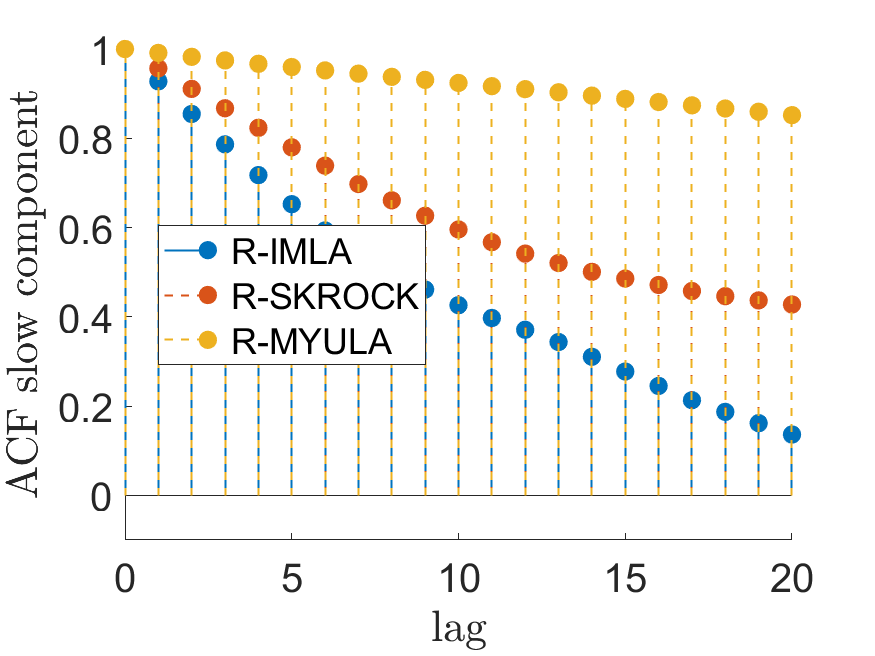}
    \end{minipage}
    \begin{minipage}{0.49\textwidth}
        \includegraphics[width=\linewidth]{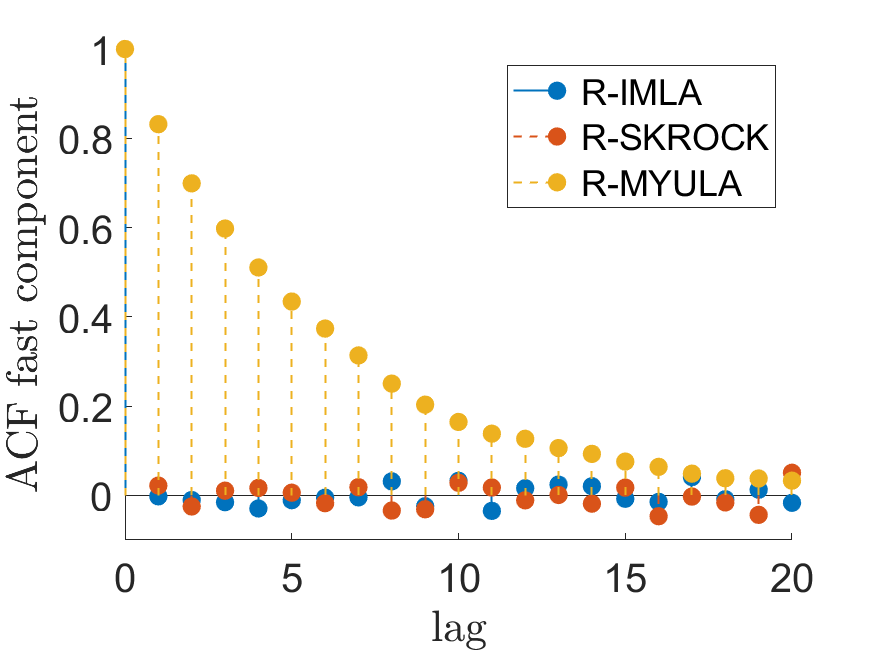}
    \end{minipage}
    \captionof{figure}{Comparison of ACF for the slowest (left) and the fastest component (right) for the Poisson deconvolution experiment \ref{sec:poiss-deconv-cost}. R-IMLA was run using $h=2.82\times 10^{-4}$ (top), and $h=1.16\times 10^{-3}$ (bottom) which is equivalent to the effective time step of R-SKROCK when $s=20$ (top) or $s=40$ (bottom), R-MYULA was run using $h=1/L$.}
    \label{fig:poisson-acf-improved}
    \end{minipage}
\end{figure}

 For this experiment, $A$ is a $5\times 5$ box blur operator, while  we use the \texttt{cameraman} image  scaled  to have  a mean intensity value (MIV) of 10 (ground truth and the observed data can be observed in Figure~\ref{fig:poisson-data-cman}). Similarly to \cite{MDAMZ}, the background noise level $\beta$ is set to $1\%$ of the MIV. Using the guidelines from \cite[Section~4.4]{MDAMZ} we set $\lambda=L_{f_y}^{-1}$ where $L_{f_y}$ is the Lipschitz constant of $x \mapsto \nabla \log p(y|x)$, and set the $\theta_{\mathrm{TV^\lambda}}$ automatically by marginal likelihood estimation by using the SAPG algorithm \cite{vidal2019maximum}.

 We apply the R-IMLA, R-SKROCK, and R-MYULA algorithms to generate Monte Carlo samples from $\pi^\lambda$. As in the previous experiment, we set the time-step for R-MYULA to half the stability barrier $h_{\rm{R-MYULA}}=1/L$, with $L = L_{f_y} + 1/\lambda$. For R-SKROCK, \cite{MDAMZ} recommends setting $s \in \{5,\ldots,15\}$ and conducts experiments with $s = 10$, so we also set $s=10$. For completeness, we also report results with $s=20$ and $s=40$ in order to assess the effect of the step size on the convergence speed and the estimation bias - the corresponding step size values can be seen in Table \ref{tab:poisson-time-table}. Again to make the comparison fair between R-MYULA and R-SKROCK we use $8\times 10^{5}$ iterations for R-MYULA and $8\times 10^{5}/s$ iterations for R-SKROCK.

With regard to IMLA, unlike most examples seen before, the posterior $\pi^{\lambda}$ is not strongly log-concave and hence we do not have an optimal choice for the time-step. Thus, we consider comparisons with R-IMLA in two different ways. In the first comparison, we choose to ignore the computational cost of solving the optimisation problem \eqref{eq:subproblem} at each iteration and use the same time-step and the same number of iterations for IMLA and SKROCK. These comparisons are reported in Section \ref{sec:poiss-deconv-int}. In contrast, in the second comparison, we keep the time steps equal but take into account the cost per iteration of each algorithm by modifying the number of iterations for IMLA so that IMLA and SKROCK have the same clock time. These experiments are reported in Section \ref{sec:poiss-deconv-cost}.

\begin{table}[t]
\centering
\def\arraystretch{1.5}
\setlength{\tabcolsep}{0.5em}
\begin{tabular}{|c|c|cc|c|cc|}
\hline
 \multicolumn{1}{|c|}{\multirow{2}{*}{stages}} & \multicolumn{1}{c|}{\multirow{2}{*}{step size}} & \multicolumn{2}{c|}{cost/iter(s)}                      & factor                    & \multicolumn{2}{c|}{iterations}                                            \\ \cline{3-7} 
\multicolumn{1}{|c|}{}                           & \multicolumn{1}{c|}{}                        & \multicolumn{1}{c|}{R-SKROCK} & \multicolumn{1}{c|}{R-IMLA} & $\frac{t(\text{R-IMLA})}{t(\text{R-SKROCK})}$          & \multicolumn{1}{c|}{R-SKROCK} & \multicolumn{1}{c|}{R-IMLA} \\ \hline\hline
1                                      & $6.65\times 10^{-6} $                                              & \multicolumn{1}{c|}{}       & \multicolumn{1}{c|}{}    &                           & \multicolumn{1}{c|}{800000} & \multicolumn{1}{c|}{}      \\ \hline
10                                  & $6.65\times 10^{-5}  $                                            & \multicolumn{1}{c|}{0.22}   & 0.22                     & \multicolumn{1}{c|}{0.99} & \multicolumn{1}{c|}{80000}  & 80000                                           \\ \hline
20                                       & $2.82\times 10^{-4}  $                                          & \multicolumn{1}{c|}{0.42}   & 0.28                     & \multicolumn{1}{c|}{0.67} & \multicolumn{1}{c|}{40000}  & 62000                                         \\ \hline
40 & $1.16\times 10^{-3}      $                                        & \multicolumn{1}{c|}{0.84}   & 0.45                     & \multicolumn{1}{c|}{0.54} & \multicolumn{1}{c|}{20000}  & 36000                                     \\ \hline
\end{tabular}
\caption{Cost in seconds per iteration for different step sizes used to run R-IMLA and the equivalent number stages for R-SKROCK. Listing the factor showing the cost improvement of R-IMLA and the number of iterations used to run the experiments. R-SKROCK with 1 stage is equivalent to R-MYULA.}
\label{tab:poisson-time-table}
\end{table}

\subsubsection{Comparison according to integration time}\label{sec:poiss-deconv-int}

For each of the algorithms R-MYULA, R-IMLA and R-SKROCK the estimated posterior means are reported in Figure~\ref{fig:post-mean-integration-time-cman} and the pixel-wise standard deviations are reported in Figure~\ref{fig:std-integration-time-cman}. In both figures, the results for R-SKROCK and R-IMLA are almost identical, which is consistent with the PSNRs obtained with the posterior means. In contrast, the reconstruction with R-MYULA is significantly worse which illustrates cases where R-MYULA has not converged.

In Figure~\ref{fig:psnr-int} we compare the PSNR of the running mean, the trace of the statistic $\log \pi^{\lambda}(x)$, and non-parametric density estimates of the statistic $\log \pi^{\lambda}(x)$ at stationarity. In each case, R-IMLA and R-SKROCK show similar convergence speeds while being significantly faster than R-MYULA. To assess the true posterior distribution of the statistics $\log\pi^{\lambda}$, Figure~\ref{fig:psnr-int} also reports the results obtained by running the MH-corrected R-PMALA method, which produces asymptotically unbiased samples from $\pi^\lambda$ (we calculated the density plots displayed in Figure~\ref{fig:psnr-int} by using $9.8 \times 10^6$ iterations of R-PMALA). We observe that all the algorithms produce accurate estimates of the posterior distribution of the statistics $\log\pi^{\lambda}$, with errors of the order of $0.5\%$ or smaller. In this specific experiment IMLA achieves a slightly smaller bias than the other methods in the case of a smaller time step ($s=10$), and the other methods are more accurate than IMLA when the step size is large ($s = 40$), but the differences are very small and this ordering could be modified by considering a marginally different MIV value or a different choice of $\theta$.

We conclude that IMLA and SKROCK perform similarly in terms of integration time, and they significantly outperform MYULA, both in situations that are strongly log-concave as well also in settings involving weak log-concavity and constraints, as illustrated in this last set of experiments.

\subsubsection{Comparison according to computational cost}\label{sec:poiss-deconv-cost}
The experiments that we have conducted so far suggest that IMLA and SKROCK perform similarly when compared without taking into account the computational cost per iteration. The computational cost per iteration of SKROCK depends directly on the number of internal stages $s$ used, whereas the computational cost of IMLA depends on the choice of solver used. To demonstrate that IMLA can provide a competitive alternative to SKROCK despite requiring implicit steps, and to highlight the importance of choosing a suitable solver, we now perform a comparison where IMLA is implemented by using a limited memory quasi-Newton solver\footnote{we use the FORTRAN implementation \url{https://github.com/stephenbeckr/L-BFGS-B-C} with a MATLAB wrapper, and set the tolerance level to $10^{-4}$.} \cite{byrd1995}. As in the previous section, we use the same time step for R-IMLA and R-SKROCK, but we now run R-IMLA for the same clock time as for R-SKROCK (see Table~\ref{tab:poisson-time-table}). We observe that for $s = 10$, IMLA and SKROCK complete the same number of iterations in the same clock time, whereas IMLA is $50\%$ more computationally efficient than SKROCK when $s=20$, and almost twice as efficient in the large step regime ($s=40$). Figure \ref{fig:poisson-acf-improved} compares the convergence speed of IMLA and SKROCK with time-normalised ACF plots, for the time steps related to $s=20$ and $s=40$. Observe that IMLA exhibits a moderately higher convergence speed than SKROCK for the moderately large step ($s=20$) and a significantly better speed when a large step is used ($s=40$).

\section{Discussion and conclusions}\label{sec:conclusion}
This paper presented IMLA, a new proximal MCMC methodology for Bayesian inference in imaging inverse problems with an underlying convex geometry. Similarly to previous proximal MCMC methods such as MYULA \cite{durmus2017} and SKROCK \cite{PVZ20}, IMLA is derived from a discrete-time approximation of an overdamped Langevin diffusion that targets the posterior distribution of interest (or a Moreau-Yosida approximation when the original target is not smooth). IMLA takes the form of a stochastic relaxed proximal-point iteration that admits two complementary interpretations. For models that are smooth or regularised by Moreau-Yosida smoothing, the approximation constructed by IMLA is equivalent to an implicit midpoint approximation targeting the posterior distribution of interest. 
This discretisation is asymptotically unbiased for Gaussian targets and shown to converge in an accelerated manner for any target that is $\kappa$-strongly log-concave, (i.e., requiring in the order of $\sqrt{\kappa}$ iterations to converge, similarly to accelerated optimisation schemes, with $\kappa$ the condition number), comparing favourably to previous proximal MCMC methods which are biased and either not provably accelerated or only provably accelerated for Gaussian targets.
For models that are not smooth, the algorithm is equivalent to a Leimkuhler–Matthews discretisation of a Langevin diffusion targeting a Moreau-Yosida approximation of the posterior distribution of interest. For targets that are $\kappa$-strongly log-concave, the provided non-asymptotic convergence analysis also identifies the optimal time step which maximizes the convergence speed. We demonstrated the proposed methodology with a range of numerical experiments, including some tractable toy models designed to illustrate different classes of models, as well as two challenging experiments related to image deconvolution with Gaussian and Poisson noise.

From a methodological perspective, future works could focus on embedding IMLA within more Bayesian inference complex schemes, such as the stochastic approximation proximal gradient schemes for empirical Bayesian estimation \cite{vidal2019maximum} and proximal nested sampling \cite{Cai2022} for Bayesian model selection. Extending IMLA to mildly non-convex problems is also an important perspective for future work. 
Moreover, with regard to applications, we believe that IMLA has great potential for low-photon scientific imaging applications, where uncertainty quantification is critical and where the combination of challenging noise statistics (e.g., Binomial and Geometric noise) and the lack of ground truth data make log-concave priors a very convenient option \cite{MDAMZ}.

\appendix
\section{Appendix}
The appendix contains collected proofs of the results from Section \ref{subsec:Convex} and details of the model used Section ~\ref{sec:gmm-model}.

\subsection{Proof of Proposition~\ref{prop:summary}}\label{proof-prop-2}

Before giving the proof of Proposition \ref{prop:n_nonlinear} we give two auxiliary lemmas which summarise the properties of the constant $C$ as a function of $\delta$ and the bias for $\theta=1$. As $C$ has the same form for both the Gaussian and the strongly log-concave setting the following Lemma holds for both cases.

\begin{lemma}\label{lem:C}
    Let $C$ be given by \eqref{eq:contb} or \eqref{eq:Cnonlinear}, then we have
    \begin{equation*}
        C=\begin{cases}
            \frac{(1-\theta)L\delta-1}{\theta L\delta+1} & \delta>\delta_*\\
            \frac{1-(1-\theta)m\delta}{1+\theta m\delta} & \delta\leq \delta_*
        \end{cases}
    \end{equation*}
    where $\delta_*$ is the value of $\delta$ which minimises $C$ and is given by
    \begin{equation*}
    \delta_{*}=\frac{(2 \theta-1)( L+m)+\sqrt{(1-2 \theta )^2 (L+m)^2+16 (1-\theta) \theta  L m}}{4 (1-\theta) \theta  L m}
\end{equation*}
\end{lemma}
The constant $C$ here corresponds to constant $\rho$ in \cite[Theorem~2]{hodgkinson2021implicit} with $\delta=h/2$, we include the proof of the properties of $C$ here for completeness.

\begin{proof}[Proof of Lemma \ref{lem:C}]
    We split the range of $z$ in the maximum of $C$ based on when the term $R_1(-z)$ is positive, for $z\leq (1-\theta)^{-1}$, or negative, for $z\geq (1-\theta)^{-1}$:
    \begin{equation*}
        C=\left(\max_{z\in [m\delta,(1-\theta)^{-1}\vee L\delta]}R_{1}(-z)\right)\wedge\left(\max_{z\in [(1-\theta)^{-1}\wedge m\delta,L\delta]}-R_{1}(-z)^{2}\right)
    \end{equation*}
    Here we use the notation $\wedge$ to denote the maximum and $\vee$ to denote the minimum. Since $z\mapsto R_1(-z)$ is a decreasing function we can write $C$ as
    \begin{equation*}
        C=\begin{cases}
            R_1(-m\delta) & \text{ if } L\delta \leq (1-\theta)^{-1},\\
            R_1(-m\delta)\wedge (-R_1(-L\delta)) & \text{ if } m\delta \leq (1-\theta)^{-1}\leq L\delta,\\
            -R_1(-L\delta) & \text{ if } (1-\theta)^{-1} \leq m\delta.
        \end{cases}
    \end{equation*}
    As $R_1(-m\delta)$ is decreasing in $\delta$ and $-R_1(-L\delta)$ is increasing in $\delta$ to determine the maximum we only need to determine for what value of $\delta$ they are equal. Observe that $R_1(-m\delta)=-R_1(-L\delta)$ if $\delta =\delta_*$ with $\delta_*$ as in the statement of Lemma \ref{lem:C}. Therefore we have that $C=R_1(-m\delta)$ for $\delta \leq \delta_*$ and $C=-R_1(-L\delta)$ for $\delta>\delta_*$, using the definition of $R_1$ we obtain the claim of the lemma.
\end{proof}

\begin{proposition}\label{prop:theta_1_bias}
Set $\theta = 1$ then the bias may be bounded as
\begin{equation}\label{eq:implicitbound}
W_2(\pi,\tilde \pi) \leq
\min\left(\frac{1}{2} \sqrt{d\delta}, \frac{\sqrt{d}\delta}{4\sigma_{\min}}\right).
\end{equation}
\end{proposition}
\begin{proof}[Proof of Proposition \ref{prop:theta_1_bias}]
Taking the limit as $n\to\infty$ in \eqref{eqn:wassersteinDistanceFinal} we have
$$
W_2(\pi,\tilde \pi)^2 = \sum_{i=1}^d \phi(\sigma_i,\delta)^2
$$
where, from \eqref{eq:numericalinvariant},
$$
\phi(\sigma,\delta) = \sigma \left(1-\frac{1}{\sqrt{1+\frac{\delta}{2\sigma^2} }}\right).
$$
It is easy to prove the bounds
\begin{equation}\label{eq:phibounds}
\phi \leq \frac{\delta}{4\sigma}, \qquad\phi \leq \sigma.
\end{equation}
Then
\begin{eqnarray*}
W_2(\pi,\tilde \pi)^2 &=& \sum_{\delta \sigma_i^{-2}\leq 4} \phi(\sigma_i,\delta)^2
+\sum_{\delta \sigma_i^{-2}> 4} \phi(\sigma_i,\delta)^2\\
&\leq & \sum_{\delta \sigma_i^{-2}\leq 4} \frac{\delta^2}{16\sigma_i^{2}}
+\sum_{\delta \sigma_i^{-2}> 4} \sigma_i^{2}\\
&\leq& \sum_{\delta \sigma_i^{-2}\leq 4} \frac{\delta}{4}
+\sum_{\delta \sigma_i^{-2}> 4} \frac{\delta}{4}
\end{eqnarray*}
so that $W_2(\pi,\tilde \pi) \leq \sqrt{d\delta}/2$. Furthermore, from the first bound in \eqref{eq:phibounds} we
readily obtain
$
W_2(\pi,\tilde \pi) \leq \frac{\sqrt{d}\delta}{4\sigma_{\min}}. 
$
\end{proof}

\begin{proof}[Proof of Proposition~\ref{prop:summary}]
    We first consider the case $\theta=1/2$. Formula \eqref{eq:numericalinvariant} shows that in this case there is no bias: $\tilde \pi = \pi$. From
\eqref{eq:bbound}, in order to ensure that $W_2(\pi, Q_n)\leq \epsilon$ we have to take $n$ such that
$$
    n |\log(C)| \geq  \log\big(W_2(\pi,Q_0)\big)-\log(\epsilon).
$$
The best performance will correspond to the value  of $\delta$ that minimizes $C$ in \eqref{eq:contb}.
By Lemma \ref{lem:C}, $C$ is minimized when $\delta=\delta_\ast = 2\sigma_{\max}\sigma_{ \min}$, and 
$$
C_{\theta=1/2}^{\rm opt} = \frac{1-\sigma_{\max}^{-2}\delta_\ast/2}{1+\sigma_{\max}^{-2}\delta_\ast/2}.
$$
Since we are interested in case where $\sigma_d$ is small, leading to small $\delta_*$,
$C_{\theta=1/2}^{\rm opt}  \approx 1-\sigma_{\max}^{-2}\delta_*=1-2/\sqrt{\kappa}$ and $|\log C| \approx 2/\sqrt{\kappa}$. This yields
\begin{equation}\label{eq:n1/2}
n\approx \frac{\sqrt{\kappa}}{2}\big[\log\big(W_2(\pi,Q_0)\big)-\log(\epsilon)\big]
\end{equation}
where we note that in the righthand side there is \emph{no dependence on $d$, and the dependence on $\epsilon$ is logarithmic}.

Now consider the case $\theta=1$. By \eqref{eq:bbound} in order to ensure that $W_2(\pi,Q_n)\leq \epsilon$ we first take $\delta$ small enough to get
$W_2(\pi,\tilde \pi)\leq \epsilon/2$ and then $n$ large enough to get $C^n W_2(\pi,Q_0)\leq \epsilon/2$, i.e.
$$
n|\log(C)| \geq \log(W_2(\pi,Q_0))-\log(2^{-1}\epsilon).
$$
 From \eqref{eq:implicitbound},  the bias has the upper bound $\sqrt{d\delta}/2$, which is independent of $\kappa$. This independence is achieved at the prize of the bound being proportional to $\sqrt{\delta}$ rather than to $\delta$, as one may naively expect for a first-order integrator. For $\delta$ small, more precisely for $\delta \leq 4\sigma_{\min}^2$, \eqref{eq:implicitbound} yields the upper bound $\sqrt{d}\delta /(4\sigma_{\min})$, with a first-order, i.e. $\mathcal{O}(\delta)$, behaviour. However the regime $\delta \leq 4\sigma_{\min}^2$ is without any practical interest, because for values of $\delta$ of the order of $\sigma_{\min}^2$, one may use the explicit, $\theta=0$, integrator, with a much lower computational complexity per time-step.

According to \eqref{eq:implicitbound}, we choose the value of $\delta$
\begin{equation*}\label{eq:delta1}
\delta^\star = \max\left(\frac{\epsilon^2}{d},\frac{2\epsilon\sigma_{\min}}{\sqrt{d}}\right)
\end{equation*}
Unlike the optimal $\delta$ for $\theta = 1/2$ we see that $\delta^\star$ remains bounded away from zero as $\sigma_{\min}\downarrow 0$. On the other hand $\delta^\star$ depends on $\epsilon$ and  $d$. Since $R_1 = 1/(1+\delta \lambda)$ is a positive, monotonically decreasing function of $\lambda>0$, the contractivity constant is then
$$C_{\theta=1} = R_1(z_1) = 1/(1+\sigma_{\max}^{-2}\delta^*).$$
and, when $\delta^*$ is small, $C_{\theta=1} \approx 1-\sigma_{\max}^{-2}\delta^*$ and $|\log C| \approx\sigma_{\max}^{-2} \delta^*$
which leads to
\begin{equation}
\label{eq:n1}
n\approx \min\left( \frac{d\sigma_{\max}^{2}}{\epsilon^2}, \frac{\sqrt{d\kappa}\sigma_{\max}}{2\epsilon}\right) \big[\log(W_2(\pi,Q_0))-\log(2^{-1}\epsilon)\big].
\end{equation}
Comparison with \eqref{eq:n1/2} shows that now there is a dependence of $n$ on $d$ and, in addition, the dependence  on $\epsilon$ is now much worse than logarithmic. On the other hand the right hand side of \eqref{eq:n1} remains bounded as $\kappa\uparrow \infty$. The bounds show that $\theta =1/2$ will outperform $\theta = 1$ unless $d/(\sqrt{\kappa} \epsilon^2)$ is small.
\end{proof}

\subsection{Proof of Theorem \ref{thm:nonasymconv}}\label{proof-theorem}

\begin{proof}[Proof of Theorem \ref{thm:nonasymconv}]
    We will estimate $W_2(Q_n,\pi)$ by setting up a coupling between the numerical scheme $X_k$ and a Langevin SDE $L_t$ which is in stationarity with $\pi$ as its invariant distribution. Let $L_0\sim \pi$ and set
\begin{equation*}
L_t=L_0-\int_0^t \nabla \pot(L_s) ds +\sqrt{2} W_t.
\end{equation*}
Let $(X_0,L_0)$ be such that $\mathbb{E}[\lvert X_0-L_0\rvert^2]=W_2(Q_0,\pi)$. Define $L_s^{k}=L_{k\delta+s}$ and
\begin{equation*}
    D_k=X_k-L_0^{k}.
\end{equation*}
We will set up a synchronous coupling by setting $\xi_k=\delta^{-1/2}(W_{k\delta}-W_{(k-1)\delta})$ then
\begin{equation*}
X_k=X_{k-1}-\delta\nabla \pot(\theta X_k+(1-\theta)X_{k-1}) +\sqrt{2}(W_{k\delta}-W_{(k-1)\delta}) +E_k.
\end{equation*}
Here $E_k$ is given by
\begin{align*}
    E_k&=X_k-X_{k-1}+\nabla f(\theta X_k+(1-\theta)X_{k-1})\delta -\sqrt{2}(W_{k\delta}-W_{(k-1)\delta})\\
    &=\delta\nabla F(X_k;X_{k-1}, \xi_k).
\end{align*}
By assumption we have $\lVert E_k\rVert \leq \delta \varepsilon$.
Subtracting the expression for $L_t$ from $X_k$ we have 
\begin{align}\label{eq:Dk}
D_{k+1} =D_k -\delta\nabla \pot\left(\theta D_{k+1}+(1-\theta)D_k+\theta L_{\delta}^{k}+(1-\theta)L_0^k\right)  +\int_0^\delta \nabla \pot(L_s^k) ds+E_k.
\end{align}

By the mean value theorem there exists a matrix $\overline{H}=\overline{H}(L_k,L_{k+1},D_k,D_{k+1})$ with $\sigma(\overline{H})\subseteq [m,L]$ such that
\begin{equation}\label{eq:MVT}
\nabla \pot(\theta D_{k+1}+(1-\theta)D_k+\theta L_{\delta}^{k}+(1-\theta)L_0^k) =\nabla \pot(\theta L_{\delta}^{k}+(1-\theta)L_0^k) +\theta \overline{H}D_{k+1}+(1-\theta)\overline{H}D_k.
\end{equation}
Substituting \eqref{eq:MVT} into \eqref{eq:Dk} we have
\begin{align*}
(1+\theta\delta \overline{H})D_{k+1} =(1-(1-\theta)\delta \overline{H})D_k+\int_0^\delta [\nabla \pot(L_s^k) -\nabla \pot(\theta L_{\delta}^{k}+(1-\theta)L_0^k) ]ds+E_k.
\end{align*}
Observe that this now has a similar form to the Gaussian case since the $D_k$ terms correspond to the numerical contraction of the Gaussian scheme with  $\Sigma^{-1}=\overline{H}$, the integral term controls the bias and $E_k$ corresponds to the error incurred by not solving the implicit step exactly.

Taking $L^2$-norms (with respect to $\mathbb{E}$) we have
\begin{align*}
\lVert D_{k+1} \rVert_{L^2}&\leq \lVert R_1(-\delta \overline{H})D_k\rVert_{L^2} \\
&+ \int_0^\delta \lVert R_2(-\delta \overline{H})(\nabla \pot(L_s^k) -\nabla \pot(\theta L_\delta^k+(1-\theta)L_0^k)) \rVert_{L^2} ds+\lVert R_2(-\delta \overline{H})E_k\rVert_{L^2}.
\end{align*}
Recall $R_1,R_2$ were defined as functions on $\R$ in \eqref{eq:R1R2}, we lift these to functions on matrices. Note that $R_1(-\delta \overline{H})$ is bounded by C, where $C$ is given by \eqref{eq:Cnonlinear}, and $R_2(-\delta\overline{H})\leq R_2(-m\delta)$. Using these bounds and that $\lVert E_k\rVert\leq \delta \varepsilon$ we have
\begin{align}
\lVert D_{k+1} \rVert_{L^2}\leq & C\lVert D_k\rVert_{L^2} + \frac{1}{1+\theta\delta m}\int_0^\delta \lVert (\nabla \pot(L_s^k) -\nabla \pot(\theta L_\delta^k+(1-\theta)L_0^k)) \rVert_{L^2} ds+\frac{\varepsilon\delta}{1+\theta\delta m}. \label{eq:Dkexp}
\end{align}

Now consider
\begin{equation*}
    \int_0^t \lVert (\nabla \pot(L_s^k) -\nabla \pot(\theta L_\delta^k+(1-\theta)L_0^k)) \rVert_{L^2} ds
\end{equation*}
We will now follow the procedure of \cite{Dalalyan} to bound this term. Using that $\pot$ is $L$-Lipschitz we have
\begin{align}\label{eq:midpoint}
\lVert \nabla \pot(L_s^k) -\nabla \pot(\theta L_\delta^k+(1-\theta)L_0^k) \rVert_{L^2} 
\leq L \theta\lVert (L_s^k- L_\delta^k) \rVert_{L^2}+ L(1-\theta) \lVert (L_s^k-L_0^k) \rVert_{L^2}.
\end{align}
By the definition of $L_s$ we have
\begin{align*}
\lVert (L_s-L_0) \rVert_{L^2} &= \lVert -\int_0^s \nabla \pot(L_r) dr +\sqrt{2}W_s\rVert_{L^2}\\
&\leq \lVert -\int_0^s \nabla \pot(L_r) dr\rVert_{L^2} +\sqrt{2}\lVert W_s\rVert_{L^2}\\
&\leq \int_0^s\lVert  \nabla \pot(L_r) \rVert_{L^2}dr +\sqrt{2sd}.
\end{align*}
Since $L_r$ is stationary we have
\begin{align*}
\lVert (L_s-L_0) \rVert_{L^2} \leq s \lVert  \nabla \pot(L_0) \rVert_{L^2} +\sqrt{2sd}.
\end{align*}
By stationarity we also have
\begin{align*}
\lVert (L_\delta-L_s) \rVert_{L^2} \leq (\delta-s) \lVert  \nabla \pot(L_0) \rVert_{L^2} +\sqrt{2(\delta-s)d}.
\end{align*}
By \cite[Lemma 3]{Dalalyan} we have $\lVert  \nabla \pot(L_0) \rVert_{L^2} \leq \sqrt{Ld}$. Combining these estimates with \eqref{eq:midpoint}
\begin{align*}
\lVert \nabla \pot(L_s) &-\nabla \pot(\theta L_1+(1-\theta)L_0) \rVert_{L^2} \leq \\ L \theta((\delta-s) \sqrt{Ld} &+\sqrt{2(\delta-s)d})+ L(1-\theta)(s \sqrt{Ld} +\sqrt{2sd}).
\end{align*}
Integrating gives
\begin{align}\label{eq:bias}
\int_0^\delta\lVert \nabla \pot(L_s) -\nabla \pot(\theta L_1+(1-\theta)L_0) \rVert_{L^2} ds&\leq L \theta(\frac{1}{2}\delta^2 \sqrt{Ld} +\frac{2}{3}\delta^{\frac{3}{2}}\sqrt{2d})\\+ L(1-\theta)(\frac{1}{2}\delta^2 \sqrt{Ld} +\frac{2}{3}\delta^{\frac{3}{2}}\sqrt{2d})\nonumber
&= \frac{1}{2}\delta^2 L^{\frac{3}{2}}\sqrt{d} +\frac{2}{3}L\delta^{\frac{3}{2}}\sqrt{2d}.
\end{align}

Combining \eqref{eq:Dkexp} and \eqref{eq:bias} we have
\begin{align*}
\lVert D_{k+1} \rVert_{L^2}\leq & C\lVert D_k\rVert_{L^2} + \frac{\frac{1}{2}\delta^2 L^{\frac{3}{2}}\sqrt{d} +\frac{2}{3}L\delta^{\frac{3}{2}}\sqrt{2d}+\varepsilon\delta}{1+\theta\delta m}.
\end{align*}
Therefore we have
\begin{align*}
\lVert D_{k+1} \rVert_{L^2}\leq & C^{k+1}\lVert D_0\rVert_{L^2} + \frac{1-C^{k+2}}{1-C}\frac{\frac{1}{2}\delta^2 L^{\frac{3}{2}}\sqrt{d} +\frac{2}{3}L\delta^{\frac{3}{2}}\sqrt{2d}+\varepsilon\delta}{1+\theta\delta m}.
\end{align*}
Using that $W_2(Q_n,\pi) \leq D_n$ and $D_0=W_2( Q_0,\pi)$ we have
\begin{align*}
W_2(Q_n,\pi)\leq & C^{n}W_2( Q_0,\pi) + \frac{1-C^{n+1}}{1-C}\frac{\frac{1}{2}\delta^2 L^{\frac{3}{2}}\sqrt{d} +\frac{2}{3}L\delta^{\frac{3}{2}}\sqrt{2d}+\varepsilon\delta}{1+\theta\delta m}.
\end{align*}
\end{proof}

\subsection{Proof of Proposition \ref{prop:n_nonlinear}}\label{proof-prop-4}

\begin{proof}[Proof of Proposition \ref{prop:n_nonlinear}]
    Fix $\epsilon>0$, then we ensure that $W_2(Q_n,\pi)\leq \epsilon$ holds by first finding $\delta$ such that the bias $\lim_{n\to\infty}W_2(Q_n,\pi)\leq \epsilon/2$ and then finding $n$ such that the contraction term $C^nW_2(Q_0,\pi)\leq \epsilon/2$. If these hold we have by \eqref{eq:non_asm_bound} that $W_2(Q_n,\pi)\leq\epsilon$ using that $1-C^{n+1}\leq 1$.

    Setting $\delta=\delta_*$ gives the optimal contraction rate and hence we either have that $\delta=\delta_*$ gives a bias smaller than $\epsilon/2$ or we take $\delta$ to be the largest time step such that the bias is less than $\epsilon/2$. For this reason we calculate the smallest value of $\epsilon$ such that the bias is less than $\epsilon/2$ when $\delta=\delta_*$. By taking the limit $n\to \infty$ in \eqref{eq:non_asm_bound} we have the following estimate of the bias:
    \begin{equation}\label{eq:biasest}
        \lim_{n\to\infty}W_2(Q_n,\pi)\leq \frac{1}{1-C}\frac{\frac{1}{2}\delta^2 L^{\frac{3}{2}}\sqrt{d} +\frac{2}{3}L\delta^{\frac{3}{2}}\sqrt{2d}}{1+\frac{1}{2} \delta m}.
    \end{equation}
    Setting $\delta=\delta_*$ in  \eqref{eq:biasest} and using Lemma \ref{lem:C} to determine $C$ we have
    \begin{align*}
        \lim_{n\to\infty}W_2(Q_n,\pi)&\leq \frac{2+m\delta_*}{2m\delta_*}\frac{\frac{1}{2}\delta_*^2 L^{\frac{3}{2}}\sqrt{d} +\frac{2}{3}L\delta_*^{\frac{3}{2}}\sqrt{2d}}{1+\frac{1}{2} \delta_* m}\\
        &=\frac{\sqrt{d}\kappa}{\sqrt{m}}\left( 1 +\frac{4}{3}\kappa^{-\frac{1}{4}}\right).
    \end{align*}
    Hence if $\varepsilon\geq \frac{\sqrt{d}\kappa}{\sqrt{m}}\left( 1 +\frac{4}{3}\kappa^{-\frac{1}{4}}\right)$ then the optimal choice of time step $\delta=\delta_*$ is valid.

    On the other hand if $\varepsilon<\frac{\sqrt{d}\kappa}{\sqrt{m}}\left( 1 +\frac{4}{3}\kappa^{-\frac{1}{4}}\right)$ then we need to determine the largest $\delta$ for which the bias is less than $\epsilon/2$.

    Case 1: If $\delta>\frac{32}{9}L^{-1} $ then we have that 
    \begin{equation*}
        \lim_{n\to\infty}W_2(Q_n,\pi)\leq \frac{1}{1-C}\frac{\delta^2 L^{\frac{3}{2}}\sqrt{d}}{1+\frac{1}{2} \delta m}.
    \end{equation*}
    Since $\delta<\delta_*$ we have by Lemma \ref{lem:C}
    \begin{equation*}
        \lim_{n\to\infty}W_2(Q_n,\pi)\leq \sqrt{L}\kappa \delta \sqrt{d}.
    \end{equation*}
    Therefore we can ensure that the bias is less than $\epsilon/2$ by using 
    \begin{equation*}
         \delta = \frac{\epsilon}{2\sqrt{Ld}\kappa}.
    \end{equation*}
    Note this is permissable provided $\delta>\frac{32}{9}L^{-1} $ which requires
    \begin{equation*}
        \epsilon>\frac{64}{9}\sqrt{d}\kappa L^{-\frac{1}{2}}. 
    \end{equation*}

    Case 2: If $\delta<\frac{32}{9}L^{-1} $, then as above we have
    \begin{equation*}
        \lim_{n\to\infty}W_2(Q_n,\pi)\leq \frac{1}{m}\frac{4}{3}L\delta^{\frac{1}{2}}\sqrt{2d}.
    \end{equation*}
    This is less than $\epsilon/2$ if we set
    \begin{equation*}
        \delta= \frac{9\epsilon^2}{128d\kappa^2 }.
    \end{equation*}
    Which gives a permissable value of $\delta$ for
    \begin{equation*}
        \epsilon^2<d\kappa^2\frac{4096}{81}L^{-1} = \left(\frac{64}{9}\sqrt{d}\kappa L^{-\frac{1}{2}}\right)^2.
    \end{equation*}

It remains to find the value of $n$ such that $C^nW_2(\pi_0,\pi)\leq \varepsilon/2$, that is we require 
\begin{equation*}
    n\geq \frac{\log(\epsilon/2)-\log(W_2(\pi,Q_0))}{|\log C|}.
\end{equation*} 
Since we always use $\delta\leq \delta_*$ we have $C=1-4m\delta/(2+m\delta)$ and we use the approximation $\log(C) \approx 2m\delta$.

From the above analysis we have that the best time step to obtain accuracy $\epsilon$ is
\begin{equation*}
    \delta = \min\left\{\frac{2}{\sqrt{Lm}}, \frac{\epsilon}{2\kappa\sqrt{Ld}}, \frac{9}{128}\frac{\epsilon^2}{d\kappa}\right\}.
\end{equation*}
Therefore we have
\begin{equation*}
    n\approx \frac{1}{2m}\max\left\{\frac{\sqrt{Lm}}{2}, \frac{2\kappa\sqrt{Ld}}{\epsilon}, \frac{128}{9}\frac{d\kappa}{\epsilon^2}\right\} [\log(\epsilon/2)-\log(W_2(\pi,Q_0))].
\end{equation*}
\end{proof}

\subsection{Details of illustrative experiment~\ref{sec:gmm-model}}\label{gmm-appendix}

The posterior distribution $\pi(x)$ for the denoising problem using a Gaussian mixture prior model can be written for $(x)_{i=1}^d \in \R^d$ and data $(y)_{i=1}^d \in \R^d$ as follows

\begin{equation*}
    \pi(x)= \prod_{i=1}^d \left[\omega(y_i) f(x_i, y_i, \mu_0, \delta_0^2)+ (1-\omega(y_i)) f(x_i, y_i,\mu_1, \delta_1^2)\right]
\end{equation*}
\begin{equation*}
    f(x_i, y_i, \mu_k, \delta_k^2) = \exp\left( - \frac{(x_i-\mu_k(y_i))^2}{2 \delta_k^2}\right) 
\end{equation*}
\begin{equation*}
    \mu_k(y_i) = \left(\frac{y_i}{\sigma^2} + \frac{m_k}{\sigma_k^2} \right) \delta_k^2; \,\,\,\,\, \delta_k^2 = \frac{\sigma^2 \sigma_k^2}{\sigma^2 + \sigma_k^2}
\end{equation*}
\begin{equation*}
    \omega(y_i) = \frac{C_0(y_i) \tilde{\omega}}{C}; \,\,\,\,\,  C = \tilde{\omega}C_0(y_i) + (1-\tilde{\omega})C_1(y_i).
\end{equation*}
\begin{equation*}
    C_k(y_i) = \frac{1}{\sqrt{2\pi(\sigma_k^2 + \sigma^2)}} \exp \left( - \frac{(m_k - y_i)^2} {2(\sigma_k^2 + \sigma^2)} \right)
\end{equation*}

Refer to Table \ref{tab:model_param} for the parameter values $(m_{0},m_{1},\sigma^{2}_{0},\sigma^{2}_{1},\tilde{\omega})$ and noise variance $\sigma^{2}$ used in the experiment.

\subsection{\review{Proximal operator for Cauchy distribution}\label{app:cauchy_prox}}

\review{
In Section \ref{sec:1d-distributions} we calculate the proximal operator for the negative log density function of 4 different distributions. For the Laplace, uniform and light-tailed case this is straightforward since the negative log density function is convex. On the other hand, for the Cauchy distribution the negative log density $f(x)=\log(1+x^2)$ is not convex so it is not clear that the proximal operator for this function is well-defined. The proximal operator of $f$ is well-defined if and only if the function 
$$
F(y;x) = f(y) +\frac{1}{2\lambda} \lVert x-y\rVert^2
$$
has a unique minimiser. For the Cauchy setting since $f$ grows slower than the quadratic term we see that $F(y;x)\to\infty$ as $\lVert y\rVert$ tends to infinity and therefore as $F$ is smooth the minimiser of $F$ exists and occurs when $\nabla F(y;x)=0$. Rearranging the equation  $\nabla F(y;x)=0$ we see that any minimiser $y$ of $F$ must satisfy the cubic equation
\begin{equation}\label{eq:cauchy_prox_cubic}
    y^3-xy^2+(1+2\lambda)y-x=0.
\end{equation}
By the fundamental theorem of algebra there are at most $3$ distinct complex roots and at least one real root by the intermediate value theorem. In order to verify how many real roots there are we observe that when $\lambda =0$ the roots are at $y=x$ and $y=\pm i$, therefore for $\lambda$ sufficiently small there is a complex conjugate pair of roots close to $i$ and $-i$ and the unique real root must be close to $x$. Hence for $\lambda$ small the proximal operator of $f$ is well-defined and is given by the unique real solution to \eqref{eq:cauchy_prox_cubic}.  
}

\bibliographystyle{plain}  
\bibliography{references}

\end{document}